\pdfoutput=1
\RequirePackage{snapshot}

\documentclass[a4paper,11pt,twoside,leqno,english,fabfinal]{fabsmfart}

\usepackage{fabmacromath,fababbreviations,fabalphabets,fabmaps,fabQFT}
\usepackage{thmtools}

\declaretheorem[name=Notation, thmbox=M]{notation}

\usepackage{soul,xparse,pifont}
\setstcolor{red}
\usepackage{algorithm}
\usepackage{algpseudocode}
\algnewcommand\aAnd{\textbf{and}\xspace}
\algnewcommand\aReturn{\textbf{return}\xspace}
\algnewcommand\algorithmicto{\textbf{to}}
\algrenewtext{For}[3][]%
{\ifthenelse{\equal{#1}{}}{\algorithmicfor\ #2 \textbf{in} #3 \algorithmicdo}{\algorithmicfor\ #1 \textbf{from} #2 \algorithmicto\ #3
  \algorithmicdo}}
\algrenewcommand\algorithmiccomment[1]{\hfill\(\triangleright\)
  {\small #1}}%

\usepackage{multirow,array}
\newcommand{\thead}[1]{\text{\large\bfseries #1}}
\newcommand{\theadt}[1]{{\large\bfseries #1}}

\numberwithin{equation}{section}

\usepackage{hyperref}
\hypersetup{%
 pdftitle={},%
  pdfauthor={Vincent Rivasseau, Fabien Vignes-Tourneret},%
  pdfsubject={Constructive Tensor Field Theory: The T44 Model},%
  pdfkeywords={},%
 }
\usepackage[level=2]{bookmark}

\usepackage{graphbox,color,pbox}
\graphicspath{{figures/}{./}}
\usepackage{tikz}
\usetikzlibrary{positioning,matrix}
\usepackage{caption}
\usepackage{subcaption}
\usepackage[backend=bibtex, natbib, backref=false, style=alphabetic,
giveninits=true]{biblatex} 
\usepackage[capitalise,nameinlink]{cleveref}

\usepackage{enumitem}

\allowdisplaybreaks[1]

\newcommand{\scale}{.55}
\newcommand{\scaletwo}{.45}

\newcommand{\bee}{\begin{equation}}
\newcommand{\ee}{\end{equation}}
\newcommand{\beq}{\begin{eqnarray}}
\newcommand{\eeq}{\end{eqnarray}}
\newcommand{\bqa}{\begin{eqnarray}}
\newcommand{\eqa}{\end{eqnarray}}
\newcommand{\bea}{\begin{eqnarray}}
\newcommand{\eea}{\end{eqnarray}}
\newcommand{\beann}{\begin{eqnarray*}}
\newcommand{\eeann}{\end{eqnarray*}}  

\newcommand{\nuonefigs}[2][0]{\raisebox{-.4\height}{\includegraphics[scale=#2,angle=#1,origin=c]{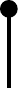}}}
\newcommand{\resinss}[2][0]{\raisebox{-.4\height}{\includegraphics[scale=#2,angle=#1,origin=c]{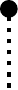}}}
\newcommand{\nutwofigs}[2][0]{\raisebox{-.4\height}{\includegraphics[scale=#2,angle=#1,origin=c]{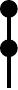}}}
\newcommand{\nuonefig}[1][0]{\nuonefigs[#1]{.6}}
\newcommand{\resins}[1][0]{\resinss[#1]{.6}}
\newcommand{\nutwofig}[1][0]{\nutwofigs[#1]{.6}}
\newcommand{\nuonefigtau}[1][0]{\raisebox{-.4\height}{\includegraphics[scale=.3,angle=#1,origin=c]{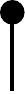}}}
\newcommand{\nutwofigtau}[1][0]{\raisebox{-.4\height}{\includegraphics[scale=.3,angle=#1,origin=c]{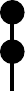}}}

\newcommand{\sbe}[2]{#1^{\raisebox{.3\height}{\scalebox{.6}{(#2)}}}}

\newcommand{\cst}[3]{{#1}^{\raisebox{.3\height}{\scalebox{.6}{\textup{$(#3)$}}}}_{#2}}
\newcommand{\cstK}[1]{K^{\raisebox{.3\height}{\scalebox{.6}{(#1)}}}}
\newcommand{\Itrois}[1]{\cst{I}{#1}{3}}

\newcommand{\ktilde}{\widetilde{k}}

\newcommand{\scalprodtens}[2]{\langle #1,#2\rangle}
\newcommand{\normtenssq}[1]{\scalprodtens{#1}{#1}}

\newcommand{\skel}[1]{\mathsf{#1}}
\newcommand{\skG}{\skel{G}}

\newcommand{\resgraph}[1]{{\mathcal #1}}
\newcommand{\resG}{\resgraph{G}}
\newcommand{\resC}{\resgraph{C}}

\newcommand{\secured}[1]{#1_{s}}
\newcommand{\secG}{\secured{\resG}}
\newcommand{\secC}{\secured{\resC}}

\newcommand{\rooted}[2][]{\ifthenelse{\equal{#1}{}}{#2_{\bullet}}{#2_{ #1, \bullet}}}
\newcommand{\rC}[1][]{\rooted[#1]{\resgraph{C}}}

\DeclareMathOperator{\Right}{\textsc{Right}}
\DeclareMathOperator{\Left}{\textsc{Left}}

\DeclareMathOperator{\ExpandR}{\textsc{ExpandR}}
\DeclareMathOperator{\ExpandL}{\textsc{ExpandL}}
\DeclareMathOperator{\ChooseExpand}{\textsc{ChooseExpand}}
\DeclareMathOperator{\Root}{\textsc{Root}}

\DeclareMathOperator{\RightTree}{\textsc{RightTree}}
\DeclareMathOperator{\LeftTree}{\textsc{LeftTree}}
\DeclareMathOperator{\pythlen}{len}

\DeclareMathOperator{\pythremove}{remove}
\DeclareMathOperator{\pythappend}{append}

\newcommand{\algoand}{\textbf{and}\xspace}

\newcommand{\Hilb}{\cH}

\newcommand{\Htens}{\Hilb^{\otimes}}
\newcommand{\Opon}[1]{L(#1)}
\newcommand{\OpHtens}{\Opon{\Htens}}
\newcommand{\Hop}[1][]{L(\Hilb_{#1})}
\newcommand{\Hopdirect}{\Hop^{\hspace{-1pt}\times}}

\newcommand{\EndHopd}{L(\Hopdirect)}
\newcommand{\direct}[1]{\xvec{#1}}
\newcommand{\directb}[1]{\vec{#1}}

\newcommand{\bV}{\widetilde{V}}
\newcommand{\bX}{\mathbf{X}}

\newcommand{\sigmad}{\direct{\sigma}}
\newcommand{\sigmads}{\directb{\sigma}}
\newcommand{\rsigmad}{\underline{\sigmad}}
\newcommand{\rsigmads}{\underline{\sigmads}}
\newcommand{\taud}{\direct{\tau}}
\newcommand{\tauds}{\directb{\tau}}
\newcommand{\rtaud}{\underline{\taud}}

\newcommand{\Trd}{\Tr}
\newcommand{\rTrd}{\Tr}
\newcommand{\transpose}{\scalebox{.6}{$T$}}
\newcommand{\Trsb}[1]{\Tr\bigl[#1\bigr]}
\newcommand{\Iscale}{\Idirect}
\newcommand{\Dvec}{\direct{\underline{\cD}}}

\usepackage{dsfont}
\DeclareMathOperator{\Itens}{\I}
\DeclareMathOperator{\Idirect}{\mathds{I}}
\DeclareMathOperator{\rIdirect}{\underline{\Idirect}}

\newcommand{\abs}[1]{\lvert #1\rvert}
\newcommand{\dU}[2][]{\Delta^{\!s_{#1}}_{#2}}
\newcommand{\lj}{\les j}
\newcommand{\fres}{\cR}

\newcommand{\Quu}{\cst{Q}{1}{1}}
\newcommand{\Qud}{\cst{Q}{1}{2}}
\newcommand{\Qt}{\widetilde{Q}}

\newcommand{\tuple}[1]{\mathbf{#1}}
\newcommand{\ntup}{\tuple{n}}
\newcommand{\tupn}{\tuple{n}}
\newcommand{\nbtup}{\tuple{\bar n}}
\newcommand{\mtup}{\tuple{m}}
\newcommand{\tupm}{\tuple{m}}
\newcommand{\mbtup}{\tuple{\bar m}}
\newcommand{\chitup}{\tuple{\chi}}
\newcommand{\chibtup}{\tuple{\bar\chi}}

\newcommand{\Torus}{\T}

\newcommand{\Ar}{A^{\text{r}}}
\newcommand{\Arop}{A^{\text{r}}}
\newcommand{\Aprep}{\underline{A}}
\newcommand{\Fprep}{\underline{\cF}}
\newcommand{\Vr}{V^{\text{r}}}
\newcommand{\Vf}{\cV^{\text{r}}} 

\newcommand{\lastvac}{\cE}

\newcommand{\Card}{\text{Card}}

\newcommand{\unj}{\indic_{j}}
\newcommand{\uninfj}{\indic_{\les j}}
\newcommand{\jm}{j_{\text{max}}}

\newcommand{\ccJ}{\!{\mathcal J}}
\newcommand{\nset}{[n]}

\newcommand{\Qj}{\cQ_{j}}

\newcommand{\Qzu}[1][]{\cst{Q}{#1}{01}}
\newcommand{\Gammad}{\mathbf{\Gamma}}

\newcommand{\be}{\mathbf{e}}

\renewcommand{\card}[1]{\abs{#1}}
\newcommand{\Qadm}{Q_{\textup{adm}}}

\newcommand{\Floc}{F_{\text{loc}}}
\newcommand{\Fnl}{F_{\text{nl}}}

\renewcommand{\ot}{\leftarrow}

\usepackage[xindy,section,sort=use,indexonlyfirst,savewrites,nomain,hyperfirst=false]{glossaries}
\usepackage{glossary-mcols}
\setglossarystyle{mcolindex}

\newglossary[const-glg]{constant}{const-gls}{const-glo}{Scalars}
\newglossary[tens-glg]{tensor}{tens-gls}{tens-glo}{Tensors}
\newglossary[spa-glg]{space}{spa-gls}{spa-glo}{Spaces}
\newglossary[opt-glg]{operatortens}{opt-gls}{opt-glo}{Operators on $\Htens$}
\newglossary[opd-glg]{operatordirect}{opd-gls}{opd-glo}{Operators on
  $\Hopdirect$}
\newglossary[misc-glg]{misc}{misc-gls}{misc-glo}{Miscellanea}
\newglossary[graphs-glg]{graphs}{graphs-gls}{graphs-glo}{Feynman Graphs}
\newignoredglossary{junk}
\makeglossaries

\NewDocumentCommand{\newconstant}{m m O{} O{}}{%
  \newglossaryentry{#1}{type=constant,name={$#2$},text={#2},description={#3}
    ,#4}
}
\NewDocumentCommand{\newtensor}{m m O{} O{}}{%
  \newglossaryentry{#1}{type=tensor,name={$#2$},text={#2},description={#3}
    ,#4}
}
\NewDocumentCommand{\newspace}{m m O{} O{}}{%
  \newglossaryentry{#1}{type=space,name={$#2$},text={#2},description={#3}
    ,#4}
}
\NewDocumentCommand{\newoperatort}{m m O{} O{}}{%
  \newglossaryentry{#1}{type=operatortens,name={$#2$},text={#2},description={#3}
    ,#4}
}
\NewDocumentCommand{\newoperatord}{m m O{} O{}}{%
  \newglossaryentry{#1}{type=operatordirect,name={$#2$},text={#2},description={#3}
    ,#4}
}
\NewDocumentCommand{\newmisc}{m m O{} O{}}{%
  \newglossaryentry{#1}{type=misc,name={$#2$},text={#2},description={#3}
    ,#4}
}
\NewDocumentCommand{\newgraph}{m m O{} O{}}{%
  \newglossaryentry{#1}{type=graphs,name={$#2$},text={#2},description={#3}
    ,#4}
}

\newspace{Hi}{\cH_{i}}[$\defi\ell_{2}(\Z)$]
\newspace{Htens}{\Htens}[$\defi\bigotimes_{i=1}^{4}\cH_{i}$]
\newspace{LHi}{\Hop[i]}[linear operators on $\cH_{i}$]
\newspace{Hoptens}{L(\Htens)}[linear operators on $\Htens$]
\newspace{Hopdirect}{\Hopdirect}[$\defi\bigtimes_{i=1}^{4} L(\cH_{i})$]
\newglossaryentry{EndHopd}{type=space,name={$\EndHopd$},text={\EndHopd},description={}}
\newspace{V}{\cV}[set of divergent vacuum graphs]
\newspace{cM}{\cM}[set of divergent $2$-point graphs]
\newspace{Cardrho}{\Card_{\rho}}[$\defi\{g\in\C\tqs
  |g|<\rho\cos(\tfrac 12\arg g)\}$]
\newspace{VB}{V_{\cB}}[$\defi \R^{\abs{\cB}} \otimes \Hopdirect$]

\newconstant{lambda}{\lambda}[square root of $g$]
\newconstant{g}{g}[coupling constant]
\newconstant{deltaG}{\delta_{G}}[counterterm]
\newconstant{cN}{\cN}[normalization constant]
\newconstant{cN1}{\cN_{1}}[normalization constant]
\newconstant{cN2}{\cN_{2}}[normalization constant]
\newconstant{cN3}{\cN_{3}}[normalization constant]
\newconstant{cN4}{\cN_{4}}[normalization constant]
\newconstant{cN5}{\cN_{5}}[normalization constant]
\newconstant{b1}{b_{1}}[matrix element of $B_{1}$]
\newconstant{Aprep}{\Aprep}[prepared amplitude]

\newoperatort{dm}{\delta_{m}}[mass counterterms]
\newoperatort{dt}{\delta_{t}}[mass counterterms squared]
\newoperatort{Itens}{\Itens}[identity operator]
\newoperatort{sigma}{\sigma}[$\defi\sum_{c}\sigma_{c}\otimes\Itens_{\hat
c}$]
\newoperatort{Sigma}{\Sigma}[$\defi i\lambda\sqrt C\sigma\sqrt C$]
\newoperatort{Sigmaprime}{\Sigma'_{j}}[$\defi\frac{d\Sigma}{dt_{j}}(t_{j})$]
\newoperatort{H}{H}[$\defi -i\lambda^{-1}\Sigma$]
\newoperatort{C}{C}[propagator]
\newoperatort{B1}{B_{1}}[$\defi i\lambda\Ar_{\cM_{1}}$]
\newoperatort{B2}{B_{2}}[$\defi -i\lambda^{3}\Ar_{\cM_{2}}$]
\newoperatort{D1}{D_{1}}[$\defi i\lambda C^{1/2}B_{1}C^{1/2}$]
\newoperatort{D2}{D_{2}}[$\defi i\lambda C^{1/2}B_{2}C^{1/2}$]
\newoperatort{D}{D}[$\defi D_{1}+D_{2}$]
\newoperatort{Dprimeun}{D'_{1,j}}[$\defi
\frac{dD_{1}}{dt_{j}}(t_{j})$]
\newoperatort{Dprimedeux}{D'_{2,j}}[$\defi \frac{dD_{2}}{dt_{j}}(t_{j})$]
\newoperatort{R}{R}[$\defi (\Itens-\Sigma)^{-1}$]
\newoperatort{R1}{R_{1}}[$\defi (\Itens-U_{1})^{-1}$]
\newoperatort{cR}{\fres}[$\defi (\Itens-U)^{-1}$]
\newoperatort{U1}{U_{1}}[$\defi \Sigma+D_{1}$]
\newoperatort{U}{U}[$\defi\Sigma+D$]
\newoperatort{Uprime}{U'_{j}}[$\defi \frac{dU}{dt_{j}}(t_{j})$]
\newoperatort{Ar}{\Ar}[renormalised amplitude]
\newoperatort{Vf}{\Vf}
\newoperatort{lastvac}{\lastvac}[$\defi\Aprep_{\cV_{5}}+\Aprep_{\cV_{6}}+\Aprep_{\cV_{7}}$]
\newoperatort{dU}{\dU{}{}}[$\defi\frac{\partial U}{\partial\sigma^{s}}$]
\newoperatort{cutofflesj}{\indic_{\les j}}[$\defi 1_{1+\norm{n}^{2}\les M^{2j}}\delta_{\mtup\ntup}$]
\newoperatort{cutoffj}{\indic_{j}}[$\defi \indic_{\les j} - \indic_{\les j-1}$]

\newoperatord{Gamma}{\Gammad}
\newoperatord{Q}{Q}
\newoperatord{Q0}{Q_{0}}
\newoperatord{Q1}{Q_{1}}

\newtensor{T}{T, \bar T}[original fields]
\newtensor{sigd}{\sigmad}[$\defi (\sigma_{c})_{c}$, main intermediate field]
\newtensor{taud}{\taud}[intermediate field for $\kN_{2}$ counterterm]
\newtensor{dB1}{\direct B_{1}}[$\defi i\lambda\direct{\Ar}_{\cM_{1}}$]
\newtensor{dB2}{\direct B_{2}}[$\defi -i\lambda^{3}\direct{\Ar}_{\cM_{2}}$]
\newtensor{ArdMun}{\direct{\Ar}_{\cM_{1}}}[$\defi(\Ar_{\cM_{1}^{c}})_{c}$]
\newtensor{ArdMdeux}{\direct{\Ar}_{\cM_{2}}}[$\defi(\Ar_{\cM_{2}^{c}})_{c}$]
\newglossaryentry{Ardb}{type=junk,name={$\directb{\Ar}$},text={\directb{\Ar}}}

\newgraph{tensorG}{G}[tensor graphs]
\newgraph{skG}{\skG}[skeleton graph]
\newgraph{J}{\cJ}[jungle]
\newgraph{resG}{\resG}[resolvent graph]
\newgraph{G}{G}[tensor graph]
\newgraph{cB}{\cB}[connected component of a Bosonic forest]
\newgraph{mapm}{\mapm}[map]
\newgraph{resC}{\resC}[chord diagram]
\newgraph{secG}{\secG}[secured graph]

\newmisc{Idirect}{\Idirect}[identity element of $L(\Hopdirect)$]
\newmisc{Iscale}{\Idirect_{S}}[$\abs S\times \abs S$ identity matrix]
\newmisc{S}{S}[$\defi \lnat 1,\jm\rnat$]
\newmisc{nset}{\nset}[$\defi\set{1,2,\dotsc,n}$]
\newmisc{One}{\bbbone}[matrix full of $1$'s]
\newmisc{Qj}{\Qj}[quadratic form on $\Hopdirect$] 
\newmisc{Z}{\cZ}[partition function]
\newmisc{cW}{\cW}[$\defi\log\cZ$]
\newmisc{Vc}{V_{c}}[interaction polynomial]
\newmisc{Vcr}{\Vr_{c}}[renormalised interaction]
\newmisc{W}{W}[amplitude of a node of a jungle]
\newglossaryentry{X}{type=misc,name={$X(\tuple{w})$},text={X},description={interpolated
  Bosonic covariance}}
\newglossaryentry{XB}{type=misc,name={$X_{\cB}$},text={X_{\cB}},description={diagonal
  block of $X$}}

\newcommand{\scadvdual}{.5}

\title{Constructive tensor field theory: The $T^{4}_{4}$ model}
\shorttitle{Constructive TFT: The $T^{4}_{4}$ model}
\author{V.~Rivasseau \& F.~Vignes-Tourneret}
\dedicatory{}

\begin{document}
\maketitle

\begin{fabsmfabstract}
  We continue our constructive study of tensor field theory through
  the next natural model, namely the rank four tensor theory with
  quartic melonic interactions and propagator inverse of the Laplacian
  on $U(1)^4$.  This superrenormalizable tensor field theory has a
  power counting quite similar to ordinary $\phi^4_3$. We control the
  model via a multiscale loop vertex expansion which has to be pushed
  quite beyond the one of the $T^4_3$ model and we establish its Borel
  summability in the coupling constant. This paper is also a step to
  prepare the constructive treatment of just renormalizable models,
  such as the $T^4_5$ model with quartic melonic interactions.
\end{fabsmfabstract}
\begin{fabsmfMSC}
  Primary 81T08; Secondary 83C45.
\end{fabsmfMSC}

\vfil
\tableofcontents
\newpage

\section*{Context and outline}
\label{intro}
\etoctoccontentsline*{section}{Context and outline}{1}

Perturbative quantum field theory develops the functional integrals of Lagrangian quantum field theories 
such as those of the standard model  into a formal series of Feynman graphs and their amplitudes.
The latter are the basic objects to compute in order to compare weakly coupled theories with actual
particle physics experiments. However isolated Feynman amplitudes or even
the full formal perturbative series cannot be considered as a complete physical theory.  Indeed the 
non-perturbative content of Feynman functional integrals is essential to their
physical interpretation, in particular when investigating stability of the vacuum and  
the phase structure of the model. 

Axiomatic field theory, in contrast,
typically does not introduce Lagrangians nor Feynman graphs but studies
rigorously the general properties that any local quantum field theory ought to possess \cite{Streater1964aa,Haag1996aa}.
Locality is indeed at the core of  the
mathematically rigorous formulation of quantum field theory. It is a key Wightman axiom \cite{Streater1964aa} and in  algebraic  quantum field theory \cite{Haag1996aa}
the fundamental structures are the algebras of \emph{local observables}.

Constructive field theory is some kind of compromise between both points of view. 
From the start it was conceived as a model building program \cite{Erice1973,MR887102,Riv1} in which
specific Lagrangian field theories, typically of the superrenormalizable and renormalizable type 
would be studied in increasing order of complexity. Its main characteristic is the mathematical rigor with which 
it addresses the basic issue of divergence of the perturbative series.

The founding success of constructive field theory was the construction of the ultraviolet \cite{Nelson1965aa} and thermodynnamic \cite{Glimm1973ab} limits of
the massive $\phi^4_2$ field theory  \cite{Simon1974aa} in Euclidean space. Thanks to Osterwalder-Schrader axioms
it implied the existence of a Wightman theory in real (Minkowski) space-time. Beyond this intial breakthrough, two other steps were critical 
for future developments. The first one was the introduction of multiscale analysis
by Glimm and Jaffe to build the more complicated $\phi^4_3$ model \cite{Glimm1973aa}. 
It was developped as a kind of independent mathematical counterpoint to Wilson's 
renormalization group. All the following progress in constructive field theory and in particular the construction of just renormalizable models 
relied in some way on deepening this basic idea of renormalization group and multiscale analysis \cite{Gawedzki1986fk,Feldman:1986lr}.

A bit later an other key mathematical concept was introduced in constructive field theory, namely Borel summability. It is 
a fundamental result of the constructive quantum field theory program that the Euclidean functional integrals
of many (Euclidean) quantum field theories with quartic interactions are the Borel sum of their renormalized perturbative series \cite{Eckmann1974aa,Magnen1977aa,Feldman:1986lr}. This result builds a solid bridge between the Feynman amplitudes
used by physicists and the Feynman-Kac functional integral which generates them.
Borel summable quantum field theories have indeed a \emph{unique} non-perturbative definition, 
independent of the particular cutoffs used as intermediate tools.
Moreover all information contained in such theories, including the so-called ``non-perturbative" issues,
is \emph{embedded} in the list of coefficients of the renormalized perturbative series. Of course to extract this information 
often requires an analytic continuation beyond the domains which constructive theory currently controls.

As impressive as may be the success of the standard model, it does not include gravity, the fundamental force which is the most obvious
in daily life. Quantization of gravity remains a central puzzle of theoretical physics.
It may require to use generalized quantum field theories with non-local interactions.
Indeed near the Planck scale, space-time should fluctuate so violently 
that the ordinary notion of locality may no longer be the relevant concept. Among the many arguments one can 
list pointing into this direction are the Doplicher-Fredenhagen-Roberts remark 
that to distinguish two objects closer than the Planck scale would require to concentrate so much energy in such a little volume that it would create a black hole, preventing the observation \cite{Doplicher1994aa}. String theory, which (in the case of closed strings) contains a gravitational sector, is another powerful reason to abandon strict locality. Indeed strings are one-dimensional \emph{extended} objects, whose interaction
cannot be localized at any space time point. Moreover, closed strings moving in compactified background may not distinguish between small and large 
such backgrounds because of dualities that exchange their translational and ``wrapping around'' degrees of freedom.
Another important remark is that in two and three dimensions pure quantum gravity is topological. In such theories, observables, being 
functions of the topology only, cannot be localized in a particular region of space-time. 

Many approaches currently compete towards a mathematically consistent quantization of gravity, and a 
constructive program in this direction may seem premature. Nevertheless
random tensor models have received recently increased attention, both
as fundamental models for random geometry pondered by a discretized Einstein-Hilbert action 
\cite{Rivasseau2013ac} and as efficient toy models of holography in the vicinity of a horizon
\cite{Witten2016aa,Gurau2016aa,Klebanov2017aa,Krishnan2017aa,Ferrari2017aa,Gurau2017aa,Bonzom2017aa}. 

Tensor models are background invariant and avoid (at least at the start) the formidable issue of fixing the
gauge invariance of general relativity under diffeomorphisms (change of coordinates).  Another advantage is that they
remain based on functional integrals. Therefore they can be investigated
with standard quantum field theory tools such as the renormalization group, 
and in contrast with many other approaches, with (suitably modified) constructive techniques. This paper is a step in that direction.

Random matrix and tensor models can be considered as a kind of simplification of Regge calculus \cite{Regge1961aa}, 
which one could call simplicial gravity or \emph{equilateral} Regge calculus \cite{Ambjorn2002aa}. Other important discretized approaches to quantum gravity
are the causal dynamical triangulations \cite{Loll2006aa,Ambjorn2013ab} and
group field theory \cite{Boulatov1992aa,Freidel2005aa,Krajewski2012aa,BenGeloun2010aa}, in which either causality constraints or holonomy and simplicity constraints are added to bring the discretization closer to the usual formulation of  general relativity in the continuum. 

Random matrices are relatively well-developed and have been used 
successfully for discretization of two dimensional quantum gravity \cite{David1985aa,Kazakov1985aa,Di-Francesco1995aa}. They have interesting 
field-theoretic counterparts, such as the renormalizable Grosse-Wulkenhaar model \cite{GrWu04-3,GrWu04-2,Disertori2007aa,Disertori2006lr,Grosse2009aa,Grosse2013aa,Grosse2014aa,Grosse2016aa}.

Tensor models extend matrix models and 
were therefore introduced as promising candidates for an \emph{ab initio} quantization of gravity 
in rank/dimension higher than 2 \cite{Ambjorn1991aa,Sasakura1991aa,Gross1992aa,Ambjorn2002aa}.
However their study is much less advanced since they lacked for a long time an analog of the famous 't~Hooft 
$1/N$ expansion for random matrix models \cite{t-Hooft1974aa} to probe their large $N$ limit.
Their modern reformulation
\cite{Gurau2011aa,Gurau2012ac,Gurau2011ad,Bonzom2012ac} considers
\emph{unsymmetrized} random tensors, a crucial improvement. 
Such tensors in fact have a larger, truly tensorial symmetry (typically in the complex case
a $U(N)^{\otimes d}$ symmetry at rank $d$ instead of the single $U(N)$ of symmetric tensors). 
This larger symmetry allows to probe their large $N$ limit through 
$1/N$ expansions of a new type \cite{Gurau2011ab,Gurau2011ac,Gurau2012aa,Bonzom2012ad,Bonzom2015aa,Bonzom2016aa}.

Random tensor models can be further divided into fully invariant models, in which both propagator and interaction are invariant, 
and field theories in which the interaction is invariant but the propagator is not \cite{Ben-Geloun2011aa}. 
This propagator can incorporate or not a gauge invariance of the Boulatov group field theory type. In such field theories the
use of tensor invariant interactions is the critical ingredient allowing in many cases for their successful renormalization \cite{Ben-Geloun2011aa,Ben-Geloun2012ab,Ousmane-Samary2012ab,BenGeloun2013aa,Carrozza2012aa,Carrozza2012ab}. Surprisingly
the simplest just renormalizable models turn out to be asymptotically free \cite{Ben-Geloun2012ab,Geloun2012ab,BenGeloun2012ab,Ousmane-Samary2013aa,Rivasseau2015aa}.

In all examples of random matrix and tensor models, the key issue is to understand in detail the limit
in which the matrix or the tensor has many entries. Accordingly, 
the main constructive issue is not simply Borel summability but uniform Borel summability
with the right scaling in $N$ as $N \to \infty$. 
In the field theory case the corresponding key issue is to prove Borel summability of the
\emph{renormalized} perturbation expansion without cutoffs. 

Recent progress has been fast on this front \cite{Rivasseau2016ab}. 
Uniform Borel summability in the coupling constant has been proven for vector, matrix and tensor \emph{quartic} 
models  \cite{Rivasseau2007aa,Magnen2009ab,Gurau2013ac,Delepouve2014ab,Gurau2014aa}, based on the loop vertex expansion (LVE) \cite{Rivasseau2007aa,Magnen2008aa,Rivasseau2013ab}, 
which combines an intermediate field representation\footnote{More recently the \LVEac
  has been extended to higher order interactions by introducing
  another related functional integral representation called the
  \emph{loop vertex representation}. It is based on the idea of
  forcing functional integration of a single field per vertex \cite{Rivasseau2017aa}. For quartic models like the one studied in this paper, this other representation
is however essentially equivalent to the intermediate field representation.} with the use of a \emph{forest formula} \cite{Brydges1987aa,Abdesselam1995aa}.
This relatively recent constructive technique is adapted to the study of theories without any space-time, as it works more directly at the combinatorial 
level and does not introduce any lattice. 
It was introduced precisely to make constructive sense of 't~Hooft $1/N$ expansion for quartic matrix models 
\cite{Rivasseau2007aa,Gurau2014aa}. 

The constructive tensor field theory program started in \cite{Delepouve2014aa}, in which Borel summability of the renormalized series has 
been proved  for the simplest such theory which requires some infinite renormalization, namely the $U(1)$ rank-three model with inverse Laplacian propagator and quartic interactions nicknamed $T^4_3$. 
This model has power counting similar to the one of $\phi^4_2$.
The main tool is the multiscale loop vertex expansion (MLVE) \cite{Gurau2014ab}, which combines
an intermediate field representation with the use of a more complicated \emph{two-level jungle formula} \cite{Abdesselam1995aa}.
An important additional technique is the iterated Cauchy-Schwarz bounds which allow 
to bound the LVE contributions. They are indeed not just standard perturbative amplitudes, but include resolvents which are delicate to bound.

The program has been also extended recently to similar models with Boulatov-type group field theory projector
\cite{Lahoche2015ab,Lahoche2015ac}. 

The next natural step  in this constructive tensor field theory program is to build the $U(1)$ rank-four model with inverse Laplacian propagator and quartic melonic interactions, which we nickname $T^4_4$. This model is comparable in renormalization difficulty to the ordinary $\phi^4_3$ theory, hence requires several additional non-trivial arguments. This is the problem we solve in the present paper.\\

The plan of this paper essentially extends the one of \cite{Delepouve2014aa}, as we follow roughly the same 
general strategy, but with many important additions
due to the more complicated divergences of the model. As the proof of our main result, namely \cref{thetheorem}, is somewhat
lengthy, we now outline its main steps and use this occasion to give
the actual plan of this paper and to define the various classes of Feynman graphs we
will encounter.\\
 
In \cref{model} we provide the mathematical definition of the
model. Its original or tensor representation is given in
\cref{sec-lapl-bare-renorm} as well as the full list of its
perturbative counterterms. This model is a quantum field theory the
fields of which are tensors namely elements of $\ell_{2}(\Z)^{\otimes
  4}$. As usual in quantum field theory, it is convenient to represent
analytical expressions by Feynman graphs. The latter will cover many
different graphical notions. As a first example, the Feynman graphs of the tensor
field theory under study here (see \cref{eq-Z0,eq-Z}) will be called \emph{tensor
graphs}. They will be depicted as (edge-)coloured graphs like in \crefrange{f-masdivergences}{f-VacuumNonMelonicDivergences}.

\Cref{sec-interm-field-repr} then provides the intermediate
field representation, at the heart of the \LVE. It rewrites the
partition function as a functional integral over both a main Hermitian matrix
intermediate field $\sigma$ and an auxiliary intermediate
field $\tau$ (which is also a matrix). We will simply write \emph{graphs} for
the Feynman graphs of the intermediate field representation of the
model, whereas these ``graphs'' are maps really, since intermediate
fields are matrices.

A multiscale decomposition is
introduced in \cref{sec-multiscale-analysis}.\\

\Cref{sec-mult-loop-vert-exp} provides the multiscale loop vertex
expansion (hereafter \MLVEac) for that model, which is surprisingly close to the one used
in \cite{Delepouve2014aa}, with just a little bit of extra structure
due to a single one of the ten divergent vacuum graphs of the
theory. \MLVEac consists in an ordered Bosonic and Fermionic 2-jungle
formula which expresses each ``order'' $n$ of the partiton function $\cZ$
(or the moments of the to-be-defined functional measure) as a sum over
forests on $n$ nodes. One of the benefits of such an expansion is that
the free energey \ie the logarithm of the partition function can very
easily be expressed as a similar sum but over \emph{connected} jungles
namely some sort of trees.
\begin{defn}[Trees, forests and jungles]\label{def-TreeForestJungle}
  A \firstdef{forest} on $\gls{nset}\defi\set{1,2,\dotsc,n}$ is an acyclic graph the
  vertex-set of which is $[n]$. A \firstdef{tree} is a connected acyclic
  graph. Connected components of forests are trees. Note that the
  graph with one vertex and no edges is considered a tree. A
  ($2$-)\firstdef{jungle} is a forest the edges of which are marked either
  $0$ or $1$. The vertices of a jungle are called \firstdef{nodes}.
\end{defn}
Jungles on $[n]$ will index the various terms composing order
$n$ of the \LVE of the partition function and of the free energy of our
model. More precisely, a jungle comes equipped
\begin{itemize}
\item with a scale attribution of its nodes (\ie a function
from the set of its nodes to the non-negative integers smaller than a
general UV cutoff $\jm$),
\item and intermediate field derivatives at both ends of each of its edges.
\end{itemize}
Each node $a$ of a jungle represents a functional expression, namely
$W_{j_{a}}=e^{-V_{j_{a}}}-1$ where $V_{j_{a}}$ is the quartic
interaction of the model at scale $j_{a}$. The \MLVEac expresses $\log\cZ$
as follows:
\begin{multline}
  \tag{\ref{eq-treerep}}
  \cW_{\les\jm}(g)\defi \log \cZ_{\les\jm}(g)=
  \sum_{n=1}^\infty \frac{1}{n!}  \sum_{\cJ\text{ tree}}
  \,\sum_{j_1=1}^{\jm} 
  \dotsm\sum_{j_n=1}^{\jm}\\
  \int d\tuple{w_{\ccJ}} \int d\nu_{ \ccJ}  
  \,\partial_{\ccJ}   \Bigl[ \prod_{\cB} \prod_{a\in \cB}   \bigl(   -\bar \chi^{\cB}_{j_a}  W_{j_a}   (\sigmad^a , \taud^a )  
  \chi^{ \cB }_{j_a}   \bigr)  \Bigr]
\end{multline}
where $\cB$ represents a connected component of the Bosonic part of
the jungle $\cJ$. Each Bosonic block $\cB$ is thus a subtree of
$\cJ$. Our main result, \cref{thetheorem}, consists in the analyticity of
$\lim_{\jm\to\infty}\cW_{\les\jm}(g)$ in a non empty cardioid domain of the
complex plane as well as the Borel summability of its perturbative
\emph{renormalised} series. The rest of the paper is entirely devoted to its proof.\\

The jungles of the \MLVEac are considered hereafter \emph{abstract
  graphs}. Each edge of an abstract forest comes equipped with
intermediate field derivatives at both of its ends (represented by the
$\partial_{\cJ}$ operator in the preceding equation). 
The result of these derivatives (\wrt the $\sigma$- and $\chi$-fields) on
the $W_{j_{a}}$'s is a sum, the terms of which can be indexed by still another type of graphs that
we name \emph{skeletons}, see \cref{sec-comp-boson-integr}.
\begin{defn}[Skeleton graphs]\label{def-skeletons}
  \emph{Skeleton graphs} are plane forests possibly with external
  edges, marked subgraphs, marked external edges and marked
  corners. External edges are unpaired half-edges. We will denote
  skeleton graphs with sans serif characters such as $\gls{skG}$. The possibly marked
  subgraphs are \nuonefig{} and \nutwofig. The marked ones will be depicted
  in gray and basically represent renormalised amplitudes of $2$-point
  subgraphs noted respectively $D_{1}$ and $D_{2}$. Unmarked external
  edges will be pictured \nuonefig[180] and marked ones by dotted
  lines \resins[180]. The latter represent resolvent insertions. Each vertex of a skeleton graph
  has a unique marked corner (\ie an angular sector between two
  consecutive half-edges, marked or not, adjacent to a same
  vertex). Each such marked corner bears an integer between $1$ and $\floor{\frac{m+1}{2}}+1$ if the graph has $m$ vertices.
\end{defn}
Let us consider a skeleton graph $\skG(\cJ)$ derived from a jungle $\gls{J}$ on
$[n]$. Thanks to the Faà di Bruno formula, \cref{eq-partitio}, each
node $a$ of $\cJ$ might be split into several (in fact up to the degree of
$a$) vertices of $\skG$. For $a\in [n]$, let $V_{a}(\skG)$ be the subset
of vertices of $\skG$ originating from node $a$ of $\cJ$. The set
$\set{V_{a}(\skG),\,a\in [n]}$ forms a partition of $V(\skG)$. For all
$a\in [n]$, the marked corners of the vertices in $V_{a}(\skG)$ bear
integer $a$.\\

To reach analyticity of $\cW_{\les\jm}$ we prove that it converges
normally. We must then compute an upper bound on the module of its
order $n$. The Fermionic integrale is standard and can be performed
exactly, see \cref{sec-grassmann-integrals}. It leads to the following bound
  \begin{align*}
 \vert \cW_{\les\jm}(g) \vert   &\les\sum_{n=1}^\infty \frac{2^{n}}{n!} \sum_{\cJ\text{
        tree}}\,\sum_{\set{j_{a}}} \Bigl( \prod_{\cB}
  \prod_{\substack{a,b\in \cB\\a\neq b}} (1-\delta_{j_aj_b})
  \Bigr) \Bigl( \prod_{\substack{\ell_F \in
      \cF_F\\\ell_F=(a,b)}} \delta_{j_{a } j_{b } } \Bigr)\
  \prod_{\cB}|I_{\cB}|,\\
  I_{\cB}&= \int d\tuple{w_{\cB}}\int d\nu_{\cB}\,\partial_{\cT_{\cB}}\prod_{a\in \cB}  ( e^{-V_{j_a}} -1  ) (\sigmad^a, \taud^a).
  \end{align*}
The main difficulty resides in the estimation of the Bosonic
contributions $I_{\cB}$. A Hölder inequality rewrites it as (see \cref{eq-CS-Pert-NonPert})
\begin{equation*}
  \abs{I_{\cB}}\les I_{\cB}^{\mathit{NP}}\, \sum_\skG \Bigl( \underbrace{\int d \nu_\cB\, \abs{\wo{A_\skG(\sigmad)}}^4}_{perturbative} \Bigr)^{\!1/4}.
\end{equation*}
This bound consists in two parts: a perturbative one, the terms of
which are indexed by skeleton graphs $\skG$ and a non perturbative
one, $I_{\cB}^{\mathit{NP}}$, made of exponentials of interaction terms and
counterterms. \Cref{sec-expl-form-v_j,sec-funct-integr-bounds} are
devoted to the non perturbative terms and lead in particular to
\cref{thm-GeneralnpBound}. \Cref{sec-expl-form-v_j} is a technical
preparation for the next section and consists in proving two very different but essential  
bounds, one of which is
\emph{quadratic}, see \cref{thm-lemmaquadbound}, and the other \emph{quartic},
\cref{thm-eighticbound}, on the main part of $V_{j}$. In
\cref{sec-funct-integr-bounds} we find some echo of the main Glimm and
Jaffe idea of expanding more and more the functional integral at
higher and higher energy scale \cite{Glimm1973aa}. Indeed to
compensate for linearly divergent vacuum graphs we need to push quite
far a Taylor expansion of the non-perturbative factor. However of
course a key difference is that there are no geometrical objects such
as the scaled ``Russian dolls'' lattices of cubes so central to
traditional multiscale constructive analysis.\\

In \cref{sec-pert-funct-integr} we bound the perturbative terms in
$I_{\cB}$ using an improved version of the Iterated Cauchy-Schwarz bounds.
Indeed the trees of the \LVEac and \MLVEac are not perturbative; they still
resum infinite power series through resolvents, which are however
uniformly bounded in norm, see \cref{thm-lemmaresbounded}. The ICS bound is a technique which allows
to bound such ``quasi-perturbative'' \LVEac contributions by truly
perturbative contributions, but with no longer any resolvent
included. More precisely, remember that skeleton graphs $\skG$ are
intermediate field graphs (thus maps) both with unmarked external edges
(corresponding to $\sigma$-fields still to be integrated out) and
marked ones representing resolvents. We first get rid of those
external $\sigma$-fields by integrating by parts (\wrt the Gaussian measure $d\nu_{\cB}$), see
\cref{eq-intbyparts}, in what we call the contraction process (see
\cref{sec-contraction-process}). Note that unmarked external edges
will then be paired both with marked and unmarked external edges. When
an unmarked external edge contracts to another unmarked external edge,
it simply creates a new edge. But when it contracts to a marked
external edge, it actually creates a new corner, as depicted in
\cref{fig-ContractHalfEdges} and according to \cref{eq-DerivationOfSigma}.
\begin{figure}[!htp]
  \centering
  \includegraphics[align=b,scale=1.5]{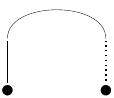}\quad $=$\quad \includegraphics[align=b,scale=1.5]{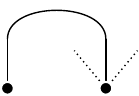}
  \caption{Contraction of half-edges in skeleton graphs.}
  \label{fig-ContractHalfEdges}
\end{figure}
The result of all the possible contractions of all the unmarked external edges
of a skeleton graph $\skG$ consists in a set of \emph{resolvent graphs}.
\begin{defn}[Resolvent graphs]\label{def-ResolventGraphs}
  A \emph{resolvent graph} is a map with external edges, marked
  subgraphs and marked corners. External edges, pictured \resins[180],
  represent resolvents. Possible marked subgraphs are the same than
  for skeleton graphs. Marked corners bear an integer between $1$ and
  $\floor{\frac{m+1}{2}}+1$ if the graph has $m$ vertices. Resolvent
  graphs will be denoted with calligraphic fonts such as $\gls{resG}$ for
  example. We also let $s(\resG)$ be the set of marked corners of
  $\resG$ and for any corner $c$ in $s(\resG)$, we let $i_{c}$ be the
  corresponding integer.
\end{defn}
Let $\skG$ be a skeleton graph and $\resG(\skG)$ one of the resolvent
graphs created from $\skG$ by the contraction process. As the latter
does not create nor destroy vertices, the sets of vertices of $\skG$
and $\resG$ have the same cardinality. Nevertheless the contraction
process may create new corners. In fact it creates two new corners each time an unmarked external
edge is paired to a marked external one. Thus there is a natural
injection $\iota$ from the corners of $\skG$ to the ones of
$\resG$. Moreover it is such that the marked corners of $\resG$ are
the images of the marked corners of $\skG$ via $\iota$.

Amplitudes of resolvent graphs still contain $\sigma$-fields in the
resolvents. In \cref{sec-iter-cauchy-schw} we will apply iterated
Cauchy-Schwarz estimates to such amplitudes in order to bound them by
the geometric mean of resolvent-free amplitudes, using that the norm
of the resolvent is bounded in a cardioid domain of the complex
plane. To this aim, it will be convenient to represent resolvent graph
amplitudes by the partial duals of resolvent graphs \wrt a spanning
subtree, see \cref{sec-iter-cauchy-schw}. It results in one-vertex
maps that we will actually represent as chord diagrams. Resolvents in
such maps will not be pictured anymore as dotted external edges but as
encircled $\fres$'s. See \cref{f-chorddiagexintro} for an example.
\begin{figure}[!htp]
  \centering
  \includegraphics[scale=.8]{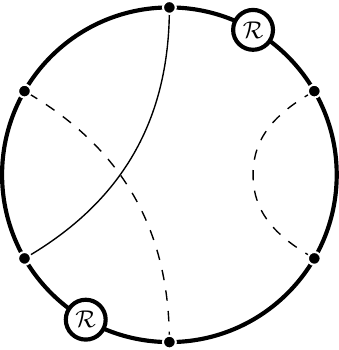}
  \caption{Example of the partial dual \wrt a spanning subtree of a
    resolvent graph, represented as a chord diagram. Edges of the tree
    correspond to plain lines whereas edges in the complement are dashed lines. Resolvent insertions are explicitely represented.}
  \label{f-chorddiagexintro}
\end{figure}

In \cref{sec-final-sums} we prove that the good power counting of
convergent amplitudes is sufficient to both compensate the large
combinatorial factors inherent in the perturbative sector of the
theory and sum over the scales $j_{a}$ of the jungle $\cJ$.\\

Finally appendices contain some of the proofs and details.

\section{The model}\label{model}

\subsection{Laplacian, bare and renormalized action}
\label{sec-lapl-bare-renorm}

Consider a pair of conjugate rank-4 tensor fields 
\begin{equation*}
\glslink{T}{T_{\ntup}}, \glslink{T}{\bar T_{\nbtup}}, \text{ with } \ntup = (n_1,n_2,n_3,n_4) \in \Z^4\text{, }\nbtup = (\bar n_1,\bar n_2,\bar n_3,\bar n_4 ) \in \Z^4.
\end{equation*}
They belong respectively to the tensor product 
$\gls{Htens} \defi \Hilb_1 \otimes \Hilb_2 \otimes\Hilb_3\otimes\Hilb_4$ and to its dual,
where each $\gls{Hi}$ is an independent copy of  $\ell_2 (\Z)=
L_2 (U(1))$, and  the colour or strand index $i$ takes values in $\set{1,2,3,4}$.
Indeed by Fourier transform these fields can be considered also as ordinary scalar fields  
$T  (\theta_{1},\theta_{2},\theta_{3}, \theta_4)$ and $\bar T(\bar \theta_{1},\bar \theta_{2},\bar \theta_{3}, \bar \theta_4 )$  
on the four torus $\Torus_4 = U(1)^4$ \cite{Ben-Geloun2011aa,Delepouve2014aa}.

If we restrict the indices $\ntup$ to lie in $[-N, N]^4$ rather than in $\mathbb{Z}^4$ we have a proper (finite dimensional) tensor model.
We can consider $N$ as the ultraviolet cutoff, and we are interested
in performing the ultraviolet limit $N \to \infty$.

Unless specified explicitly, short notations such as $\sum_{\ntup}$,
$\prod_{\ntup}$ mean either cutoff sums $\sum_{\ntup \in [-N, N]^4}$,
 $\prod_{\ntup \in [-N, N]^4}$ in the initial sections of this paper, before renormalization has been performed, or simply $\sum_{\ntup \in \mathbb{Z}^4}$
and $\prod_{\ntup \in \mathbb{Z}^4}$ in the later sections when renormalization has been performed.\\

We introduce the normalized Gaussian measure
\begin{equation*}
d\mu_C(T, \bar T)\defi \left(\prod_{\ntup, \nbtup} \frac{dT_n d\bar
    T_{\bar n}}{2i\pi} \right) \det C^{-1} \ 
e^{-\sum_{n, \bar n} T_n C^{-1}_{n\bar n}\bar T_{\bar n}}
\end{equation*}
where the covariance $\gls{C}$ is  the inverse of the Laplacian on $\Torus_4$ plus a unit mass term
\begin{equation*}
C_{\ntup,\nbtup}=\frac{\delta_{\ntup,\nbtup}}{\ntup^{2}+1},\, \ntup^{2}\defi n_{1}^2+n_2^2+n_3^2 +n_4^2.
\end{equation*}
The \emph{formal}\footnote{Here \emph{formal} simply means that $\cZ_{0}$ is
  ill-defined in the limit $N\to\infty$.} generating function for the moments of the model is then
\begin{equation}\label{eq-Z0}
\gls{Z}[_{0}](g,J, \bar J)= \gls{cN}    \int e^{T\scalprod\bar J+ J\scalprod\bar T}
e^{-\frac{g}{2} \sum_c V_c(T, \bar T)} d\mu_C(T, \bar T),
\end{equation}
where the scalar product of two tensors $A\scalprod B$ means
$\sum_{\ntup}  A_{\ntup}  B_{\ntup}$, $\gls{g}$ is the coupling constant, the source tensors $J$ and $\bar J$
are dual respectively to $\bar T$ and $T$ and $\cN$ is a
normalization. To compute correlation functions it is common to choose ${\cN}^{-1}=\int e^{-\frac{g}{2} \sum_c V_c(T,
  \bar T)} d\mu_C(T, \bar T)$ which is the sum of all vacuum amplitudes. However following the constructive tradition for such superrenormalizable models,
we shall limit $\cN$ to be the exponential of the finite
sum of the \emph{divergent} connected vacuum amplitudes. The interaction is
$\sum_{c}V_{c}(T,\bar T)$ with
\begin{equation}\label{eq-originalInteraction}
 \gls{Vc}(T, \bar T)\defi\Tr_{c}(T\Itens_{\hat c}\bar T)^{2}
 =\sum_{\substack{n_{c},\bar n_{c},\\m_{c},\bar m_{c}}} \Big(\sum_{\ntup_{\hat
     c},\nbtup_{\hat c}} T_{\ntup}\bar T_{\nbtup}\,\delta_{\ntup_{\hat c} \nbtup_{\hat c}} \Big) \delta_{n_c \bar m_c} 
     \delta_{m_c \bar n_c}\Big(\sum_{\mtup_{\hat c}, \mbtup_{\hat c}} T_{\mtup}\bar T_{\mbtup}\,\delta_{\mtup_{\hat c} \mbtup_{\hat c}}  \Big),
\end{equation}
and where $\Tr_{c}$ means the trace over $\Hilb_{c}$, $\ntup_{\hat c}\defi\set{n_{c'}, c'\neq c}$ (and similarly for
$\nbtup_{\hat c}$, $\mtup_{\hat c}$, $\mbtup_{\hat c}$) and
$(\Itens_{\hat c})_{\ntup_{\hat c}\nbtup_{\hat c}}=\delta_{\ntup_{\hat c}\nbtup_{\hat c}}$. Hence it is the symmetric sum of the four quartic melonic interactions of random tensors at rank four \cite{Delepouve2014ab}
with equal couplings.

This model is globally symmetric under colour permutations and has a 
power counting almost similar to the one of ordinary $\phi^4_3$
\cite{Glimm1973aa,Feldman1976aa,Magnen1976aa}. It has eleven divergent
graphs (regardless of their colours) including two (melonic) two-point graphs: the tadpole $\cM_1$,
linearly divergent, and the graph $\cM_2$ $\log$ divergent (see \cref{f-masdivergences}).  Note that each of these eleven graphs has
several coloured versions. For example, there are four different
coloured graphs corresponding to $\cM_{1}$, sixteen to $\cM_{2}$, and
ten to the unique melonic divergent vacuum graph of order two (see \cref{f-VacuumMelonicDivergences}).
\begin{figure}[!htp]
  \centering
  \begin{subfigure}[b]{.3\linewidth}
    \centering
    \includegraphics[height=3cm]{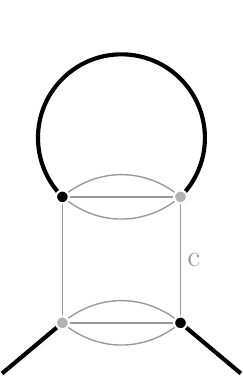}
    \caption{$\cM_{1}^{c}$}
    \label{f-masdivergences-1}
  \end{subfigure}
  \begin{subfigure}[b]{.3\linewidth}
    \centering
    \includegraphics[height=5cm]{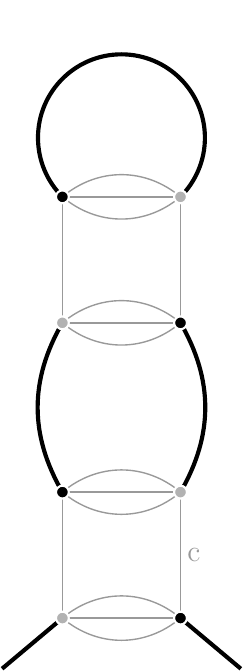}
    \caption{$\cM_{2}^{c}$}
    \label{f-masdivergences-2}
  \end{subfigure}
  \caption{The two divergent (melonic) two-point graphs. The melonic quartic vertex is shown with gray edges, and the bold edges correspond to
    Wick contractions of $T$ with $\bar T $, hence bear an inverse
    Laplacian.}
  \label{f-masdivergences}
\end{figure}
\begin{figure}[!htp]
  \centering
  \begin{tikzpicture}[every node/.style={node distance=0cm}]
    \node (un) at (0,0) {\includefigp[\scale]{-.5}{figures/videmelon1}};
    \node (deux) [right=of un] {\includefigp[\scale]{-.5}{figures/vide2quad}};
    \node (trois) [right=of deux]
    {\includefigp[\scale]{-.5}{figures/vide3lin}};
    \node (quatre) [right=of trois] {\includefigp[\scale]{-.5}{figures/vide4log}};
    \node (c3) [above right=0cm and 0cm of quatre.east]
    {\includefigp[\scaletwo]{.05}{figures/circle3}};
    \node (c4) [below right=-1cm and -1cm of quatre.east]
    {\includefigo[scale=\scaletwo,angle=45]{.05}{figures/circle4}};
  \node (c3a) [below right=-3.1cm and -1cm of c4.east] {\includefigo[scale=\scaletwo]{.05}{figures/circle3asym}};
  \end{tikzpicture}
  \caption{The seven divergent melonic vacuum connected graphs.}
  \label{f-VacuumMelonicDivergences}
\end{figure}
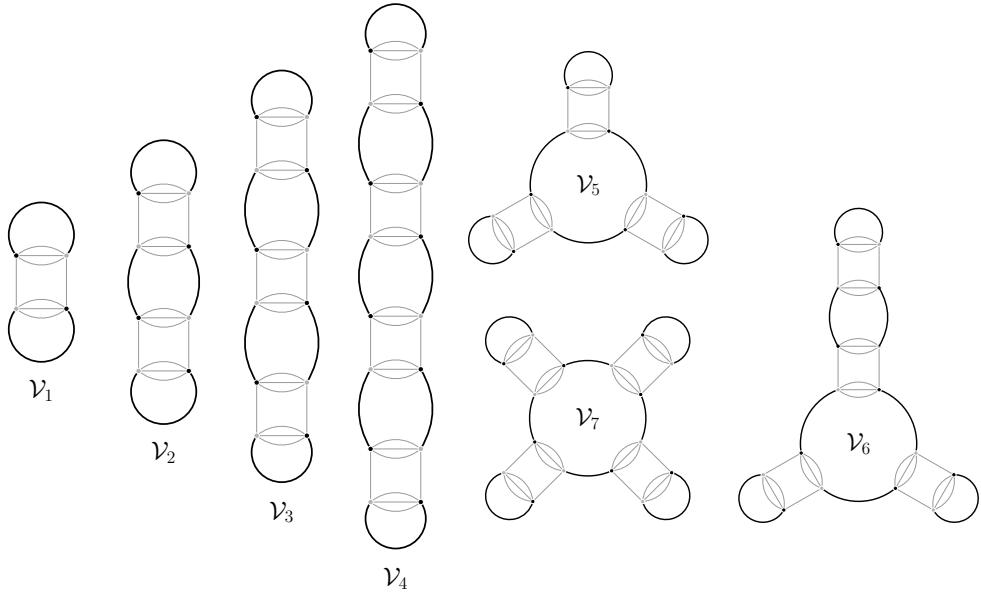
\begin{figure}[!htp]
  \centering
  \begin{subfigure}[c]{.2\linewidth}
      \centering
      \includegraphics[scale=.8]{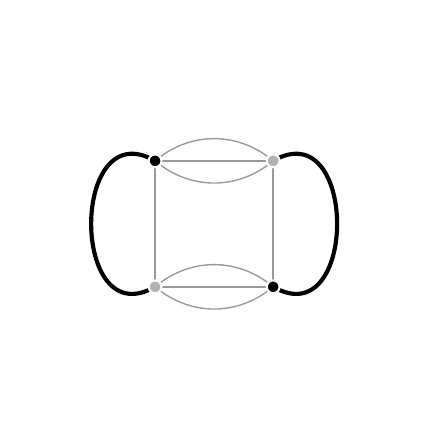}
      \caption{$\kN_{1}$}
      \label{f-VacuumNonMelonicDivergences-1}
    \end{subfigure}%
    \begin{subfigure}[c]{.3\linewidth}
      \centering
      \includegraphics[scale=.8]{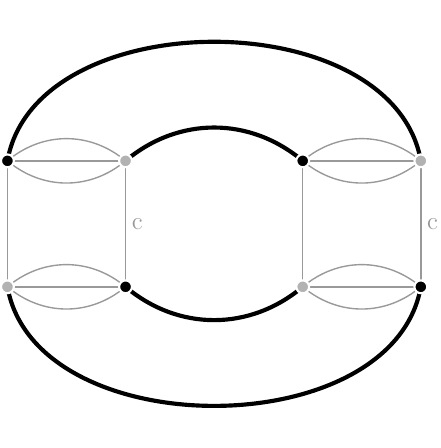}
      \caption{$\kN_{2}$}
      \label{f-VacuumNonMelonicDivergences-2}
    \end{subfigure}%
    \quad
    \begin{subfigure}[c]{.3\linewidth}
      \centering
      \includegraphics[scale=.8]{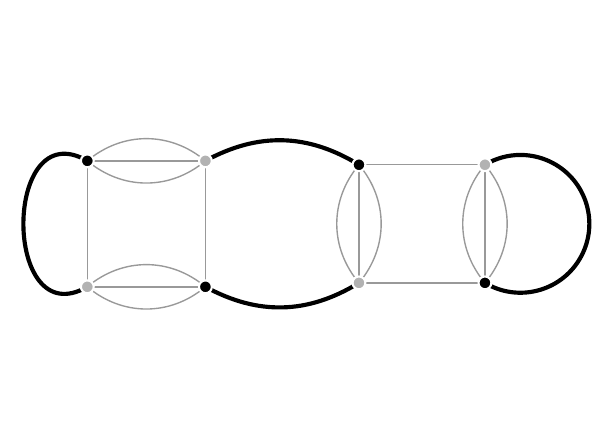}
      \caption{$\kN_{3}$}
      \label{f-VacuumNonMelonicDivergences-3}
    \end{subfigure}
    \caption{The three divergent non-melonic vacuum connected graphs.}
  \label{f-VacuumNonMelonicDivergences}
\end{figure}
\clearpage
The main problem in quantum field theory is to compute
$\cW(g,J, \bar J)  = \log  \cZ(g,J, \bar J)$ which is the generating function for the \emph{connected} Schwinger functions
\begin{equation*}
S_{2k} (\bar n_1 , \dots , \bar n_k; n_1 , \dots,  n_k ) = \frac{\partial}{\partial J_{\bar n_1}}\dotsm\frac{\partial}{\partial J_{\bar n_k}}  
 \frac{\partial}{\partial \bar J_{n_1}}\dotsm
 \frac{\partial}{\partial\bar J_{n_k}}  \cW(g,J, \bar J) \vert_{J= \bar J
   =0}.
\end{equation*}
Thus our main concern in this work will be to prove the
analyticity (in $g$) of $\cW$ in some (non empty) domain of the
complex plane, in the limit $N\to\infty$. Of course there is no chance
for $\cZ_{0}$ to be well defined in this limit and some (well-known)
modifications of the action have to be done, namely it has to be
supplemented with the counterterms of all its divergent subgraphs.\\

Let $\gls{tensorG}$ be a tensor Feynman graph and $\tau_{G}$ be the operator which sets to
$0$ the external indices of its Feynman amplitude $A_{G}$. The counterterm
associated to $G$ is given by
\begin{equation*}
  \gls{deltaG}=-\tau_{G}\big(\sum_{\cF\niton G}\prod_{g\in\cF}-\tau_{g}\big)A_{G}
\end{equation*}
where the sum runs over all the forests of divergent subgraphs of $G$
which do not contain $G$ itself (including the empty one). The
renormalized amplitude of $G$ is 
\begin{equation*}
  \gls{Ar}[_{G}]= (1- \tau_{G}) \big(\sum_{\cF\niton G}\prod_{g\in\cF}-\tau_{g}\big)A_{G}.
\end{equation*}
The behaviour of the renormalized amplitudes at large external momenta
is a remainder of the initial power counting of the graph. In
particular, let  $\cM$ be the set of the two divergent $2$-point
graphs, namely $\gls{cM}=\set{\cM_{1},\cM_{2}}$. Their renormalized
amplitudes are (neither including the coupling constants nor the
symmetry factors, and seen as linear operators on $\Htens$)
\begin{equation}
  \lb\begin{aligned}
    \gls{Ar}[_{\cM_{1}}](\ntup,\nbtup) &= \sum_c a_1 (n_c) \delta_{\ntup,\nbtup},\\ 
    \gls{Ar}[_{\cM_{2}}] (\ntup,\nbtup) &= \sum_c  \big(a_2 (n_c) + \sum_{c' \neq c}  a^{c'}_2 (n_c)\big)\delta_{\ntup,\nbtup}, \\
    a_1 (n_c) &= \sum_{p \in {\mathbb Z}^4} \frac{\delta (p_c -n_c) -
      \delta (p_c) }
    {p^2+ 1}  = \sum_{p   \in {\mathbb Z}^3}   \frac{ n_c^2} {(n_c^2 + p^2+ 1)(p^2 +1)} , \\
    a_2 (n_c) &=      \sum_{p, q    \in {\mathbb Z}^4}    \frac{\delta (p_c -n_c)  - \delta (p_c) }{p^2+ 1}  \frac{\delta (q_c -n_c)  - \delta (q_c) }{(q^2+ 1)^2} ,  \\
    a^{c'}_2 (n_c) &= \sum_{p, q \in {\mathbb Z}^4} \frac{\delta
      (p_{c'} -q_{c'}) - \delta (p_{c'}) }{p^2+ 1} \frac{\delta (q_c
      -n_c) - \delta (q_c) }{(q^2+ 1)^2} .
  \end{aligned}\right.\label{eq-ampren15}
\end{equation}

Remark that $a^{c'}_2$ is in fact independent of $c'$.\\

From now on we shall use the time-honored constructive practice of noting $\Oun $ any inessential constant. The large $\ntup$ behaviour
of the renormalized graphs $\cM_1$ and $\cM_2$ is controlled by the following
\begin{lemma}\label{thm-Aren}
  Let $\ntup\in\Z^{4}$ and $\norm{\ntup}$ be
  $\sqrt{\sum_{i=1}^{4}n_{i}^{2}}$. Then
  \begin{equation*}
    \vert\Ar_{\cM_{1}}(\ntup,\nbtup)\vert  \les \Oun\norm{\ntup}\delta_{\ntup,\nbtup}, \quad
    \vert \Ar_{\cM_{2}} (\ntup,\nbtup)  \vert \les \Oun \log  (1+ \norm{\ntup}) \delta_{\ntup,\nbtup}. 
  \end{equation*}
\end{lemma}
\begin{proof}
  Elementary from \cref{eq-ampren15}.
\end{proof}

Let $\gls{V}$ be the set of divergent vacuum graphs of the model (\ref{eq-Z0}). For any Feynman graph $G$, let $|G|$ be
its order (a.k.a.\@ number of vertices). Then the regularized generating
function $\cZ$ of the renormalized Schwinger functions is defined by
\begin{align}
  \cZ_{N}(g,J, \bar J)&\defi\cN\int e^{T\scalprod\bar J+ J\scalprod\bar T}
e^{-\frac{g}{2} \sum_c V_c(T, \bar T)+T\scalprod
  \bar T \big(\sum_{G\in\cM}\frac{(-g)^{|G|}}{S_{G}}\,\delta_{G}\big)}\, d\mu_C(T, \bar
T).\label{eq-Z}
\end{align}
where $S_{G}$ is the usual symmetry factor of the Feynman graph $G$,
and the normalization $\cN $ is, as announced, the exponential of
the finite sum of the counterterms of the divergent vacuum connected graphs, computed with cutoff $N$:
\begin{equation*}
  \gls{cN}   \defi \exp\Big(\sum_{G\in\gls{V}}\frac{(-g)^{|G|}}{S_{G}}\,\delta_{G}\Big).
\end{equation*}

As a final step of this section, let us rewrite \cref{eq-Z} a bit
differently. We want to absorb the mass counterterms in a
translation of the quartic interaction. So let us define
$g\gls{dm}\defi\sum_{G\in\cM}\tfrac{(-g)^{|G|}}{S_{G}}\,\delta_{G}$ and $\delta
_{m}\fide\sum_{c}\delta^{c}_{m}$. Then the integrand of $\cZ_{N}$
contains $e^{-g\sum_{c}I_{c}}$ with
\begin{align*}
  I_{c}&=\tfrac 12V_{c}(T,\bar T)-\delta_{m}^{c}T\scalprod\bar T=\tfrac
  12\Tr_{c}(T\Itens_{\hat c}\bar T)^{2}-\delta_{m}^{c}T\scalprod\bar T.\\
  \intertext{By simply noting that for all $c$, $T\scalprod\bar
    T=\Tr_{c}(T\Itens_{\hat c}\bar T)$, we get}
  I_{c}&=\tfrac 12\Tr_{c}(T\Itens_{\hat c}\bar
  T-\delta_{m}^{c}\Itens_{c})^{2}-\tfrac 12(2N+1)(\delta_{m}^{c})^{2}.
\end{align*}
Thus $\cZ_{N}$ rewrites as
\begin{equation*}
  \cZ_{N}(g,J, \bar J)=   \cN   e^{\delta_{t} }   \int e^{T\scalprod\bar J+ J\scalprod\bar T}
e^{-\frac{g}{2} \sum_c \Vr_c(T, \bar T)}\, d\mu_C(T, \bar T),
\end{equation*}
where $\gls{Vcr}(T,\bar T)\defi\Tr_{c}(T\Itens_{\hat c}\bar T-\delta_{m}^{c}\Itens_{c})^{2}$ and
\begin{equation*}
\gls{dt}\defi   \tfrac g2  \sum_{c} \Tr_{c} \Itens_{c}  (\delta_{m}^{c})^{2} =   \tfrac g2  (2N+1) \sum_{c}   (\delta_{m}^{c})^{2},
\end{equation*}
where the last equality uses the particular form of the cutoff $[-N, N]$.

\subsection{Intermediate field representation}
\label{sec-interm-field-repr}

The main message of the Loop Vertex Expansion (a.k.a.\@ LVE) is that it is
easier (and to a certain extent better) to perform constructive
renormalization within the intermediate field setting. Initially
designed for matrix models \cite{Rivasseau2007aa} LVE has proven to be
very efficient for tensor models in general \cite{Gurau2013ac}.

\subsubsection{Integrating out the tensors}
\label{sec-integr-out-tens}

So we now decompose the four interactions $\Vr_c$ by introducing 
four intermediate Hermitian $N\times N$ matrix fields $\sigma^{\transpose}_{c}$ acting on
$\Hilb_c$ (here the superscript $T$ refers to transposition). To simplify the formulas we put $g \fide \gls{lambda}[^2]$ and write
\begin{equation*}
 e^{-\frac{\lambda^2}{2} \Vr_c(T, \bar T) }= \int e^{i\lambda\Tr_{c}\big[(T\Itens_{\hat c}\bar T-\delta_{m}^{c}\Itens_{c})\sigma^{\transpose}_{c}\big]}\,d\nu(\sigma^{\transpose}_{c})
\end{equation*}
where $d\nu(\sigma^{\transpose}_{c})=d\nu(\sigma_{c})$ is the GUE law of covariance $1$. $\cZ_{N}(g,J, \bar J)$ is now a Gaussian integral over $(T, \bar T) $, hence can be evaluated:
\begin{align}
\cZ_{N}(g,J, \bar J)&=    \cN e^{\delta_{t}}   \int \Big(\prod_c d\nu(\sigma_{c})\Big)
 d\mu_{C}(T,\bar T)\, e^{T\scalprod\bar J+ J\scalprod\bar
   T}e^{i\lambda(T\sigma\bar T-\sum_{c}\delta_{m}^{c}\Tr_{c}\sigma_{c})}\nonumber\\
 &= \cN  e^{\delta_{t}}     \int \Big(\prod_c
 d\nu(\sigma_{c})\Big)e^{JC^{1/2}  R(\sigma)C^{1/2}  \bar
   J-\Tr\log(\Itens-\Sigma)-i\lambda \sum_{c}\delta_{m}^{c}\Tr_{c}\sigma_{c}}\label{eq-Zsigma}
\end{align}
where $\gls{sigma}\defi \sigma_{1}\otimes\Itens_2\otimes\Itens_3\otimes\Itens_4+\Itens_1\otimes\sigma_{2}\otimes\Itens_3 \otimes\Itens_4
+\Itens_1\otimes\Itens_2\otimes\sigma_{3}\otimes\Itens_4+\Itens_1\otimes\Itens_2\otimes\Itens_3\otimes\sigma_{4}$,
$\gls{Itens}$ is the identity operator on $\Htens$, $\Tr$ denotes the trace
over $\Htens$,
\begin{equation*}
\gls{Sigma} (\sigma)\defi i\lambda C^{1/2}\sigma C^{1/2}
\fide i\lambda\gls H
\end{equation*}
is the $\sigma$ operator sandwiched\footnote{Using cyclicity of the trace,
it is possible to work either with $C\sigma$ operators or with \emph{symmetrized} 
``sandwiched'' operators but the latter are more convenient for the
future constructive bounds of \cref{sec-funct-integr-bounds}.} with appropriate
square roots $C^{1/2}$ of propagators
and includes the $i \lambda$ factor, hence $H$ is always Hermitian and
$\Sigma$ is \emph{anti-Hermitian} for $g$ real positive. The \emph{symmetrized  resolvent} operator is
\begin{equation*}
\gls{R}(\sigma)\defi (\Itens-i\lambda C^{1/2}\sigma C^{1/2})^{-1}
=(\Itens-\Sigma)^{-1}.
\end{equation*}

In the sequel it will also be convenient to consider the inner product
space
$\gls{Hopdirect}\defi\Hop[1]\times\Hop[2]\times\Hop[3]\times\Hop[4]$ where
each $\gls{LHi}$ is the space of linear operators on $\Hilb_{i}$. Let
$\direct a$ and $\direct b$ be elements of $\Hopdirect$. Their inner
product, denoted $\direct a\scalprod\direct b$, is defined as
$\sum_{c}\Tr_{c}(a_{c}^{\dagger}b_{c})$. For any $\direct
a\in\Hopdirect$, to simplify notations, we will write its $c$-component
$(\direct a)_{c}$ as $a_{c}$. Similarly we define $\gls{sigd}$ as the element of
$\Hopdirect$ the $c$-component of which is $\sigma_{c}$. Finally let $\gls{Idirect}$ be the multiplicative
identity element ($\Idirect_{cc'}=\delta_{cc'}\Itens_{c}$) of the
linear operators $\gls{EndHopd}$ on $\Hopdirect$.

The Gaussian measure $\prod_{c}d\nu(\sigma_{c})$ is now interpreted as
the normalized Gaussian measure on $\Hopdirect$ of
covariance $\Idirect$ and denoted $d\nu_{\Idirect}(\sigmad)$.

\subsubsection{Renormalized action}
\label{sec-renormalized-action}

It is well known that each order of the Taylor expansion aroung $g=0$
of $\cZ_{N}$ (see \cref{eq-Z}) is finite in the limit $N\to\infty$. The counterterms
added to the action precisely compensate the divergences of the
Feynman graphs created by the bare action. Proving such a result is
by now very classical but still somewhat combinatorially involved. We
exhibit here one of the advantages of the intermediate field
representation. We are indeed going to rewrite \cref{eq-Zsigma} in
such a way that the compensations between terms and counterterms are
more explicit. Such a new form
of an action will be called \firstdef{renormalized
  $\sigma$-action}. The idea is to Taylor expand $\log(\Itens-i\lambda
C\sigma)$ ``carefully'', \ie order by order in a way somewhat
similar in spirit to the way multiscale analysis teaches us how to
renormalize a quantum field theory.

\paragraph{Order $1$}
So let us start with the first order of the $\log$:
\begin{equation*}
  \log(\Itens-i\lambda C\sigma )\fide -i\lambda C \sigma  +\log_{2}(\Itens-i\lambda C \sigma ),  
\end{equation*}
where $\log_p (1-x) = \sum_{k=1}^{p-1} x^k/k+ \log (1-x)$. The integrand now includes the exponential of a linear term in
$\sigma$, namely
$i\lambda(\Tr(C\sigma)-\sum_{c}\delta_{m}^{c}\Tr_{c}\sigma_{c})$. Recall
that
$\delta_{m}^{c}=-\delta_{\cM_{1}^{c}}+\lambda^{2}\delta_{\cM_{2}^{c}}$
(see \cref{sec-bare-renorm-ampl} for the explicit expressions). Let us
rewrite (part of) this linear term as follows:
\begin{equation*}
  i\lambda\big(\Tr(C\sigma)+\sum_{c}\delta_{\cM_{1}^{c}}\Tr_{c}\sigma_{c}\big)\fide
  i\lambda\gls{ArdMun}\scalprod\sigmad,\qquad(\direct{\Arop}_{\cM_{1}})_{c}=\Tr_{\hat c}C+\delta_{\cM_{1}^{c}}\Itens_{c}.
\end{equation*}
Note that $(\direct{\Ar}_{\cM_{1}})_{c}$ is, up to a
factor $\gls{Itens}[_{\hat c}]$, the truncated
renormalized amplitude of $\cM_{1}^{c}$, considered here as a
linear operator on $\Hilb_{c}$. Therefore
\begin{equation*}
  \cZ_{N}(g,J,\bar J) =\cN e^{\delta_{t}}\int d\nu_{\Idirect}(\sigmad)\,e^{JC^{1/2} R(\sigma) C^{1/2}  \bar
   J-\Tr\log_{2}(\Itens-\Sigma
   )+i\lambda\directb{\Arop}_{\cM_{1}}\scalprod\sigmads -i\lambda^{3}
   \sum_{c}\delta_{\cM_{2}^{c}}\Tr_{c}\sigma_{c}}.
\end{equation*}
The next step consists in translating
the $\sigmad$ field in order to absorb the
$i\lambda\direct{\Arop}_{\cM_{1}}\scalprod\sigmads$ term in the
preceding equation through a translation of integration contour for the diagonal part of $\sigmad$:
$\sigmad \to \sigmad+\direct B_1  $, where $\gls{dB1} \defi i\lambda
\direct{\Arop}_{\cM_{1}}$:
\begin{equation*}
  \cZ_{N}(g,J,\bar J) =\cN e^{\delta_{t}}\int
  d\nu_{\Idirect}(\sigmad-\direct B_{1})\,e^{JC^{1/2} R(\sigma) C^{1/2}  \bar
    J-\Tr\log_{2}(\Itens-\Sigma )-i\lambda^{3}
    \sum_{c}\delta_{\cM_{2}^{c}}\Tr_{c}\sigma_{c}+\tfrac 12(\directb B_{1})^{2}}.
\end{equation*}
To simplify the writing of the result of the translation, we introduce
the following notations:
\begin{equation*}
  \begin{alignedat}{2}
    \gls{Ar}[_{\cM_{1}}]={}&\sum_{c}(\direct\Arop_{\cM_{1}})_{c}\otimes\Itens_{\hat
      c}\in L(\Htens),&\qquad \gls{B1}&\defi i\lambda\Arop_{\cM_{1}},\\
    \gls{D1} \defi{}& i \lambda C^{1/2} B_1 C^{1/2},
    &\gls{U1}&\defi \Sigma +D_1,\\
    \gls{R1}(\sigma) \defi{}& (\Itens- U_1 )^{-1}.
  \end{alignedat}
\end{equation*}
Remark that
as $(\direct{\gls{Ar}}_{\!\!\cM_{1}})_{c}$ is diagonal (\ie
proportional to $\gls{Itens}[_{c}]$), the operator $\gls{B1}$ is
diagonal too: 
\begin{equation}
  (\gls{B1})_{\mtup\ntup}\fide\sum_{c}\gls{b1}(m_{c})\delta_{\mtup\ntup}.  \label{defB1}
\end{equation}

The partition function thus rewrites as
\begin{align*}
  \cZ_{N}(g,J,\bar J) ={}&\cN_{1} \int d\nu_{\Idirect}(\sigmad)\,e^{J
    C^{1/2} R_1(\sigma) C^{1/2} \bar J-\Tr\log_{2}(\Itens- U_1 ) - i\lambda^{3}
    \sum_{c}\delta_{\cM_{2}^{c}}\Tr_{c}\sigma_{c}}\\
  \gls{cN1} \defi{}& \cN e^{\delta_{t}} e^{\tfrac{1}{2}(\directb B_{1})^{2} - i
    \lambda^{3} \sum_{c}\delta_{\cM_{2}^{c}}\Tr_{c} (\directb B_{1})_{c}},
\end{align*}
provided the contour translation does not cross any singularity of the
integrand (which is proven in \cref{thm-translation}).

\paragraph{Order $2$}
We go on by pushing the Taylor expansion of the $\log$ to the next
order:
\begin{equation*}
  \log_{2}(\Itens-U_{1})=-\tfrac 12U_{1}^{2}+\log_{3}(\Itens-U_{1}).
\end{equation*}
Using $\Tr[D_1 \Sigma] - i  \lambda^{3} \sum_{c}\delta_{\cM_{2}^{c}} \Tr_{c} \sigma_{c}   = - i\lambda^3\gls{ArdMdeux}\scalprod\sigmad$, 
and adding and subtracting a term
$\Tr[D_{1}\Sigma^{2}]$ to prepare for the cancellation of the vacuum
non-melonic graph in \cref{f-VacuumNonMelonicDivergences-3}, we obtain
\begin{multline} \label{eqexp2}
\cZ_{N}(g,J,\bar J) = \cN_1\, e^{\frac 12\Tr D_{1}^{2}}\int d\nu_{\Idirect}(\sigmad)\,e^{J
C^{1/2} R_1(\sigma) C^{1/2} \bar J-\Tr\log_{3}(\Itens- U_1 )}\\
\times e^{\frac 12\Tr[\Sigma^{2}(\Itens+2\gls{D1})] - \Tr[D_{1}\Sigma^{2}]
- i\lambda^{3}\directb{\Ar}_{\cM_{2}}\scalprod\sigmads}
\end{multline}
where, as for $\cM_{1}$, $(\direct{\Ar}_{\cM_{2}})_{c}$ is the
truncated renormalized amplitude of $\cM_{2}^{c}$. We now define the operator 
$\gls{Q} \in\gls{EndHopd}$ as the real symmetric operator such that 
\begin{equation*}
\lambda^{2}\sigmad\scalprod\gls{Q}\sigmad = -\Tr[\Sigma^{2}(\Itens+2\gls{D1})].
\end{equation*}
Using \cref{defB1},
\begin{multline}
  (\gls{Q})_{cc';m_{c}n_{c},p_{c'}q_{c'}}\defi\delta_{cc'}\delta_{m_{c}p_{c}}\delta_{n_{c}q_{c}}\sum_{\mtup_{\hat
        c}}\frac{1}{(m_{c}^{2}+\mtup_{\hat
        c}^{2}+1)(n_{c}^{2}+\mtup_{\hat c}^{2}+1)}\\
    \shoveright{\times\big(1 + 2i\gls{lambda}\sum_{c"}
    \frac{\gls{b1}(m_{c"})}{m_{c}^{2}+\mtup_{\hat c}^{2}+1}\big)}\\
    +(1-\delta_{cc'})   \delta_{m_{c}n_{c}}\delta_{p_{c'}q_{c'}}\sum_{\tuple
    r\in[-N,N]^{2}}\frac{1}{(m_{c}^{2}+p_{c'}^{2}+\tuple
    r^{2}+1)^{2}}\\
  \times\Big(1 + \frac{2i\gls{lambda}}{m_{c}^{2}+p_{c'}^{2}+\tuple r^{2}+1}\big(\gls{b1}(m_{c}) + \gls{b1}(p_{c'})  + \sum_{c"\ne c,c' } \gls{b1}(r_{c"}) \big) \Big).\label{eq-Qexpr}
\end{multline}
It is also convenient to give a special name, $\gls{Q0}$, to
the leading part of $\gls{Q}$. More precisely $\gls{Q0}$ is
a diagonal operator, both in colour and in index space, defined by:
\begin{align}
  (\gls{Q0})_{cc';m_{c}n_{c},p_{c'}q_{c'}}&= \delta_{cc'}
  \sum_{\mtup_{\hat c},\tuple p_{\hat
      c}}C_{\substack{q_{c}n_{c}\\\tuple p_{\hat c}\tuple m_{\hat
        c}}}C_{\substack{m_{c}p_{c}\\\tuple m_{\hat c}\tuple p_{\hat
        c}}}\nonumber\\
  &= \delta_{cc'}
  \delta_{m_{c}p_{c}}\delta_{n_{c}q_{c}}\sum_{\mtup_{\hat c}}
  \frac{1}{(m_{c}^{2}+\mtup_{\hat c}^{2}+1)(n_{c}^{2}+\mtup_{\hat
      c}^{2}+1)}
  \label{eq-Q0expr},
\end{align}
so that minus half of its trace, which is linearly divergent, is precisely canceled by the $\delta_{\kN_1}$ counterterm 
\begin{equation*}
-\tfrac{\lambda^2}{2} \Tr  \gls{Q0} =  - \delta_{\kN_1} =   -\tfrac{\lambda^2}{2} \sum_{m_{c},n_{c}, \mtup_{\hat c}}       \frac{1}{(m_{c}^{2}+\mtup_{\hat
        c}^{2}+1)(n_{c}^{2}+\mtup_{\hat c}^{2}+1)}.
  \end{equation*}
We also define $\gls{Q1}\defi\gls{Q}-\gls{Q0}$. Remark that in $\Tr \gls{Q1}$, only the diagonal part of $\gls{Q}$ contributes, hence  $\Tr \gls{Q1}$
is exactly canceled by the counterterm for the graph $\kN_3$:
$-\tfrac{\lambda^2}{2}\Tr  \gls{Q1} = - \delta_{\kN_3}$. Consequently
\begin{equation}
-\tfrac{\lambda^2}{2} \sigmad\scalprod Q\sigmad + \delta_{\kN_1} +  \delta_{\kN_3} =   -\tfrac{\lambda^2}{2} (\sigmad\scalprod Q \sigmad - \Tr Q  )= -\tfrac{\lambda^2}{2} \wo{\sigmad\scalprod Q \sigmad}\label{eq-originalwo}
\end{equation}
is nothing but a \emph{Wick-ordered} quadratic interaction with
respect to the Gaussian measure $d\nu_{\Idirect}(\sigmad)$.
Therefore we can rewrite \cref{eqexp2} as 
\begin{equation*} 
\cZ_{N}(g,J,\bar J) = \cN_2\int d\nu_{\Idirect}(\sigmad)\,e^{JC^{1/2} R_1(\sigma) C^{1/2} \bar J
-\Tr\log_{3}(\Itens- U_1 )-\frac{\lambda^2}{2} \wo{\sigmads\scalprod Q \sigmads} - \Tr[D_{1}\Sigma^{2}]
- i\lambda^{3}\directb{\Ar}_{\cM_{2}}\scalprod\sigmads}
\end{equation*}
with $\gls{cN2} \defi \cN_1 e^{\frac 12\Tr D_{1}^{2}- \delta_{\kN_1} -\delta_{\kN_3}}$.\\

The counterterm $\delta_{\kN_{2}}$ for $\kN_{2}$ is a bit more difficult to express in this language since it corresponds 
to the Wick ordering of $\frac{\lambda^{4}}{4}\sigmad\scalprod \gls{Q0}[^2]\sigmad $. 
It is in fact a square: $\delta_{\kN_{2}}  = -
\frac{\lambda^{4}}{4}\Tr\gls{Q0}[^2]$. We first represent it as an integral over an auxiliary tensor $\gls{taud}$ which is also a collection of four random matrices $\tau^c_{mn} $:
\begin{equation*}
e^{ \delta_{\kN_{2}} } = \int d\nu_{\Idirect}(\taud)\, e^{i\frac{\lambda^{2}}{\sqrt 2} \gls{Q0} \scalprod \tauds}
\end{equation*}
where the scalar product is taken over both colour and $m,n$ indices
\ie
\begin{equation*}
Q_{0}\scalprod\taud\defi\sum_{c,m,n}(Q_{0})_{cc;mn,mn}\tau^{c}_{mn}.
\end{equation*}
Then
\begin{multline*}
  \cZ_{N}(g,J,\bar J) =\cN_{3} \int
  d\nu_{\gls{Idirect}}(\sigmad, \taud)\,e^{J C^{1/2} R_1(\sigma) C^{1/2} \bar  J
    -\Tr\log_{3}(\Itens- U_1 )}\\
  \times e^{-\frac{\lambda^2}{2}\wo{\sigmads\scalprod Q \sigmads} 
    + i  \frac{\lambda^{2}}{\sqrt 2} Q_0 \scalprod \tauds 
    - \Tr[D_{1}\Sigma^{2}] -i\lambda^{3}\directb{\Ar}_{\cM_{2}}\scalprod\sigmads}
\end{multline*}
with $\gls{cN3}\defi\cN_{2}\, e^{- \delta_{\kN_{2}}}$ and
$d\nu_{\Idirect}(\sigmad, \taud)\defi d\nu_{\Idirect}(\sigmad)\otimes d\nu_{\Idirect}(\taud)$. The next step of the rewriting of the $\sigma$-action consists in one
more translation of the $\sigmad$ field: $\gls{dB2}\defi -i\lambda^{3}\direct{\gls{Ar}}_{\!\!\cM_{2}}$,
\begin{multline*}
  \cZ_{N}(g,J,\bar J)
  =\cN_{3}\,e^{\tfrac 12\directb{B}_{2}\scalprod\directb{B}_{2}}
  \int d\nu_{\gls{Idirect}}(\sigmad-\gls{dB2}, \taud )\,e^{J
    C^{1/2} R_1(\sigma) C^{1/2} \bar J   -\Tr\log_{3}(\Itens- U_1 )}\\
  \times e^{-\frac{\lambda^2}{2}\wo{\sigmads\scalprod Q \sigmads} 
+ i  \frac{\lambda^{2}}{\sqrt 2} Q_0 \scalprod \tauds -\Tr[D_{1}\Sigma^{2}]}.
\end{multline*}
Finally we introduce the following notations:
\begin{equation}
  \begin{alignedat}{2}
    \Ar_{\cM_{2}}\defi{}&\sum_{c}(\direct\Arop_{\cM_{2}})_{c}\otimes\Itens_{\hat
      c},&\qquad \gls{B2}&\defi -i\lambda^{3}\Ar_{\cM_{2}},\\
    \gls{D2} \defi{}& i \lambda C^{1/2} B_2 C^{1/2},
    &\gls{U}&\defi \Sigma +D_1+D_{2}\fide\Sigma+\gls{D},\\
    \gls{cR}(\sigma) \defi{}& (\Itens- U)^{-1},& \qquad \widetilde{V}^{\ges 3} (\sigma)&\defi \Tr\log_{3}(\Itens-U).
  \end{alignedat}\label{eq-defABDUR2}
\end{equation}
Remark indeed that $-\Tr\log_3 (\Itens-U)$ expands as $\sum_{q \ges 3}  \Tr\frac{U^q}{q}$, which can be interpreted 
as a sum over cycles (also called \emph{loop vertices}) of length at least three
with $\sigma$ or $D$ insertions. 
We get
\begin{equation*}
  \begin{aligned}
    \cZ_{N}(g,J,\bar J) ={}&\cN_{4}\int
    d\nu_{\gls{Idirect}}(\sigmad , \taud )  \,e^{J C^{1/2} \cR(\sigma) C^{1/2} \bar
      J- \widetilde{V}^{\ges 3}(\sigma) - V^{\les 2}(\sigma, \tau)-\Tr[D_{1}\Sigma^{2}]  }\\
V^{\les 2}(\sigma, \tau) \defi{}& \tfrac{\lambda^2}{2}\wo{\sigmad\scalprod Q \sigmad} - i   \tfrac{\lambda^{2}}{\sqrt 2}\gls{Q0} \scalprod \taud -\Tr[D_{2}\Sigma]
\\
    \gls{cN4}\defi{}&\cN_{3}\,e^{\frac
      12\directb{B}_{2}\scalprod\directb{B}_{2} + \frac 12\Tr[D_2^{2}]}
  \end{aligned}
\end{equation*}
provided the contour translation does not cross any singularity of the
integrand, see again \cref{thm-translation}.\\

Returning to
\crefrange{f-masdivergences}{f-VacuumNonMelonicDivergences}, we see
that Feynman graphs made out solely of loop vertices of length at
least three are all convergent at the perturbative level, except the last three of the seven
divergent vacuum melonic graphs in \cref{f-VacuumMelonicDivergences},
which correpond respectively to a loop vertex of length three with
three $\cM_1$ insertions, a loop vertex of length three with two
$\cM_1$ insertions and one $\cM_2$ insertion, and a loop vertex of
length 4 with four $\cM_1$ insertions. The three missing terms
corresponding to the three remaining divergent vacuum graphs are $\gls{lastvac}\defi\Tr(\tfrac 13 D_{1}^{3}+D_{1}^{2}D_{2}+\tfrac
14D_{1}^{4})$. Once again, we add and substract those missing terms
from the action. Thus, defining $V^{\ges
  3}(\sigma)\defi\widetilde{V}^{\ges
  3}(\sigma)+\Tr[D_{1}\Sigma^{2}]+\lastvac$, we get
\begin{equation*}
\begin{aligned}
    \cZ_{N}(g,J,\bar J) ={}&\cN_{5}\int
   d\nu_{\gls{Idirect}}(\sigmad , \taud )  \,e^{J C^{1/2} \cR(\sigma) C^{1/2} \bar
      J-V^{\ges 3}(\sigma) - V^{\les 2}(\sigma, \tau)},\\
    \gls{cN5}\defi{}&\cN_{4}\,e^{\lastvac}.
  \end{aligned}
\end{equation*}
\begin{lemma}\label{thm-finitevacuum}
$\cZ^{(0)}(g)\defi\log  \cN_5=0$.
\end{lemma}
The proof is given in \cref{sec-vacuum-contrib}.\\

The goal is therefore from now on to build the $N \to \infty$ limit of
\begin{equation*}
\gls{cW}_{N}(g,J, \bar J) =\log\cZ_{N}(g,J,\bar J)
\end{equation*}
and to prove that it is the Borel sum  of its (well-defined and ultraviolet finite)
perturbative expansion in $g=\lambda^2$.  In fact, like in \cite{Delepouve2014aa}, we shall only prove the convergence theorem for the pressure 
\begin{equation}
\cW_{N}(g)\defi \cW_{N}(g,J, \bar J)\rvert_{J = \bar J =0} = \log  \cZ_{N}(g),\qquad\cZ_{N}(g)\defi\int d\nu_{\gls{Idirect}}(\sigmad , \taud )\,e^{-V}, \label{startingpoint}
\end{equation}
where the intermediate field interaction $V$ is
\begin{align*}
V\defi{}&V^{\ges 3} (\sigma)+V^{\les 2} (\sigma, \tau) 
,
\end{align*}
since adding the external sources leads to inessential technicalities
that may obscure the essential constructive argument, namely the 
perturbative and non perturbative bounds of \cref{sec-pert-funct-integr,sec-funct-integr-bounds}.

\subsection{Justifying contour translations}

In this subsection we prove  that the successive translations performed in the previous subsection did not cross singularities of the integrand. This will lead us to introduce some basic uniform bounds on $\gls D$
and $\gls{cR}$ when $g$ varies in the small open cardioid domain
$\gls{Cardrho}$ defined by $\vert g \vert < \rho \cos^{2}(\tfrac 12\arg g )$ (see \cref{cardio}).\\

\Cref{thm-Aren} easily implies
\begin{lemma}[$D, D_{1},D_{2}$ estimates]\label{thm-Dren}
  $\gls D$, $\gls{D1}$ and $\gls{D2}$ are compact operators on
  $\Htens$,
  diagonal in the momentum basis, with
  \begin{equation*}
    \sup ( \lvert(\gls{D})_{\mtup\ntup}\rvert,
    \lvert(\gls{D1})_{\mtup\ntup}\rvert ) \les \frac{\Oun  |g|}{
      1+\norm{\ntup}}\delta_{\mtup\ntup}, \quad
    \lvert  (\gls{D2})_{\mtup\ntup}\rvert \les \frac{\Oun \vert g \vert
      ^2 [1 + \log  (1+ \norm{\ntup}) ] }{1 +  \norm{\ntup}^2}\delta_{\mtup\ntup} . 
  \end{equation*}
\end{lemma} 

\begin{lemma}[Resolvent bound] \label{thm-lemmaresbounded}
For $g$ in the small open cardioid domain $\Card_\rho$,
the translated resolvent $\fres =  (\Itens  -U  )^{-1} $
is well defined and uniformly bounded:
\begin{equation*}
  \norm{\gls{cR}}\les  2  \cos^{-1}(\tfrac 12\arg g) .
\end{equation*}
\end{lemma}
\begin{proof}
  In the cardioid domain we have $\lvert\arg g\rvert < \pi$. For any
  self-adjoint operator $L$, by the spectral mapping theorem
  \cite[][Theorem $\text{VII}.1$]{Reed1980aa}, we have
  \begin{equation}
    \label{eq-myeq}
    \norm{(\Itens-i\sqrt g L )^{-1}}\les \cos^{-1}(\tfrac 12\arg g).
  \end{equation}
  Applying to $L=\gls H$, remembering that $\lambda=\sqrt g$, the
  \lcnamecref{thm-lemmaresbounded} follows from the power series expansion
  \begin{equation*}
  \norm{(\Itens -U)^{-1}} = \norm{(\Itens-i\lambda H - D)^{-1}} \les \norm{J^{-1}}\sum_{q=0}^{\infty} \norm{D J^{-1}}^q ,
\end{equation*}
with $J\defi \Itens -i \lambda\gls H $. Indeed by \cref{eq-myeq},
  $\norm{J^{-1}}\les\cos^{-1}(\tfrac 12\arg g)$, and, by \cref{thm-Dren},
\begin{equation*}
\norm{D J^{-1}}\les \Oun \abs g
  \cos^{-1}(\tfrac 12\arg g) \les \Oun \rho.
\end{equation*}
Taking $\rho$ small enough, we can ensure $\norm{D J^{-1}}< 1/2$, hence
  $ \sum_{q=0}^{\infty} \norm{D J^{-1}}^q < 2$.
\end{proof}

\begin{lemma}[Contour translation]
  \label{thm-translation}
  For $g$ in the cardioid domain $\Card_\rho$, the  contour translation from $(\sigma_{c})_{n_c n_c}  $ to
$(\sigma_{c})_{n_c n_c}+B_1$ does not cross any singularity of 
$\Tr\log_2(\Itens  -i\lambda C^{1/2}\sigma C^{1/2})$, and the translation 
$(\sigma_{c})_{n_c n_c}  + B_2 $ does not cross any singularity of 
$\Tr\log_3(\Itens -i\lambda C^{1/2}\sigma C^{1/2} + D_1)$.
\end{lemma}
\begin{proof}
  To prove that
  $\Tr\log_2(\Itens -i\lambda C^{1/2}\sigma C^{1/2})$ is analytic in the band corresponding to
  $(\sigma_{c})_{n_c n_c} +B_1$ for the $(\sigma_{c})_{n_c n_c}$
  variables, one can write
  \begin{equation*}
\log_2 (1-x) = x - \int_0^1
  \frac{x}{1-tx} dt = - \int_0^1 \frac{t x^2 }{1-tx} dt
\end{equation*}
and then use the previous \namecref{thm-lemmaresbounded} to prove that, for $g$ in the small open
  cardioid domain $\Card_\rho$, the resolvent
  $R(t)\defi(\Itens -it\lambda C^{1/2}\sigma C^{1/2})^{-1}$, is
  also well-defined for any $t \in [0,1]$ by a power series uniformly convergent in the band considered.

  For the second translation, we use a similar argument, writing
  \begin{equation*}
  \log_3 (1-x) = x + \frac{x^2}{2} - \int_0^1 \frac{x}{1-tx} dt =
  \int_0^1 \frac{x^2(1 -2t -tx) ) }{2(1-tx)} dt.
\end{equation*}
\end{proof}

\subsection{Multiscale analysis}
\label{sec-multiscale-analysis}

The cutoff $[-N, N]^4$ of the previous section is not well adapted to the rotation invariant $\ntup^2$ term in the propagator,
nor very convenient for multi-slice analysis as in \cite{Gurau2014ab}. From now on we introduce other cutoffs, 
which are still sharp in the ``momentum space" $\ell_2 ({\mathbb Z})^4$, hence equivalent\footnote{The sup and square norm in our finite dimension four are equivalent.} 
to the previous ones, but do not longer factorize over colours\footnote{We could also use parametric cutoffs
as in \cite{Riv1,Ben-Geloun2011aa}, but sharp cutoffs are simpler.}. 

We fix an integer $M>1$ as ratio of a geometric progression $M^j$, where $j\in {\mathbb N}^*$ is the slice index
and define the ultraviolet cutoff as a maximal slice index $\jm$ so
that the previous $N$ roughly corresponds to $M^{\jm}$. More precisely, our notation convention is that $1_{x}$ is the characteristic function of the event $x$,
and we define  the following diagonal operators on $\Htens$: 
  \begin{align*}
    (\indic_{\les 1})_{\mtup\ntup} ={}& (\indic_{1}) _{\mtup\ntup} \defi 1_{1+ \norm{n}^{2}\les M^{2} }\delta_{\mtup\ntup},\\
    (\gls{cutofflesj})_{\mtup\ntup}\defi{}&  1_{1+\norm{n}^{2}\les M^{2j}}\delta_{\mtup\ntup}
    &&\text{for } j\ges 2, \\
    \gls{cutoffj}\defi{}& \indic_{\les j} - \indic_{\les j-1}
    &&\text{for }j\ges 2.
  \end{align*}
(Beware we choose the convention of \emph{lower} indices for slices, as in \cite{Gurau2014ab}, not upper
indices as in \cite{Riv1}.) We also write $C^{1/2}_{\les j}$ for $
\indic_{\les j} C^{1/2} $ and $C^{1/2}_{j}$ for $  \indic_{j} C^{1/2}
$. Since our cutoffs are sharp (projectors) we still have the natural relations
\begin{equation*}
(C^{1/2}_{\les j})^2  = C_{\les j}, \quad  (C^{1/2}_{j} )^2 = C_j  .
\end{equation*}

We start with the $ (\sigma , \tau )$ functional integral \eqref{startingpoint}
which we have reached in the previous section, and organize it
according to the new cutoffs, so that the previous limit $N \to
\infty$ becomes a limit $\jm \to \infty$. The interaction with cutoff $j$ is obtained by cutting the propagators in the loop vertices. Remark 
that we do not need to introduce cutoffs on the propagators hidden in $\Ar_{\cM_{1}}$ or $\Ar_{\cM_{2}}$, as these are
convergent integrals anyway. It means we define the cutoff version of the quantities introduced in the previous subsection as
\begin{subequations}
\begin{gather}
V_{\les j} (\sigma, \tau)  \defi   V^{\ges 3}_{\les j} (\sigma)  + V^{\les 2}_{\les j} (\sigma, \tau),\label{eq-Vlesj-def}\\
V^{\ges 3}_{\les j} (\sigma)  \defi  \Tr\log_{3}(\Itens-U_{\les j}
) +\Tr[D_{1,\les j}\Sigma_{\les j}^{2}]+\lastvac_{\les j},\\
    \gls{lastvac}[_{\les j}]\defi\Tr(\tfrac 13 D_{1,\les j}^{3}+D_{1,\les j}^{2}D_{2,\les j}+\tfrac
14D_{1,\les j}^{4}),\label{smallertj2}\\
V^{\les 2}_{\les j} \defi \tfrac{\lambda^2}{2}\wo{\sigmad\scalprod
  Q_{\les j} \sigmad} - i \tfrac{\lambda^{2}}{\sqrt 2}  Q_{0, \les j} \scalprod \taud - \Tr[D_{2, \les j}\Sigma_{\les j}]\label{smallertj1}, \\
 Q_{\les j} = Q_{0, \les j}  + Q_{1, \les j},\\
    \cR_{\les j}\defi\frac{1}{\Itens - U_{\les j} },\qquad   U_{\les j} \defi \Sigma_{\les j} + D_{\les j},\qquad\Sigma_{\les j}\defi  i \lambda C^{1/2}_{\les j}  \sigma C^{1/2}_{\les j},  \label{smallertj3}\\
D_{1, \les j} \defi i\lambda C^{1/2}_{\les j} B_1 C^{1/2}_{\les j},
\qquad D_{2, \les j} \defi i\lambda C^{1/2}_{\les j} B_2
C^{1/2}_{\les j},\qquad D_{\les j}\defi D_{1, \les j} + D_{2, \les j}.\label{smallertj4}
\end{gather}
\end{subequations}
The functional integral \eqref{startingpoint} with cutoff $\jm$ is then defined as 
\begin{equation*}
\cW_{\les\jm}(g)\defi \log \cZ_{\les\jm}(g) ,\qquad\cZ_{\les\jm}(g)
\defi\int d\nu_{\gls{Idirect}}(\sigmad , \tau )\,e^{-V_{\les\jm}}.
\end{equation*}
Defining $V_{\les 0}\defi 0$ and, for all $1\les j\les\jm$,
$V_{j}\defi V_{\les j}-V_{\les j-1}$, we note that
$V_{\les\jm}=\sum_{j=1}^{\jm}V_{j}$ so that
\begin{equation}
  \label{factoredintera}
  \cZ_{\les\jm}(g)=\int d\nu_{\gls{Idirect}}(\sigmad,\taud)\,\prod_{j=1}^{\jm}e^{-V_{j}}.
\end{equation}

To define the specific part of the interaction which should be attributed to scale $j$ we introduce
\begin{equation*}
\indic_{\les j}(t_j) = \indic_{\les j-1}   + t_j \indic_{j}
\end{equation*}
where $t_j \in [0,1]$ is an interpolation parameter for the $j$-th
scale. Remark that 
\begin{equation*}
\indic^2_{\les j}(t_j) =  \indic_{\les j-1}   + t^2_j \indic_{j}.
\end{equation*}
The interpolated interaction and resolvents are defined as $V_{\les j} (t_j)$, $\Sigma_{\les j} (t_j)$, $D_{\les j} (t_j)$, $\cR_{\les j} (t_j) $ and so on by \crefrange{eq-Vlesj-def}{smallertj4} in which we substitute $\indic_{\les j}(t_j)$ for $\indic_{\les j}$.
When the context is clear, we write simply $V_{\les j} $ for $V_{\les
  j} (t_j) $, $U_{\les j} $ for $U_{\les j} (t_j) $, $\gls{Uprime}$ for
$\frac{d}{dt_j}  U_{\les j}  $ and so on. In these notations we have
\begin{equation}
  \lb\begin{aligned}
    V_{j} &=  V_j^{\ges 3}  + V_j^{\les 2},  \\
    V_j^{\ges 3} &= \lastvac_{j}+\int_0^1  dt_{j}\,\Tr\bigl[U'_{j}(\Itens+U_{\les j}-\gls{cR}[_{\les j}])+\gls{Dprimeun}\Sigma^{2}+D_{1,\les
      j}\gls{Sigmaprime}\Sigma+D_{1,\les j}\Sigma\Sigma'_{j} \bigr],\\
    V_j^{\les 2} &= \tfrac{\lambda^2}{2}\wo{\sigmad\scalprod(Q_{0,j} + Q_{1,j})
      \sigmad}-i  \tfrac{\lambda^{2}}{\sqrt 2}  Q_{0, j} \scalprod \taud -
    3\int_0^1  dt_{j}\,\Tr\bigl[ \gls{Dprimedeux}\Sigma_{\les j}\bigr], \\
    \lastvac_{j}&=\lastvac_{\les j}-\lastvac_{\les j-1},\quad Q_{1,j} = Q_{1,\les j} - Q_{1,\les j-1} , \quad Q_{0, j} =
    Q_{0,\les j} - Q_{0,\les j-1}.
  \end{aligned}\right. \label{eq-nicevj}
\end{equation}

Finally, as in \cite{Gurau2014ab}, we define
\begin{equation*}
\gls{W}[_j](\sigma,\tau) \defi e^{-V_j} -1
\end{equation*}
and encode the factorization of the interaction in \eqref{factoredintera} through Grassmann numbers as
\begin{equation*}
  \cZ_{\les\jm}(g) =\int d\nu_{\gls{Idirect}}(\sigmad , \taud )   \, \Bigl(
  \prod_{j = 1}^{\jm} d\mu (\bar \chi_j , \chi_j) \Bigr) e^{ - \sum_{j = 1}^{\jm}   \bar \chi_j  W_j(\sigma,\tau)   \chi_j },
\end{equation*}
where $d \mu(\bar \chi ,\chi ) = d\bar \chi d\chi \, e^{-\bar \chi \chi}$ is the standard normalized Grassmann Gaussian measure with covariance $1$.

\section{The Multiscale Loop Vertex Expansion}
\label{sec-mult-loop-vert-exp}

We perform now the two-level jungle expansion of \cite{Abdesselam1995aa,Gurau2014ab,Delepouve2014aa}. This section is almost 
identical to those of \cite{Gurau2014ab,Delepouve2014aa}, as it was
precisely the goal of \cite{Gurau2014ab} to create a combinatorial
constructive ``black box'' to automatically compute and control the
logarithm of a functional integral of the type of $\cZ_{N}$. Nevertheless
we reproduce the section here, in abridged form, since the MLVE
technique is still relatively recent and since
there is a slight change compared to the standard version. Indeed we have now two sets of Bosonic fields, the 
main $\sigma$ field and the auxiliary $\tau$ field, and the $\tau$ field requires slightly 
different interpolation parameters, namely $w^2$ instead of $w$ parameters.
\\

\noindent
Considering the set of scales $\gls{S}\defi \lnat 1,\jm\rnat$, we
denote $\gls{Iscale}$ the $\abs{S}$ by $\abs{S}$
identity matrix. The product Gaussian measure on the $\chi_{i}$'s and
$\bar\chi_{i}$'s can then be recast into the following form:
\begin{equation*}
  \prod_{j = 1}^{\jm} d\mu (\bar \chi_j , \chi_j)=d\mu_{\Iscale_{S}}(\chibtup,\chitup),\qquad\chitup\defi(\chi_{i})_{1\les i\les\jm},\,\chibtup\defi(\bar\chi_{i})_{1\les i\les\jm}
\end{equation*}
so that the partition function rewrites as
\begin{equation*}
\cZ_{\les\jm}(g)  =  \int d\nu_\cS \; e^{- W},
\quad d\nu_{\cS} \defi d\nu_{\gls{Idirect}}(\sigmad , \taud )   \,
d\mu_{\Iscale_{S}}(\chibtup,\chitup),
 \quad W = \sum_{j =1}^{\jm}\bar \chi_j  W_j   (\sigmad , \taud )   \chi_j.
\end{equation*}
The first step expands to infinity the exponential of the interaction:
\begin{equation*}
\cZ_{\les\jm}(g)  = \sum_{n=0}^\infty \frac{1}{n!}\int d\nu_{\cS}\,(-W)^n .
\end{equation*}
The second step introduces Bosonic replicas for all the \emph{nodes}\footnote{We use the new word 
``node'' rather than ``vertex'' for the $W$ factors, in order not to
confuse them with the ordinary vertices of the initial perturbative
expansion, nor with the loop vertices of the intermediate field
expansion, which are not equipped with Fermonic fields.} in $\gls{nset}
\defi\lnat 1,n\rnat$: 
\begin{equation*}
\cZ_{\les\jm}(g)= \sum_{n=0}^\infty \frac{1}{n!}\int d\nu_{\gls S,\gls{nset}} \,  \prod_{a=1}^n  (-W_a),
\end{equation*}
so that each node $W_a =  \sum_{j =0}^{\jm} \bar \chi_j^a W_j   (\sigmad^a,\taud^{a})\chi_j^a $ has now its own 
set of Bosonic matrix fields $\sigmad^a = \bigl((\sigma^1)^a,
(\sigma^2)^a, (\sigma^3)^a, (\sigma^4)^a\bigr)$ and $\taud^a = \bigl((\tau^1)^a,
(\tau^2)^a, (\tau^3)^a, (\tau^4)^a\bigr)$, and its own Fermionic
replicas  $ (\bar \chi_j^a, \chi_j^a)$. The sequence of Bosonic
replicas $(\sigmad^{a}; \taud^{a})_{a\in\nset}$ will be denoted by $(\rsigmad; \rtaud)$ and
belongs to the product space  for the $\sigma$ and $\tau$ fields (which is also a direct sum)
\begin{equation*}
  \bV_{\nset}\defi [\Hopdirect\otimes\R^{n}] \times [\Hopdirect\otimes\R^{n}] = [\Hopdirect \oplus \Hopdirect ] \otimes\R^{n}.
\end{equation*}
The replicated \emph{normalised} measure is completely degenerate between replicas (each of the four colours remaining independent of the others):
\begin{equation*}
d\nu_{\gls S,\gls{nset}} \defi d\nu_{\Idirect\otimes\gls{One}[_{\gls{nset}}]} (\rsigmad, \rtaud) \, d\mu_{\Iscale_{S}\otimes\gls{One}[_{\gls{nset}}]} (\chibtup ,\chitup)
\end{equation*}
where $\bbbone$ means the ``full'' matrix with all entries equal to
$1$.\\

\noindent
The obstacle to factorize the functional integral $\cZ$ over nodes and
to compute $\log\cZ$ lies in the degenerate blocks
$\gls{One}[_{\gls{nset}}]$ of both the Bosonic and Fermionic covariances. In order to remove this obstacle 
we simply apply the $2$-level jungle Taylor formula of \cite{Abdesselam1995aa} with priority to Bosonic links
 over Fermionic links.
However beware that since the $\tau$ field counts for two $\sigma$ fields, we have to introduce 
the parameters $w$ differently in $\sigma$ and $\tau$ namely we interpolate off-diagonal covariances between vertices $a$ and $b \ne a$
with ordinary parameters $w$ for the $\sigma$ covariance but with parameters $w^2$ for the $\tau$ covariance. Indeed with this precise prescription 
a sigma tree link $(a,b)$ of type $\sigmad\scalprod Q_{0,j_a} Q_{0,j_b} \sigmad$ term will be exactly Wick-ordered with respect to the
interpolated $d \nu ( \sigmad) $ measure by the associated tau link $\ell = (a,b)$, see \cref{sec-comp-boson-integr}. In other words the $\kN_{2}$
graph when it occurs as such a link, is exactly renormalized.

 It means that a first Taylor forest formula is applied to 
$\gls{One}[_{\gls{nset}}]$ in $d\nu_{\gls{Idirect}\otimes
  \gls{One}[_{\gls{nset}}]}(\rsigmad, \rtaud)$, with weakening parameters $w$ for the $\sigma$ covariance and 
parameters $w^2$ for the $\sigma$ covariance. The forest formula simply interpolates iteratively off-diagonal covariances
between 0 and 1. The prescription described is legitimate since when $w$ monotonically parametrizes the $[0,1]$ interval,
$w^2$ also parametrizes the $[0,1]$ interval monotonically; hence a Taylor formula can be written just as well as
$F(1) = F(0) + \int_0^1 F'(x)dx$ or as $F(1) = F(0) + \int_0^12x F'(x^2) dx$.

It is then followed by a second
Taylor forest formula of $\gls{One}[_{\gls{nset}}]$ in
$d\mu_{\Iscale_{S}\otimes \gls{One}[_{\gls{nset}}]}
(\chibtup ,\chitup)$, decoupling the connected components $\gls{cB}$ of the first forest.

The definition of $m$-level jungle formulas and their equivalence to $m$ successive forests formulas is given in \cite{Abdesselam1995aa}; the application (with $m=2$) to the current context is described in detail  in \cite{Gurau2014ab,Delepouve2014aa}, so we shall not repeat it here.\\

The 2-jungle Taylor formula rewrites our partition function as:
\begin{equation}
  \cZ_{\les\jm}(g)= \sum_{n=0}^\infty \frac{1}{n!}  \sum_{\cJ}\,\sum_{j_1=1}^{\jm} 
  \dotsm\sum_{j_n=1}^{\jm}
  \,\int d\tuple{w_{\!\cJ}} \int d\nu_{ \ccJ}  
  \,\partial_{\ccJ}   \Bigl[ \prod_{\cB} \prod_{a\in \cB}   \bigl(   -\bar \chi^{\cB}_{j_a}  W_{j_a}   (\sigmad^a ,\taud^a )  
  \chi^{ \cB }_{j_a}   \bigr)  \Bigr],
  \label{eq-ZafterJungle}
\end{equation}
where
\begin{itemize}
\item the sum over $\cJ$ runs over all $2$-level jungles, hence over
  all ordered pairs $\cJ = (\cF_B, \cF_F)$ of two (each possibly
  empty) disjoint forests on $\gls{nset}$, such that $\cF_B$ is a
  (Bosonic) forest, $\cF_F$ is a (Fermonic) forest and
  $\bar \cJ = \cF_B \cup \cF_F $ is still a forest on
  $\gls{nset}$. The forests $\cF_B$ and $\cF_F$ are the Bosonic and
  Fermionic components of $\cJ$. Fermionic edges
  $\ell_F \in E(\cF_F)$ carry a scale data $j$.
 
\item $\int d\tuple{w_{\ccJ}}$ means integration from 0 to 1 over parameters
  $w_\ell$, one for each edge $\ell \in E(\bar\cJ)$, namely
  $\int d\tuple{w_{\ccJ}} = \prod_{\ell\in E(\bar \cJ)} \int_0^1 dw_\ell $.  There
  is no integration for the empty forest since by convention an empty
  product is 1. A generic integration point $\tuple{w_{\ccJ}}$ is therefore made
  of $m(\bar \cJ)$ parameters $w_\ell \in [0,1]$, one for
  each $\ell \in E(\bar \cJ)$.

\item In any $\cJ=(\cF_{B},\cF_{F})$, each block $\cB$ corresponds to
  a tree $\cT_{\cB}$ of $\cF_{B}$.
  \begin{subequations}\label{eq-partialJ}
    \begin{align}
      \partial_{\ccJ}\defi{}&\partial_{F}\partial_{B},\qquad\partial_{B}\defi\prod_{\cB\in\cF_{B}}\partial_{\cT_{\cB}},\label{eq-partialJ-product}\\
      \partial_{F}\defi{}&\prod_{\substack{\ell_F \in
          E(\cF_F),\\\ell_F=(d,e)}} \delta_{j_{d } j_{e } } \Bigl(
      \frac{\partial}{\partial \bar \chi^{\cB(d)}_{j_{d} }
      }\frac{\partial}{\partial \chi^{\cB(e)}_{j_{e} } }+
      \frac{\partial}{\partial \bar \chi^{ \cB( e) }_{j_{e} } }
      \frac{\partial}{\partial \chi^{\cB(d)
        }_{j_{d} } } \Bigr),\\
      \partial_{\cT_{\cB}}\defi{}&\prod_{\substack{\ell_B \in
          E(\cT_{\cB}),\\\ell_B=(a,b)}} \Bigl[\sum_{c=1}^4
      \sum_{\substack{m,n}}\Bigl( \frac{\partial}{\partial
        (\sigma^{c}_{mn})^a} \frac{\partial}{\partial (\sigma^{c}_{mn})^b} + 2 w_\ell \frac{\partial}{\partial
        (\tau^{c}_{mn})^a}  \frac{\partial}{\partial (\tau^{c}_{mn})^b}  \Bigr)\Bigr] \label{eq-partialJ-TB}
    \end{align}
  \end{subequations}
where $ \cB(d)$ denotes the Bosonic
  block to which the node $d$ belongs. Remark the factor $2w_\ell$ in \eqref{eq-partialJ-TB} corresponding to the 
use of $w^2$ parameters for $\tau$.

\item The measure $d\nu_{\ccJ}$ has covariance
  $\gls{Idirect} \otimes X (\tuple{w_{B}}) $ on Bosonic variables $\sigma$, covariance 
  $\gls{Idirect} \otimes X^{\circ 2} (\tuple{w_{B}}) $ on Bosonic variables $\tau$
 and
  $\Iscale_{S} \otimes Y (\tuple{w_{F}})$ on Fermionic variables, hence
  \begin{multline*}
    \int d\nu_{\ccJ}\, F = \biggl[e^{\frac{1}{2} \sum_{a,b=1}^n
       \sum_{c=1}^4 \sum_{m,n}\bigl( X_{ab}(\tuple{w_{B}})
      \frac{\partial}{\partial (\sigma_{mn}^c)^a}\frac{\partial}{\partial (\sigma_{mn}^{c})^b}  +  X^{\circ 2}_{ab}(\tuple{w_{B}})
  \frac{\partial}{\partial (\tau_{mn}^c)^a}\frac{\partial}{\partial (\tau_{mn}^{c})^b} \bigr)     }  \\
   e^{ \sum_{\cB,\cB'} Y_{\cB\cB'}(\tuple{w_{F}})\sum_{a\in \cB, b\in
        \cB' } \delta_{j_aj_b} \frac{\partial}{\partial \bar
        \chi_{j_a}^{\cB} } \frac{\partial}{\partial \chi_{j_b}^{\cB'}
      } } F \biggr]_{\sigma = \tau = \bar\chi =\chi =0}
  \end{multline*}
where $X^{\circ 2}$ means the Hadamard square of the matrix, hence the matrix whose elements are 
  the squares of the matrix elements of $X$, \emph{not} the square in the ordinary matrix product sense.

\item $\gls{X}_{ab} (\tuple{w_{B}})$ is the infimum of the $w_{\ell_B}$
  parameters for all the Bosonic edges $\ell_B$ in the unique path
  $P^{\cF_B}_{a \to b}$ from node $a$ to node $b$ in $\cF_B$. The
  infimum is set to zero if such a path does not exist and to $1$ if
  $a=b$.

\item $Y_{\cB\cB'}(\tuple{w_{F}})$ is the infimum of the $w_{\ell_F}$
  parameters for all the Fermionic edges $\ell_F$ in any of the paths
  $P^{\cF_B \cup \cF_F}_{a\to b}$ from some node $a\in \cB$ to some
  node $b\in \cB'$.  The infimum is set to $0$ if there are no such
  paths, and to $1$ if $\cB=\cB'$ (i.e.\@ if such paths exist but do not contain any
  Fermionic edges).
\end{itemize}
Remember that a main property of the forest formula is that the
symmetric $n$ by $n$ matrices $X_{ab}(\tuple{w_{B}})$ or $X^{\circ 2}_{ab}(\tuple{w_{B}})$ 
are positive for any value of $\tuple{w_{\ccJ}}$, hence the Gaussian measure $d\nu_{\ccJ} $ is well-defined. The matrix $Y_{\cB\cB'}(\tuple{w_{F}})$
is also positive.\\

\noindent
Since the slice assignments, the fields, the measure and the integrand are now 
factorized over the connected components of $\bar \cJ$, the logarithm of $\cZ$ is easily computed as exactly the same sum but restricted 
to $2$-level spanning trees:
\begin{multline}
  \label{eq-treerep}
  \cW_{\les\jm}(g)= \log \cZ_{\les\jm}(g)=
  \sum_{n=1}^\infty \frac{1}{n!}  \sum_{\cJ\text{ tree}}
  \,\sum_{j_1=1}^{\jm} 
  \dotsm\sum_{j_n=1}^{\jm}\\
  \int d\tuple{w_{\ccJ}} \int d\nu_{ \ccJ}  
  \,\partial_{\ccJ}   \Bigl[ \prod_{\cB} \prod_{a\in \cB}   \bigl(   -\bar \chi^{\cB}_{j_a}  W_{j_a}   (\sigmad^a , \taud^a )  
  \chi^{ \cB }_{j_a}   \bigr)  \Bigr]
\end{multline}
where the sum is the same but conditioned on $\bar \cJ = \cF_B \cup \cF_F$ being a \emph{spanning tree} on $\gls{nset}$.\\

Our main result is similar to the one of \cite{Delepouve2014aa} in the more convergent three dimensional case:
\begin{thm} \label{thetheorem}
  Fix  $\rho >0$ small enough. The series \eqref{eq-treerep} is
  absolutely and uniformly in $\jm$ convergent for $g$ in the small
  open cardioid domain $\Card_\rho$ (defined by $\abs{\arg g} <\pi$ and
  $\abs{g} < \rho \cos^{2}(\tfrac 12\arg g)$, see \cref{cardio}). Its
  ultraviolet limit $\cW_\infty (g) \defi \lim_{\jm  \to \infty}  \log
  \cZ_{\les\jm}(g)$ is therefore well-defined and analytic in that
  cardioid domain; furthermore it is the Borel sum of its perturbative series in powers of $g$. 
\end{thm}
\begin{figure}[!htp]
\begin{center}
  {\includegraphics[width=0.2\textwidth]{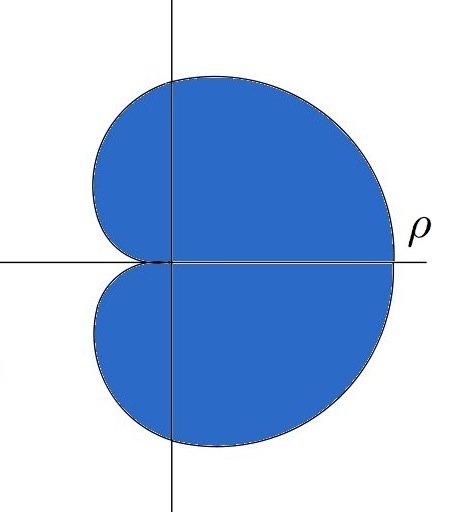}}
   \end{center}
  \caption{A Cardioid Domain}
  \label{cardio}
\end{figure}
The rest of the paper is devoted to the proof of this \namecref{thetheorem}.

\section{Block Bosonic integrals}
\label{sec-comp-boson-integr}

Since the Bosonic functional integral factorizes over the Bosonic blocks, it is sufficient to
compute and bound 
the Bosonic functional integrals over a fixed block $\cB$.

\subsection{The single node case}

Let us consider first the simple case in which the Bosonic block $\cB$ is reduced to a single node $a$.
We have then a relatively simple contribution 
\begin{equation*}
\int d \nu_{\Idirect}(\sigmad^{a},\taud^{a})\, W_{j_{a}}  =  \int d\nu_{\Idirect}\,  \bigl( e^{- V_{j_a}} - 1\bigr) = \int_0^1 dt  
\int d \nu_{\Idirect}\,   e^{- tV_{j_a}} (-V_{j_a}).
\end{equation*}

We consider in the term $-V_{j_a}$ down from the exponential two particular pieces of $V^{\les 2}_{j_a}$,
namely the terms $-\tfrac{\lambda^2}{2}\wo{\sigmad\scalprod Q_{0,j_a}
  \sigmad}$ and $ i  \tfrac{\lambda^{2}}{\sqrt 2}  Q_{0, j_a}
\scalprod \taud $. In the first one, we integrate by parts one of its two $\sigma$ fields, obtaining 
$t\tfrac{\lambda^{4}}{2} \sigmad\scalprod Q^2_{0,j_a} \sigmad$ plus (perturbatively convergent) terms
\begin{equation*}
PC_{j_a}(\sigmad)= t\tfrac{\lambda^{2}}{2} \sigmad\scalprod Q_{0,j_a} \frac{\partial}{\partial \sigmad}  
\Bigl(\tfrac{\lambda^2}{2} \sigmad\scalprod Q_{1,j_a} \sigmad  +
3\int_0^1  dt_{j}\,\Tr\bigl[ D'_{2,j}\Sigma_{\les j}\bigr] +    V^{\ges 3}_{j_a}  (\sigma) \Bigr).
\end{equation*}
We also integrate by parts the $\tau$ term in $V_{j_a}^{\les 2}$ and remark that it 
gives  $-t\tfrac{\lambda^{4}}{2}\Tr[Q_{0,j_{a}}^{2}]$, hence exactly
Wick-orders the previous $ \sigmad\scalprod Q^2_{0,j_a} \sigmad$
term. Finally we integrate out the $\tau$ field, which gives back the $t^2\delta_{\kN_2, a}$ counterterm. Hence altoghether we have proven:
\begin{lemma}\label{thm-SingleNodeComputation}
  The result of this computation is
  \begin{equation*}
    \int d \nu_{\Idirect}(\sigmad,\taud)\, W_{j_{a}}(\sigmad,\taud)  = -    \int_0^1 dt\,  e^{t^2\delta_{\kN_2, a}} 
    \int d\nu_{\Idirect}(\sigmad)\,   e^{- tV_{j_a} (\sigmad)}
    \woo{V_{j_a}(\sigmad)}
  \end{equation*}
  where
  \begin{equation*}
    \woo{V_{j_a}(\sigmad )}  \defi V^{\ges 3}_{j_a}  (\sigma)   - \tfrac{t\lambda^4}{2} \wo{\sigmad\scalprod Q^2_{0,j_a} \sigmad}
    - PC_{j_a}(\sigma)  + \tfrac{\lambda^2}{2}\wo{\sigmad\scalprod Q_{1,j} \sigmad}  -3\int_0^1  dt_{j}\,\Tr\bigl[ D'_{2,j}\Sigma_{\les j}\bigr].
  \end{equation*}
\end{lemma}
This Lemma will be sufficient to bound the single node contribution by
$\Oun M^{-\Oun j_a}$, see next sections.

In order to treat the single node case and the cases of Bosonic blocks
with more than one node in a unified manner, it is convenient to regard
$\woo{V_{j_a}(\sigmad )}$ as a sum of (Wick-ordered) skeleton graph amplitudes, see
\cref{def-skeletons}. These Feynman graphs are one-vertex maps except
those which correspond to the terms in $PC_{j_a}(\sigma)$ which are
trees with only one edge. Therefore we will write
\begin{equation*}
  \woo{V_{j_a}(\sigmad )}\fide\sum_{\skG}\wo{A_{\skG}(\sigmad)}.
\end{equation*}

\subsection{Blocks with more than one node}

In a Bosonic block with two or more nodes, 
the Bosonic forest $\cF_B$ is a non-empty Bosonic tree $\cT_{\cB}$. Consider a fixed such block $\cB$, a fixed tree $\cT_{\cB}$
and the fixed set of  frequencies $\{j _a\}$, $a \in \cB$, \emph{all distinct}. 
We shall write simply $d\nu_\cB$ for $d\nu_{\cT_{\cB}} (\sigmad, \taud)$.
The corresponding covariance  
\label{page-Interpol-Cov} of the Gaussian measure $d\nu_\cB$ is also a symmetric matrix on the vector space $\bV_\cB$,
whose vectors, in addition to the colour and double momentum components and their type $\sigma$ or $\tau$ have also a node 
index $a \in \cB$; hence $\bV_{\cB}= \R^{\abs{\cB}} \otimes  \bigl[ \gls{Hopdirect}  \oplus  \gls{Hopdirect} \bigr]$.
It can be written as $\bX_\cB \defi \gls{Idirect} \otimes  [ X (\tuple{w_{\cB}}) +  X^{\circ 2} (\tuple{w_{\cB}})]$ where $X$ acts
on the $\sigma$ part hence on the first factor in  $\bigl[ \gls{Hopdirect}  \oplus  \gls{Hopdirect} \bigr]$ and $X^{\circ 2}$
on the $\tau$ part hence on the second factor in  $\bigl[ \gls{Hopdirect}  \oplus  \gls{Hopdirect} \bigr]$.\\

\subsubsection{From trees to forests}
\label{sec-from-trees-effective-forests}

We want to compute
\begin{equation*}
I_{\cB}\defi\int d\nu_{\cB}\,\partial_{\cT_{\cB}}\prod_{a\in \cB}  ( e^{-V_{j_a}} -1  ) (\sigmad^a, \taud^a).
\end{equation*}
When $\cB$ has more than one node,
since $\cT_{\cB}$ is a tree, each node $a \in \cB$ is touched by at least one derivative and we can replace 
$W_{j_a} =e^{- V_{j_a}} -1$ by $ e^{- V_{j_a}}$ (the derivative of 1
giving 0). The partial derivative $\partial_{\cT_{\cB}}$ can be
rewritten as follows:
\begin{equation*}
\partial_{\cT_{\cB}}=\Bigl(\prod_{\substack{\ell\in E(\cT_{\cB}),\\  \ell=(a,b)}}  
 \sum_{c_{\ell}=1}^4  \sum_{\substack{m_{\ell},n_{\ell}}}\Bigr)
 \prod_{a\in\cB}\prod_{s\in S_{\cB}^{a}} (\partial_{\sigma_{s}} + \partial_{\tau_{s}})
\end{equation*}
where $S_{\cB}^{a}$ is the set of edges of $\cT_{\cB}$ which ends at
$a$, and
\begin{equation*}
\partial_{\sigma_{s}}\defi\frac{\partial}{\partial(\sigma^{c_{s}}_{m_{s} n_{s} })^{a}},\qquad 
\partial_{\tau_{s}}\defi\frac{\partial}{\partial(\tau^{c_{s}}_{m_{s} n_{s} })^{a}}.
\end{equation*}
We thus have to compute
\begin{equation*}
  I_{\cB}=\int d\nu_{\cB}\, \prod_{\substack{\ell\in
      E(\cT_{\cB}),\\\ell=(a,b)}}
  \sum_{c_{\ell}=1}^4  \sum_{\substack{m_{\ell},n_{\ell}}}
  F_\cB,\qquad      F_\cB \defi{}  \prod_{a\in \cB}\bigl[\prod_{s\in S_{\cB}^{a}} (\partial_{\sigma_{s}} + \partial_{\tau_{s}})
  e^{- V_{j_a}}\bigr].
\end{equation*}
We can evaluate the derivatives in the preceding equation through the Fa\`a di Bruno formula:
\begin{equation*}
\prod_{s\in S}  [\partial_{\sigma_{s}} + \partial_{\tau_{s}}] f\bigl( g( \sigma , \tau) \bigr) =   \sum_{\pi } f^{\abs{\pi}}\!\bigl( g( \sigma , \tau) \bigr) \prod_{b\in \pi}  \Bigl(\bigl( \prod_{s\in b} [\partial_{\sigma_{s}} + \partial_{\tau_{s}}]\bigr) g (\sigma, \tau)\Bigr),
\end{equation*}
where $\pi$ runs over the partitions of the set $S$, $b$ runs through
the blocks of the partition $\pi$, and $\abs{\pi}$ denotes the number
of blocks of $\pi$. In our case $f$, the exponential
function, is its own derivative, hence the formula simplifies to
\begin{equation}
F_\cB=   \prod_{a\in \cB}   e^{- V_{j_a}}   \biggl( \sum_{\pi^a} \prod_{b^a\in \pi^a} \;    \Bigl[\bigl[\prod_{s\in b^a} 
(\partial_{\sigma_{s}} + \partial_{\tau_{s}})\bigr]  (-V_{j_a}) \Bigr]  \biggr),
\label{eq-partitio}
\end{equation}
where $\pi^a$ runs over partitions of $S^a_\cB$ into blocks $b^a$. The Bosonic integral in a block $\cB$ can be written therefore in a simplified manner as:
\begin{equation}
\label{eq-bosogauss}
I_{\cB} = \sum_\skG \int d \nu_\cB\, \Bigl(\prod_{a\in \cB}  e^{-V_{j_a} (\sigmads^{a} , \tauds^{a} )}\Bigr)  A_\skG(\sigmad),
\end{equation}
where we gather the result of the derivatives as a sum over graphs $\skG$
of corresponding amplitudes $A_\skG(\sigmad)$. Indeed, the dependence of $V_j$ being linear in $\tau$, the 
corresponding $\tau$ derivatives are constant, hence amplitudes  
$A_\skG(\sigmad)$ do not depend on $\taud$. The graphs $\skG$ will be
called \firstdef{skeleton graphs}, see \cref{def-skeletons}. They are still
forests, with loop vertices\footnote{We recall that loop
  vertices are the traces obtained by $\sigma$ derivatives acting on
  the intermediate field action \cite{Rivasseau2007aa}.}, one for each
$b^a \in \pi^a, a \in \cB$. We now detail the different types of those
four-stranded loop vertices.\\

To this aim, let us actually compute
$\partial_{b^{a}}\defi\bigl[\prod_{s\in b^a} (\partial_{\sigma_s}+\partial_{\tau_{s}})\bigr]  (-V_{j_a})$, part of
\cref{eq-partitio}. First of all, remark that as $V_{j}$ is linear in
$\tau$ and $\partial_{\tau_{s}}V_{j}$ is independent of $\sigma$, if
$\card{b^{a}}\ges 2$, $\partial_{b^{a}}=\bigl[\prod_{s\in
  b^a}\partial_{\sigma_s}\bigr]  (-V_{j_a})$. Then, we rewrite \cref{eq-nicevj} using
$\Itens + U - \cR = -U^2 \cR$:
\begin{multline}
  V_{j} = \lastvac_{j}+\tfrac{\lambda^2}{2}\wo{\sigmad\scalprod
    Q_{j}\sigmad} -i  \tfrac{\lambda^{2}}{\sqrt 2}  Q_{0, j} \scalprod
  \taud +\int_0^1 dt_{j}\,\Tr\bigl[-U'_{j}U^{2}_{\les
    j}\gls{cR}[_{\les j}]\bigr.\\
  \bigl.+ D'_{1,j}\Sigma_{\les
    j}^2   + D_{1,\les j}(\Sigma'_{j}\Sigma_{\les j} + \Sigma_{\les
    j}\Sigma'_{j} ) - 3D'_{2,j}\Sigma_{\les j}\bigr].  \label{eq-nicevj1}
\end{multline}
Remembering that $\partial_{\sigma_s}$ and $\partial_{\tau_{s}}$ stand for derivatives with well
defined colour and matrix elements, we introduce the notations
  \begin{align*}
    \gls{dU}[_{\!\!\les j}]&\defi \frac{\partial U_{\les j}}{\partial \sigma^{c_{s}}_{m_{s}n_{s}}} =\frac{\partial \Sigma_{\les j}}{\partial \sigma^{c_{s}}_{m_{s}n_{s}}}=
    i\lambda C_{\les j}^{1/2} \delta^s C_{\les
      j}^{1/2},\\
    \dU{j}&\defi \frac{\partial U'_{j}}{\partial
      \sigma^{c_{s}}_{m_{s}n_{s}}}=\frac{\partial \Sigma'_{j}}{\partial
      \sigma^{c_{s}}_{m_{s}n_{s}}}= i\lambda (C_{j}^{1/2} \delta^s
    C_{\les j}^{1/2} + C_{\les j}^{1/2} \delta^s C_{j}^{1/2})
  \end{align*}
where $\delta^s$, defined as
 $(\delta^s)_{mn}\defi \frac{\partial\sigma}{\partial
   \sigma^{c_{s}}_{m_{s}n_{s}}}=\frac{\partial\tau}{\partial
   \tau^{c_{s}}_{m_{s}n_{s}}}$, equals
 $\be_{m_{s}n_{s}}\otimes\Itens_{\hat c_{s}}$ where
 $\be_{m_{s}n_{s}}$ has zero entries everywhere except at position $m_{s}n_{s}$ where it has entry one.\\

As noticed above, only one $\tau$ derivative needs to be applied to
$-V_{j}$:
\begin{equation*}
\partial_{\tau_{s}}(-V_{j})=i\tfrac{\lambda^{2}}{\sqrt
  2}\Tr_{c_{s}}[(Q_{0,j})_{c_{s}c_{s}}\be_{m_{s}n_{s}}].
\end{equation*}
We now concentrate on
the $\sigma$ derivatives. Since $\partial_{\sigma_s} \cR_{\lj} =
\gls{cR}[_{\lj}] \dU{\lj}\gls{cR}[_{\lj}]$, we get
\begin{multline}
  \partial_{\sigma_s} ( -V_j ) =
  -\lambda^{2}\Tr_{c_{s}}[\be_{m_{s}n_{s}}(Q_{j}\sigmad)_{c_{s}}]\\
  +\int_0^1 dt_j\, \Tr \bigl[ \dU{_j} U^2_{\lj} \gls{cR}[_{\lj}] + U'_{j}
  \dU{\lj} U_{\lj} \cR_{\lj}
  + U'_{j}U_{\lj} \dU{\lj}\cR_{\lj}+ U'_{j} U^2_{\lj} \cR_{\lj} \dU{\lj} \cR_{\lj}\\
  -D'_{1,j}(\dU{\lj}\Sigma_{\lj}+\Sigma_{\lj}\dU{\lj})
  -D_{1,\lj}(\dU{j}\Sigma_{\lj}+\Sigma'_{j}\dU{\lj}+\dU{\lj}\Sigma'_{j}+\Sigma_{\lj}\dU{j})
  +3D'_{2,j}\dU{\lj}\bigr].\label{eq-DerivSigmak1}
\end{multline}
In this formula notice the first term which is the $\sigma$ derivative
of $\wo{\sigmad\scalprod Q_{j}\sigmad}$, the sum of the next four terms, depending on whether
$\partial_{\sigma_s}$ acts on $\cR$ or on one of the three
explicit $U$-like numerators, and also the seven simpler terms with
explicit $D$-like factors.

\begin{notation}
  From now on, to shorten formulas and since $j$ is fixed, we shall
  omit most of the time the $\lj$ subscripts (but not the all-important $j$ subscript).
\end{notation}

\noindent The explicit formula for $k=2$ is also straightforward but longer. We
give it here for completeness:
\begin{multline}
  \partial_{\sigma_{s_{2}}}\partial_{\sigma_{s_{1}}}(-V_j )
  =-\lambda^{2}\Tr[\be_{m_{s_{1}}n_{s_{1}}}(Q_{j})_{c_{s_{1}}c_{s_{2}}}\be_{m_{s_{2}}n_{s_{2}}}]\\
  +\int_0^1 dt_j\, \Tr\bigl[\dU[1]{j}\dU[2]{}U\gls{cR}+\dU[1]{j}U\dU[2]{}\gls{cR}+\dU[1]{j}U^{2}\gls{cR}\dU[2]{}\gls{cR}\\
  +\dU[2]{j}\dU[1]{}U\gls{cR}+U'_{j}\dU[1]{}\dU[2]{}\gls{cR}+U'_{j}\dU[1]{}U\gls{cR}\dU[2]{}\gls{cR}\\
  +\dU[2]{j}U\dU[1]{}\gls{cR}+U'_{j}\dU[2]{}\dU[1]{}\gls{cR}+U'_{j}U\dU[1]{}\gls{cR}\dU[2]{}\gls{cR}\\
  +\dU[2]{j}U^{2}\gls{cR}\dU[1]{}\gls{cR}+U'_{j}\dU[2]{}U
  \gls{cR}\dU[1]{}\gls{cR}+U'_{j}U\dU[2]{}\gls{cR}\dU[1]{}\gls{cR}+U'_{j}U^{2}\gls{cR}\dU[2]{}\gls{cR}\dU[1]{}\gls{cR}+U'_{j}U^{2}\gls{cR}\dU[1]{}\gls{cR}\dU[2]{}\gls{cR}\\
  -D'_{1,j}(\dU[1]{}\dU[2]{}+\dU[2]{}\dU[1]{})-D_{1}(\dU[1]{j}\dU[2]{}+\dU[2]{j}\dU[1]{}+\dU[1]{}\dU[2]{j}+\dU[2]{}\dU[1]{j})\bigr].\label{eq-DerivSigmak2}
\end{multline}
The formula for $k\ges 3 $ is similar but has no longer the $D$ terms:
as they are quadratic in $\sigma$, they ``die out'' for $k \ges 3$
derivatives. Derivatives can only hit $p$ times the $U$ terms and
$k-p$ times the resolvent $\gls{cR}$, for $0\les p\les 3$. All in all,
the application of $k\ges 3$ $\sigma$-derivatives on $-V_{j}$ gives:
\begin{multline}
  \Bigl(\prod_{i=1}^k\partial_{\sigma_i}\Bigr) (-V_j )= \int_0^1 dt_j\,\Tr\Bigl[\sum_{\tau\in\cS_{[k]}}U'_{j}U^2\gls{cR}\bigl(\prod_{i=1}^{k}
  \dU[\tau(i)]{}\gls{cR}\bigr)\\
  +\sum_{i_{0}=1}^{k}\,\sum_{\tau\in\cS_{[k]\setminus\set{i_{0}}}}(\dU[i_{0}]{j}U^{2}+U'_{j}\dU[i_{0}]{}U+U'_{j}U\dU[i_{0}]{})\gls{cR}\bigl(\prod_{\substack{i=1\\i\neq
      i_{0}}}^{k}\dU[\tau(i)]{}\gls{cR}\bigr)\\
  +\sum_{\substack{i_{0},i_{1}=1\\i_{0}<
      i_{1}}}^{k}\,\sum_{\tau\in\cS_{[k]\setminus\set{i_{0},i_{1}}}}(\dU[i_{0}]{j}\dU[i_{1}]{}U+\dU[i_{0}]{j}U\dU[i_{1}]{}+\dU[i_{1}]{j}\dU[i_{0}]{}U+U'_{j}\dU[i_{0}]{}\dU[i_{1}]{}\\
  \hspace{5cm}+\dU[i_{1}]{j}U\dU[i_{0}]{}+U'_{j}\dU[i_{1}]{}\dU[i_{0}]{})\gls{cR}\Bigl(\
  \prod_{\mathclap{\substack{i=1\\i\neq
      i_{0},i_{1}}}}^{k}\dU[\tau(i)]{}\gls{cR}\Bigr)\\
  +\sum_{\substack{i_{0},i_{1},i_{2}=1\\i_{0}<
      i_{1}<i_{2}}}^{k}\,
  \sum_{\kappa\in\cS_{\set{i_{0},i_{1},i_{2}}}}\sum_{\tau\in\cS_{[k]\setminus\set{i_{0},i_{1},i_{2}}}}\dU[\kappa(i_{0})]{j}\dU[\kappa(i_{1})]{}\dU[\kappa(i_{2})]{}\gls{cR}\Bigl(\
  \prod_{\mathclap{\substack{i=1\\i\neq
      i_{0},i_{1},i_{2}}}}^{k}\dU[\tau(i)]{}\gls{cR}\Bigr)\Bigr] \label{eq-developcycles}
\end{multline}
where for any finite set $E$, $\cS_{E}$ denotes the permutations on
$E$. Remark that the special $C_j$ propagator is never lost in such formulas. 
They express the derivatives of $V_j$ as a sum over traces of
four-stranded cycles (also called loop vertices) corresponding to the
trace of an alternating product of propagators ($C_{\lj}$ or, only
once, $C_{j}$) and other operators on $\Htens$ nicknamed
\emph{insertions}. The number and nature of these insertions depend
on the number of derivatives applied to $V_{j}$. For $k<3$
derivatives, loop vertices contain between $4$ and $8$ insertions of
type $\delta,\sigma+B,\gls{cR},D_{1},D'_{1}$ or $D'_{2}$. For
$k\ges 3$, loop vertices of length $\ell$,
\ie having exactly $\ell$ insertions, with $2k-2 \les \ell \les
2k+4 $, bear insertions of type $\delta,\sigma+B$ or $\gls{cR}$. Each loop vertex has
exactly one \emph{marked propagator} $C_j$ which breaks the cyclic
symmetry. All the other ones are $C_{\lj}$. The corresponding sum over
all possible choices of insertions and their number is \emph{constrained} by the condition that there must be exactly $k$ $\delta$ insertions
in the cycle. A particular example is shown in \cref{f-loopvertex}.
\begin{figure}[!ht]
\begin{center}
  \includegraphics[width=8cm]{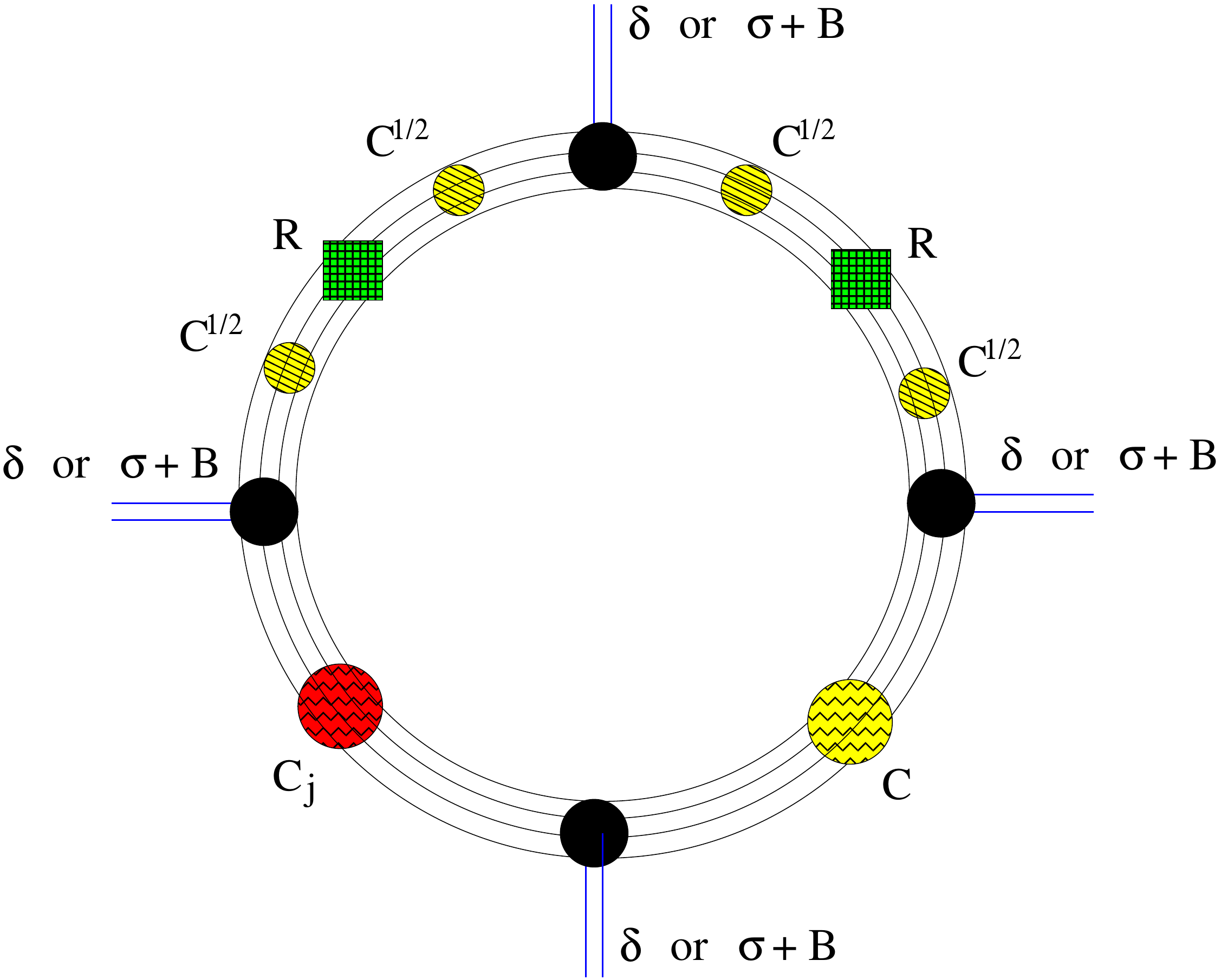}
\end{center}
\caption{An example of a four-stranded vertex of length four with its typical cycle of insertions. 
Black (matrix type) dots correspond to $\delta$ or $\sigma +B$ operators. Each has its well-defined colour, hence opens a well-defined
strand. Any $\delta$ insertion is in fact an half edge of the tree $\cT_\cB$, hence pairs with another vertex (not shown in the picture).
The marked insertion (pictured a bit larger and in red, together with its neighboring corners) indicates the presence of the slice $j$ 
propagator $C_j$. Resolvents are pictured as green squares. 
The sum is constrained to have exactly $k$ derived insertions of the $\delta$ type, the others are $\sigma +B$.}
\label{f-loopvertex}
\end{figure}
\\

Each effective vertex of $\skG$ now bears
exactly $\vert b^a \vert$ $\delta$ derivative insertions, which are paired together between vertices 
via the coloured edges of the tree $\cT_{\cB}$,
plus some additional (see above) remaining insertions. Note that to each initial $W_{j_a}$ may correspond several loop vertices $V_{b^a}$,
depending on the partitioning of $S^a_\cB$ in \eqref{eq-partitio}. Therefore although
at fixed $|\cB|$ the number of edges $m(\skG)$ for any $\skG$ in the sum \eqref{eq-bosogauss} is exactly $|\cB|-1$,
the number of connected components $c(\skG)$ 
is not fixed but simply bounded (above) by $|\cB|-1$ (each edge can
belong to a single connected component). Similarly the number $n(\skG)= c(\skG) + e(\skG)$ of effective loop vertices of $\skG$ is not fixed, and simply obeys the bounds
\begin{equation}
\abs{\cB}\les n(\skG) \les 2(\abs{\cB}-1). \label{eq-foreboun}
\end{equation}
From now on we shall simply call ``vertices'' the loop
vertices of $\skG$.

\subsubsection{Wick ordering by the \texorpdfstring{$\tau$}{tau} field}
\label{sec-wick-ordering-tau}

Each $\ell = (a,b)$ for which the $\tau$ derivatives
have been chosen, see \cref{eq-partialJ-TB}, creates exactly a divergent
vacuum graph $\kN_{2}$ (see \cref{f-VacuumNonMelonicDivergences-2})
obtained by contracting two quadratic $Q_0$ factors, one with scale $j_a$ and the other with scale
$j_b$. Fortunately this cancels out with a very special potentially
divergent quadratic $\sigma$ link. To check it, let us
perform exactly the remaining $\tau$ integral. The result is expressed in the following \namecref{thm-wotau}.
\begin{lemma}\label{thm-wotau}
  After integrating out the $\tau$ field, the expansion is the same as
  if there had never been any $\tau$ fields, but with two
  modifications:
  \begin{itemize}
  \item there exists an exponential of the counterterm
    \begin{equation*}
      \delta_{\kN_2, \cB} (\tuple{w}) = - \tfrac{\lambda^4}{4}\sum_{a,b \in \cB}  X^{\circ 2} (a,b) \Tr[Q_{0,j_a}  Q_{0,j_b}],
    \end{equation*}
  \item each $\sigma$ link for $\ell = (a, b)$ made of exactly one
    link between two $Q_0$ factors, is \emph{exactly Wick ordered with
      respect to the $d\nu_{\cB}(\sigmad)$ covariance}, namely its
    value in $A_\skG$ is
    $\wo{\sigmad^a\scalprod Q_{0, j_a} Q_{0,j_b} \sigmad^b}$.
  \end{itemize}
  In other words
  \begin{equation*}
    I_{\cB} = 
    \sum_\skG   \int d \nu_\cB (\sigmad) \, e^{\delta_{\kN_2, \cB} (\tuple w)} \, \Bigl(\prod_{a\in \cB}  e^{-V_{j_a} (\sigmads^{a} ) }   \Bigr)  \wo{A_\skG(\sigmad)},
  \end{equation*}
  where $\wo{A_\skG(\sigmad)}$ is obtained by the same formula as if
  there had never been any $\tau$ field, but with one modification:
  the Wick ordering indicates that each link of the type
  $\sigmad^a\scalprod Q_{0, j_a} Q_{0, j_b} \sigmad^b $ is Wick
  ordered with respect to the $d \nu_\cB (\sigmad)$ measure.
\end{lemma}
\begin{proof}
  The first part of the statement is obvious: integrating the linear
  $e^{i \frac{\lambda^{2}}{\sqrt 2} Q_0 \scalprod \tauds^{a}}$ terms with the
  $d\nu_\cB (\taud)$ interpolated covariance must give back the
  exponential of the full $\delta_{\kN_2}$ counterterm but with the
  weakening covariance factors $ X^{\circ 2} (a,b) $ between nodes $a$
  and $b$. The second statement is also not too surprising since the
  counterterm $\delta_{\kN_2}$ should compensate the divergent
  graphs $\kN_2$ which are brought down the exponential by the MLVE
  expansion. But let us check it explicitly. Any tree link
  $\ell = (a,b)$ in the Faà di Bruno formula either is a $\tau$ link
  hence created a term 
  \begin{equation*}
    2w_\ell \lbt\tfrac{i\lambda^2}{\sqrt 2}\rbt^{\!2} \Tr
  [Q_{0,j_a} Q_{0,j_b}] = - w_\ell \lambda^4 \Tr[Q_{0,j_a} Q_{0,j_b}],
\end{equation*}
or a $\sigma$ link. In this case it either has joined two
  $Q_0$ loop vertices each with one $\sigma$ field at its free end, or
  done something else.  In the first case, the expectation value of
  the corresponding term is
  \begin{equation*}
  \int d \nu_\cB (\sigmad)\, \lbt\tfrac{-\lambda^2}{2}\rbt^{\!2}\!2^{2}\,
  \sigmad^a\scalprod Q_{0,j_a} Q_{0,j_b} \sigmad^b = w_\ell \lambda^4 \Tr[
  Q_{0,j_a} Q_{0,j_b}].
\end{equation*}
This proves the second statement: such
  $\sigma$ links are exactly Wick-ordered by the $\tau$ links.
\end{proof}

From now on we can therefore forget the auxiliary $\tau$ field. Its
only purpose was to effectuate the compensations expressed by
\cref{thm-wotau}, without disturbing too much the ``black box'' of the
\MLVEac. Moreover, anticipating on \cref{sec-pert-funct-integr}, notice that the
functional integration (with respect to the Gaussian measure
$\nu_{\cB}$) of the ``graphs'' $\skG$ would result in (perturbative
series of) purely convergent Feynman graphs.

\subsubsection{Perturbative and non-perturbative contributions}
\label{sec-pert-non-pert}

In all cases (including the single isolated block case)
we  apply a Hölder inequality with respect to the 
positive measure $d\nu_\cB$ to separate four parts: the perturbative part ``down from the exponential'', the particular
$\frac{\lambda^2}{2}\wo{\sigmad\scalprod Q_{0,j}  \sigmad}_{\scalebox{.6}{$X$}}$ Wick-ordered term (which requires special care, since without the Wick ordering 
it would lead to a linearly divergent bound which could not be paid for),
the other non perturbative quadratic or less than quadratic factors,
which we define as
\begin{equation}
\widetilde V^{\les 2}_j  \defi  \tfrac{\lambda^2}{2}\wo{\sigmad\scalprod Q_{1,j} \sigmad}_{\scalebox{.6}{$X$}}  - 3\int_{0}^{1}dt_{j}\,\Tr[D'_{2,j}\Sigma_{\les j}]\label{eq-Vtildej2}
\end{equation}
(remember the $\tau$ field has been integrated out, hence replaced by
the $\delta_{\kN_2, \cB}(\tuple w)$ counterterm),
and finally the higher order non-perturbative factor $V^{\ges 3}_j$. This last factor 
will require extra care and the full \cref{sec-boson-non-pert-int} for its non perturbative
bound.
\begin{rem}
  The careful reader would have noticed the extra index $X$ associated
  to the Wick ordering of both $\sigmad\scalprod Q_{0,j}\sigmad$ and
  $\sigmad\scalprod Q_{1,j}\sigmad$. The Wick ordering of those terms
  were originally defined \wrt the Gaussian measure of covariance
  $\Idirect$ \ie before the jungle formula and thus before the
  interpolation of the covariance (see
  \cref{eq-originalwo}). Nevertheless the contraction of the two
  $\sigmad$'s (in both expressions) corresponds to a tadpole
  intermediate graph and is thus never accompanied by weakening
  factors $w$. We can therefore equally well consider that the two terms
  above-mentionned are Wick ordered \wrt the interpolated measure
  of covariance $X$.
\end{rem}
Finally we write:
\begin{multline}
  \label{eq-CS-Pert-NonPert}
  \abs{I_{\cB}}\les \vert e^{\delta_{\kN_2, \cB}
    (\tuple w)}
  \vert \Bigl(\underbrace{ \int d \nu_\cB \prod_{a\in\cB} e^{-2\Re
      (\lambda^2 ) \wo{\sigmads\scalprod Q_{0,j_a}
        \sigmads}_{\scalebox{.5}{$X$}}}}_{\text{$I_{1}$, non-perturbative}} \Bigr)^{\!1/4}\, \Bigl( \underbrace{\int d \nu_\cB \prod_{a\in\cB} e^{-4 \Re(\widetilde V^{\les
        2}_{j_a} (\sigmads^{a}))}}_{\text{$I_{2}$, non-perturbative}}
  \Bigr)^{\!1/4}
  \\
  \times \Bigl( \underbrace{\int d \nu_\cB \prod_{a\in\cB} e^{4 \vert V^{\ges
        3}_{j_a} (\sigmads^{a}) \vert }}_{\text{$I_{3}$, non-perturbative}}
  \Bigr)^{\!1/4}\, \sum_\skG \Bigl( \underbrace{\int d \nu_\cB\, \abs{\wo{A_\skG
      (\sigmad)}}^4}_{\text{$I_{4}$, perturbative}} \Bigr)^{\!1/4}.
\end{multline}
To bound such expressions, and in particular the ``non-perturbative" terms,
requires to now work out in more details explicit formulae which in particular 
show the compensation between the terms of \cref{eq-nicevj1}
. 

\newpage
\section{Estimates for the interaction}
\label{sec-expl-form-v_j}

This section is a technical interlude before estimating the
non-perturbative terms of \cref{eq-CS-Pert-NonPert} in \cref{sec-boson-non-pert-int}. In its first
subsection, we make explicit the cancellations at work in $V_{j}$ and
derive a quadratic bound (in $\Sigma$) on $\abs{V_{j}^{\ges 3}}$. It
will be used to prove \cref{propnonpert} which constitutes one step
towards the bound on $I_{3}$ of \cref{eq-CS-Pert-NonPert}. In
its second subsection, we get a quartic bound on $\abs{V_{j}^{\ges
    3}}$ used both in \cref{sec-pert-funct-integr} and in
\cref{sec-boson-non-pert-int} but this time to prove another step of our final bound on $I_{3}$, namely \cref{thm-AGestimate}.

\subsection{Cancellations and quadratic bound}

In this section, we first derive a new expression for $V^{\ges 3}_{j}$
(\cref{eq-startvjeqnice}) in order to explicitely show the cancellation involving the
$\gls{lastvac}[_{j}]$ counterterm. Then we prove a so-called \emph{quadratic
bound} on $\abs{V^{\ges 3}_{j}}$ (\cref{thm-lemmaquadbound}) in terms of a
quadratic form in $\sigma$. This estimate will be useful in
\cref{sec-boson-non-pert-int}.\\

\noindent
In the sequel, we will be using repeatedly the few following facts:
  \begin{equation*}
  \begin{alignedat}{3}
    [U,\fres]=0,&\qquad&\indic_{j}=\frac{d\indic_{\lj}}{dt}(t_{j}),&\qquad&\indic_{j}^{2}=\indic_{j},\\
    [D_{1},\indic_{j}]=[D_{2},\indic_{j}]=0,&&D'_{j}=D\indic_{j},&&\Sigma'_{j}=\indic_{j}\Sigma+\Sigma\indic_{j}.
  \end{alignedat}
\end{equation*}
\vspace{-\abovedisplayskip}
\begin{notation}
  From now on, in order to simplify long expressions, we will mainly
  trade the ${}'$ notation (\eg $D',\Sigma'$) for the ones with
  explicit cutoff $\indic_{j}$ (\eg $D\indic_{j},\indic_{j}\Sigma+\Sigma\indic_{j}$).
\end{notation}
So let us return to \cref{eq-nicevj}, using cyclicity of the trace, $(\Itens + U  -  \fres  )  =  (
\Itens  -  \fres)U = U(\Itens-\fres) = -U\fres U$, and
$D=D_{1}+D_{2}$, we define
\begin{align*}
  V^{\ges 3}_{j}   &\fide \cE_{j}+\int_0^1  dt_{j}\, \tilde v_j,\\
  \tilde v_j &=\Tr\bigl[U'_{j}(\Itens+U_{\les j}-\gls{cR}[_{\les j}])+D'_{1,j}\Sigma^{2}+D_{1,\les
      j}\Sigma'_{j}\Sigma+D_{1,\les j}\Sigma\Sigma'_{j} \bigr]\\
  &=\Tr\bigl[(U\indic_{j}\Sigma+\Sigma\indic_{j}U)(\Itens-\fres)-U\indic_{j}D\unj
  U\fres +3D_{1}\unj\Sigma^{2}\unj+2D_{1}\Sigma\unj\Sigma\bigr]\\
  &=\Tr\bigl[(U\indic_{j}\Sigma+\Sigma\indic_{j}U+\Sigma\indic_{j}D\indic_{j}\Sigma)(\Itens-\fres)-D^{3}\indic_{j}\fres-D^{2}\indic_{j}\Sigma\fres-\Sigma\indic_{j}D^{2}\fres\nonumber\\
  &\hspace{8cm}-D_{2}\indic_{j}\Sigma^{2}\unj+2D_{1}(\Sigma\unj\Sigma+\unj\Sigma^{2}\unj).
\end{align*}
In order to show the compensation involving $\gls{lastvac}[_{j}]$, we now expand the $D^3\unj\fres$ term, as
\begin{equation*}
  \Tr\bigl[D^3\unj\fres\bigr] = \Tr\bigl[D^3\unj   + D^4\unj + D^5\unj\fres + D^{3}( \unj +D\unj )\Sigma\fres \bigr].
\end{equation*}
We further expand the pure $D$ terms as  $\Tr\bigl[D^3\unj + D^4\unj \bigr]\fide\cD_{conv,j} +  \cD_{div,j}$ with
  \begin{align*}
    \cD_{conv,j} &\defi   \cD_{conv,\lj} -  \cD_{conv,\lj-1}, \quad \cD_{div,j}\defi   \cD_{div,\lj} -  \cD_{div,\lj-1},\\
    \cD_{conv,\lj}&\defi\Tr \bigl[ \tfrac{1}{3}   D^3_{2,\lj}  +D_{1,\lj} D^2_{2,\lj}    + \tfrac{1}{4}\bigl((D_{1,\lj}  + D_{2,\lj} )^4 - D^4_{1,\lj}\bigr)  \bigr], \\
    \cD_{div,\lj}&\defi \Tr\bigl[ \tfrac{1}{3} D^3_{1,\lj} +
    D^2_{1,\lj} D_{2,\lj} + \tfrac{1}{4} D_{1,\lj}^4 \bigr] =
    \gls{lastvac}[_{\lj}].
  \end{align*}
Clearly, $\int_{0}^{1}dt_{j}\,\cD_{div,j}=\gls{lastvac}[_{j}]$. Hence,
redefining $V_{j}^{\ges 3}\fide\int_{0}^{1}dt_{j}\,v_{j}$ and $v_{j}\defi v_{j}^{(0)}+v_{j}^{(1)}+v_{j}^{(2)}$, we have
\begin{equation}\label{eq-startvjeqnice}
  \lb\begin{aligned}
    v_{j}^{(0)}&=-\Tr\bigl[D^{5}\unj\fres\bigr]-\cD_{conv,j},\\
    v_{j}^{(1)}&=\Tr\bigl[(D\unj\Sigma+\Sigma\unj
    D)(\Itens-\fres)-D^{2}\unj\Sigma\fres-\Sigma\unj
    D^{2}\fres -D^{3}\unj\Sigma\fres-D^{4}\unj\Sigma\fres\bigr],\\
    v_{j}^{(2)}&=\Tr\bigl[(2\Sigma\unj\Sigma+\Sigma\unj
    D\unj\Sigma)(\Itens-\fres)-D_{2}\unj\Sigma^{2}\unj+2D_{1}(\Sigma\unj\Sigma+\unj\Sigma^{2}\unj)\bigr].
  \end{aligned}\right.
\end{equation}
This has shown the desired cancellation of the $\gls{lastvac}[_{j}]$
counterterm with the $-\cD_{div,j}$ term.\\

We now turn to the proof of the following
\namecref{thm-lemmaquadbound}, suited to a non-perturbative sector of
the model analysis, which bounds $\abs{V_j}$ in terms of a quadratic
form
$\gls{Qj}(\sigmad)\defi\tfrac{1}{\abs{g}}\Tr\bigl[\Sigma^{*}\unj\Sigma\bigr]$,
since higher order bounds can certainly not be integrated out with
respect to the Gaussian measure $d\nu_{\cB}$.
\begin{lemma}[Quadratic bound]\label{thm-lemmaquadbound}
For $g$ in the cardioid domain $\Card_\rho$, there exists a real
positive number $k$ such that
\begin{equation*}
\abs{V^{\ges 3}_{j} }\les k\rho\,(1 + \Qj(\sigmad)). \label{boundlemmanopert}
\end{equation*}
\end{lemma}
The proof of \cref{thm-lemmaquadbound} requires the following upper bounds

\begin{prop}[Norms and traces]\label{easylemma}
For all $0<\veps< 1$, for any $t_j \in [0,1]$ and $g$ in the cardioid,
\begin{align*}
  \norm{\fres}&\les 2\rho/\abs g,&\Tr[D^{4}\unj]&\les\Oun \abs{g}^{4},\\
  \norm{D}&\les\Oun \abs g,&\abs{\Tr[D^{5}\unj]}&\les\Oun \abs{g}^{5}M^{-j},\\
  \abs{\cD_{conv,j}}&\les\Oun \abs{g}^{5}M^{-(1-\veps)j},&\Tr[D^{6}\unj]&\les\Oun \abs{g}^{6}M^{-2j},\\
  &&\Tr[D^{8}\unj]&\les\Oun \abs{g}^{8}M^{-4j}.
\end{align*}
\end{prop} 
\begin{proof}
  Apart from the bound on $\norm{\fres}$ which uses
  \cref{thm-lemmaresbounded} and the definition of the cardioid
  domain, the other ones are standard exercises in perturbative power counting.
\end{proof}
Finally, before we prove \cref{thm-lemmaquadbound}, let us state the
following inequalities that we shall use extensively in this
section and the next one.
\begin{prop}[Trace inequalities]
  Let $A,B,C,E$ be complex square matrices of the same size. Let
  $\norm{A}[2]$ denote $(\Tr[AA^{*}])^{1/2}$ where ${}^{*}$ denotes the Hermitian conjugation. We have:
  \begin{enumerate}
  \item Hilbert-Schmidt bound (hereafter HS)
    \begin{equation}
      \label{eq-HSdef}
      \abs{\Tr[AB]}\les\norm{A}[2]^{2}+\norm{B}[2]^{2}.
    \end{equation}
  \item $L^{1}/L^{\infty}$ bound: if $A$ is Hermitian (and $B$ bounded),
    \begin{equation}\label{eq-LunLinftyDef}
      \abs{\Tr[AB]}\les\norm{B}\Tr[\abs A]
    \end{equation}
    where $\norm{\scalprod}$ denotes the operator norm.
  \item Cauchy-Schwarz inequality:
    \begin{equation}
      \label{eq-TensorCSDef}
      \abs{\Tr[ABCE]}\les\norm{A}\norm{C}\norm{B}[2]\norm{E}[2].
    \end{equation}
  \end{enumerate}
\end{prop}
The proofs are very standard and anyway simple enough to be avoided here.
\begin{proof}[of \cref{thm-lemmaquadbound}]
  We first notice that
  $\abs{V^{\ges 3}_{j}}\les\int_{0}^{1}dt_{j}\,\abs{v_{j}}$. Then $\abs{v_{j}}$
  is smaller than the sum of the modules of each of its terms. As all
  our bounds will be uniform in $t_{j}$, we can simply focus on the
  modules of each of the terms of $v_{j}$. Starting with
  $v_{j}^{(0)}$, and according to \cref{easylemma}, we have
  $\abs{v_{j}^{(0)}}\les\Oun \rho$.\\

  As $\abs{\Tr[D^{2}]}=\cO(M^{2j})$, a price we cannot afford to pay,
  we cannot simply apply a HS bound (see \cref{eq-HSdef}) to the first
  two terms of $v_{j}^{(1)}$. We need to expand the resolvent one step
  further:
  \begin{align*}
    v_{j}^{(1)}&=\Tr\bigl[-(\Sigma+D)D\unj\Sigma-\Sigma\unj
    D(D+\Sigma)\fres-D^{2}\unj\Sigma\fres-\Sigma\unj
    D^{2}\fres-D^{3}\unj\Sigma\fres-D^{4}\unj\Sigma\fres\bigr]\\
    &=-\Tr\bigl[2\Sigma\unj D\unj\Sigma\fres+2D^{2}\unj\Sigma\fres+2\Sigma\unj
    D^{2}\fres+D^{3}\unj\Sigma\fres+D^{4}\unj\Sigma\fres\bigr].\label{eq-vjunQuad}
     \end{align*}
     To the first term we apply the bound
     \eqref{eq-TensorCSDef} with $A=\fres, B=\Sigma\unj, C=D,
     E=\unj\Sigma$ to get
     \begin{equation*}
       \abs{\Tr\bigl[\Sigma\unj
         D\unj\Sigma\fres\bigr]}\les\norm{\fres}\norm{D}\,\abs{\Tr\bigl[\Sigma\unj\Sigma\bigr]}\les
       \Oun \rho\,\Qj(\sigmad).
     \end{equation*}
     All the other terms of $v_{j}^{(1)}$ are bounded the same way:
     first a HS bound then a $L^{1}/L^{\infty}$ one. For example:
     \begin{align*}
       \abs{\Tr[D^{2}\unj\Sigma\fres]}&\les\Tr[\fres^{*}\fres
       D^{4}\unj]-\Tr[\Sigma\unj\Sigma]\\
       &\les \norm{\fres^{*}\fres}\Tr[D^{4}\unj]+\abs{g}\Qj(\sigmad)\les\Oun \rho\,(1+\Qj(\sigmad)).
     \end{align*}
The other terms of $v_{j}^{(1)}$ are in fact better behaved.\\

Finally let us turn to $v_{j}^{(2)}$. For each term, we apply the
bound \eqref{eq-TensorCSDef} with $B=\Sigma\unj$ and
$E=\unj\Sigma$. We let the reader check that it leads to the desired result.
\end{proof}

\subsection{Convergent loop vertices and quartic bound}
\label{sec-conv-loop-vert}

We want to establish a second bound on $\abs{V^{\ges 3}_{j}}$, more suited to perturbation theory 
than \cref{thm-lemmaquadbound}. The idea is to get a bound in a finite number of loop vertices types which have been freed of any 
resolvent through the successive use of a Hilbert-Schmidt inequality
and a $L^{1}/L^{\infty}$ bound.

The constraints are many. We want first the loop vertices to be
convergent (i.e. any graph built solely out of them must converge).
This excludes loop vertices of the type $\Tr[\Sigma^2]$ or $\Tr[D_{1} \Sigma^2]$. Another important constraint will be to keep a propagator of scale
exactly $j$ in each piece $A$ and $B$ which are to be separated by a
HS inequality. This forces us to be careful about the ordering of our
operators, to ensure that the HS ``cut'' keeps one $\indic_{ j}$ cutoff \emph{both} in the two halves $A$ and $B$.
\begin{defn}[Convergent loop vertices]\label{def-cvLoopVertices}
  Let us define the following convergent and positive loop vertices
  \begin{align*}
    U_{j}^{0,a}&\defi\tfrac{1}{\abs{g}^{6}}\Tr[D^{6}\unj],&U_{j}^{2,a}&\defi\tfrac{1}{\abs
      g^{3}}\Tr[D^{2}\unj\abs{\Sigma}^{2}],&U_{j}^{2,d}&\defi\tfrac{1}{\abs g^{3}}\Tr[D_{2}\unj\abs{\Sigma}^{2}],\\
    U_{j}^{0,b}&\defi\tfrac{1}{\abs
      g^{5}}\Trsb{D^{5}\unj},&U_{j}^{2,b}&\defi\tfrac{1}{\abs
      g^{3}}\Tr[D^{2}\Sigma^{*}\unj\Sigma],&U_{j}^{2,e}&\defi\tfrac{1}{\abs
      g^{3}}\Tr[D_{2}\Sigma^{*}\unj\Sigma],\\
    U_{j}^{0,c}&\defi\tfrac{1}{\abs{g}^{5}}\cD_{conv,j},&U_{j}^{2,c}&\defi\tfrac{1}{\abs
      g^{5}}\Tr[D^{4}\unj\abs{\Sigma}^{2}],&U_{j}^{4}&\defi\tfrac{1}{\abs
      g^{2}}\Tr[\abs{\Sigma}^{4}\unj].\\
    \intertext{as well as the following convergent ones}
    U_{j}^{1,a}&\defi\tfrac{1}{\abs
      g^{5/2}}\Tr[D^{2}\unj\Sigma],&U_{j}^{1,b}&\defi\tfrac{1}{\abs
      g^{7/2}}\Tr[D^{3}\unj\Sigma],&U_{j}^{3}&\defi\tfrac{1}{\abs g^{3/2}}\Tr[\Sigma^{3}\unj].
  \end{align*}
\end{defn}
\begin{lemma}[Quartic bound]\label{lemmaquarticbound}
  Let us define the following finite sets:
  $A_{3}=A_{4}\defi\set{a}$, $A_{0}\defi\set{a,b,c}$, $A_{1}\defi\set{a,b}$ and
  $A_{2}\defi\set{a,b,c,d,e}$. Let $U_{j}^{i,a}$ be defined as $U_{j}^{i}$ for
  $i\in\set{3,4}$. For all $0\les i\les 4$, let $\kU_{j}^{i}$ be
  $\sum_{\alpha\in A_{i}}\abs{U_{j}^{i,\alpha}}$. Then, for any $g$ in the cardioid domain,
  \begin{equation*}\label{eq-quarticbound}
    \abs{V^{\ges 3}_{j}}\les\Oun(\rho^{2}\kU_{j}^{4}+\rho^{3/2}\kU_{j}^{3}+\rho^{3}\kU_{j}^{2}+\rho^{5/2}\kU_{j}^{1}+\rho^{5}\kU_{j}^{0}).
  \end{equation*}
\end{lemma}
\begin{cor}\label{thm-eighticbound}
  For all $0<\epsilon< 1$, for any $g$ in the cardioid,
  \begin{equation*}
    \abs{V^{\ges 3}_{j} }^{2}\les\Oun  \rho^{3} (M^{-(2-\epsilon)j}+\sum_{i=1}^{4}\sum_{\alpha\in
    A_{i}}\abs{U_{j}^{i,\alpha}}^{2}).
  \end{equation*}
\end{cor}
\begin{proof}
  From \cref{lemmaquarticbound}, we use \cref{easylemma}, $\rho\les 1$ and
  the Cauchy-Schwarz inequality $(\sum_{i=1}^{p}a_{i})^{2}\les p\sum_{i=1}^{p}a_{i}^{2}$.
\end{proof}
We postpone the proof of \cref{lemmaquarticbound} to \cref{sec-proof-quartic} and give here
only its main structure. Starting with \cref{eq-startvjeqnice}, the
idea is to apply, to each term of $\abs{V^{\ges 3}_{j} }$, a HS bound
\eqref{eq-HSdef} (to get positive vertices) followed by a $L^{1}/L^{\infty}$ inequality
\eqref{eq-LunLinftyDef} (to get rid of the resolvents). The only problem is
that not all terms in \cref{eq-startvjeqnice} would result in
convergent vertices under such a procedure. Thus we need to expand the
resolvent until the new terms are ready for a HS bound, always taking
great care of the operator order in such a way that both sides of the
HS cut receive a cut-off operator $\unj$. All details are given in \cref{sec-proof-quartic}.

\section{Non perturbative functional integral bounds}
\label{sec-funct-integr-bounds}

\subsection{Grassmann integrals}
\label{sec-grassmann-integrals}
They are identical to those of  \cite{Gurau2014ab,Delepouve2014aa}, resulting in the same computation:
\begin{multline*}
  \int \prod_{\cB} \prod_{a\in \cB} ( d \bar \chi^{\cB}_{j_a} d
  \chi^{\cB}_{j_a} ) e^{ - \sum_{a,b=1}^n \bar \chi^{\cB(a)}_{j_a}
    \mathbf{Y}_{\!ab} \chi^{\cB(b)}_{j_b} }
  \prod_{\substack{\ell_F \in \cF_F\\\ell_F=(a,b)}}
  \delta_{j_{a } j_{b } } \Big( \chi^{\cB(a)}_{j_{a} } \bar
  \chi^{\cB(b)}_{j_{b } } + \chi^{ \cB( b) }_{j_{b} } \bar
  \chi^{\cB(a) }_{j_{a} } \Big)\\
  = \Bigl( \prod_{\cB}
  \prod_{\substack{a,b\in \cB\\a\neq b}} (1-\delta_{j_aj_b})
  \Bigr) \Bigl( \prod_{\substack{\ell_F \in
      \cF_F\\\ell_F=(a,b)}} \delta_{j_{a } j_{b } } \Bigr) \Bigl(
  \mathbf{Y}^{\hat b_1 \dots \hat b_k}_{\hat a_1 \dots \hat a_k} + \mathbf{Y}^{\hat a_1 \dots \hat b_k}_{\hat b_1 \dots \hat a_k}+\dots +
  \mathbf{Y}_{\hat b_1 \dots \hat b_k}^{\hat a_1 \dots \hat a_k} \Bigr),
\end{multline*}
where $k= \vert  \cF_F \vert $, the sum runs over the $2^k$ ways to exchange an $a_i$ and a $b_i$,
and the $Y$ factors are (up to a sign) the minors of $Y$ with the lines $b_1\dots b_k$ and the columns $a_1\dots a_k$ deleted.
The factor $\Bigl( \prod_{\cB} \prod_{\genfrac{}{}{0pt}{}{a,b\in
    \cB}{a\neq b}} (1-\delta_{j_aj_b}) \Bigr)$ ensures that the scales
obey a \emph{hard core constraint inside each block}. Positivity of the $Y$ covariance means as usual that the $Y$ minors are all bounded by 1 \cite{Abdesselam1998aa,Gurau2014ab}, namely
for any $a_1,\dots a_k$ and $b_1,\dots b_k$,
\begin{equation*}
\Big{|}  {\bf Y }^{\hat a_1 \dots \hat b_k}_{\hat b_1 \dots \hat a_k} \Big{|}\les 1.
\end{equation*}

\subsection{Bosonic integrals}
\label{sec-boson-non-pert-int}

This section is devoted to bound the non perturbative terms
\begin{multline}\label{eq-def-IBNP}
  I_{\cB}^{\mathit{NP}}\defi \vert e^{\delta_{\kN_2, \cB}(\tuple w)}
  \vert \Bigl(\int d \nu_\cB \prod_{a\in\cB} e^{4 \vert V^{\ges
        3}_{j_a} (\sigmads^{a}) \vert}\Bigr)^{\!1/4}\, \Bigl( \int d \nu_\cB \prod_{a\in\cB} e^{2\Re
      (\lambda^2 ) \wo{\sigmads^{a}\scalprod Q_{0,j_a} \sigmads^{a}}_{\scalebox{.5}{$X$}}} \Bigr)^{\!1/4}
  \\
  \times\Bigl( \int d \nu_\cB \prod_{a\in\cB} e^{-4 \Re(\widetilde V^{\les
        2}_{j_a} (\sigmads^{a}))} \Bigr)^{\!1/4}
\end{multline}
in \cref{eq-CS-Pert-NonPert}. Thus we work 
within a fixed Bosonic block $\cB$
and a fixed set of scales $S_{\cB}\defi\{j_a \}_{a \in \cB}$, \emph{all distinct}.
To simplify, we put $b =\card\cB\les n$ where $n$ is the order of
perturbation in \cref{eq-ZafterJungle}.
\begin{thm}\label{thm-npBound}
 For $\rho$ small enough and for any value of the $w$ interpolating parameters, there
exist positive $\Oun$ constants such that for $\vert \cB \vert \ges 2$
\begin{equation*}
I^{\mathit{NP}}_{\cB} \les\Oun
\,e^{ \Oun \rho^{3/2}\card{\cB}}. 
\end{equation*}
If $\cB$ is reduced to a single isolated node $a$, hence $b =1$
\begin{equation*}
\Big{\vert}  \int d\nu_{a} (\sigmad^a) \bigl(  e^{-V_{j_a} (\sigmads^a ) } -1 \bigr)   \Big{\vert} \les \Oun  \rho^{3/2} .
\end{equation*}
\end{thm}
Those results are similar to \cite{Delepouve2014aa} but their proof is completely different.
Since our theory is more divergent, we need to Taylor expand much
farther. The rest of this \lcnamecref{sec-boson-non-pert-int} is devoted to the
proof of \cref{thm-npBound}.\\

Let us first of all give some definitions:
\begin{defn}[$\Quu$, $\Qud$ and $\Qzu$]\label{def-Q012}
  Let $\Quu\in\EndHopd$ be given by its entries in the momentum basis:
  \begin{equation*}
    (\Quu)_{cc';mn,m'n'}=(1-\delta_{cc'})\delta_{mn}\delta_{m'n'}\sum_{\tuple
    r\in[-N,N]^{2}}\frac{1}{(m^{2}+m'^{2}+\tuple r^{2}+1)^{2}}
  \end{equation*}
  and $\lambda^{2}\Qud$ be $Q-Q_{0}-\Quu$, see
  \cref{eq-Qexpr,eq-Q0expr} for the definitions of $Q$ and $Q_{0}$. Finally let $\Qzu$ be $Q_{0}+\Quu$.
\end{defn}
\begin{defn}[Operators on $V_{\cB}$]
  Let $\be_{ab}$ be the $\card\cB\times\card\cB$ real matrix the elements
  of which are $(\be_{ab})_{mn}\defi\delta_{am}\delta_{bn}$. Let $\cA$
  be a subset of $\cB$ and for all $P\in\EndHopd$, let $P_{\cA}$ be
  the following linear operator on $\gls{VB}\defi \R^{\abs{\cB}} \otimes \Hopdirect$:
  \begin{equation*}
    P_{\cA}\defi\sum_{a\in\cA}P\indic_{j_{a}}\otimes\be_{aa}.
  \end{equation*}
    Let $\Qt_{1}$ be $(\Re\lambda^{2})\Quu+(\Re\lambda^{4})\Qud$.
\end{defn}
The first step consists in estimating certain determinants:
\begin{prop}[Determinants]\label{thm-determinants}
Let
  $A_{0},A_{1},A_{2}$ stand respectively for
  $\rho \gls{XB}Q_{0,\cA}$, $X_{\cB}\widetilde Q_{1,\cA}$ and
  $\rho X_{\cB}\cst{Q}{\cA}{01}$. Then, for $\rho$ small enough, we have
  \begin{equation*}
    \Det_{2}(\rIdirect-A_{0})^{-1}\les
    e^{\Oun\rho^{2}\card{\cA}},\quad\Det_{2}(\rIdirect-A_{1})^{-1}\les
    e^{\Oun\rho^{2}},\quad\Det(\rIdirect-A_{2})^{-1}\les
    e^{\Oun\rho M^{j_{1}}}
  \end{equation*}
where $\rIdirect$ is the identity operator on $V_{\cB}$, $\Det_{2}(\rIdirect-\cdot)\defi
e^{\Tr\log_{2}(\rIdirect-\cdot)}$ and $j_{1}\defi\sup_{a\in\cA}j_{a}$.
\end{prop}
\begin{proof}
  Let us start with $A_{2}$. Since
  $\Qzu[\cA] = \sum_{a \in \cA} \Qzu[j_{a}]\otimes\be_{aa}$, we find
  that
  \begin{equation*}
    \rTrd A_{2} = \rho\rTrd[X_\cB\Qzu[\cA]] =
    \rho\sum_{a \in \cA} X_{aa}(\tuple{w}_{\cB})\Trd\Qzu[j_{a}] = \rho\sum_{a \in \cB'} \Trd\Qzu[j_{a}].
  \end{equation*}
  Using \cref{thm-Qj}, we have
  \begin{equation*}
    \sum_{a \in \cB'} \Trd\Qzu[j_a]\les
    \Oun\sum_{a \in \cA} M^{j_a} \les \Oun\, M^{j_{1}}
  \end{equation*}
  where in the last inequality we used that all vertices $a\in\cB$
  have \emph{different scales} $j_a$.

  Furthermore by the triangular inequality and \cref{thm-Qj} again,
  \begin{equation*}
    \norm{A_{2}} \les \rho\sum_{a \in \cA} \norm{X(\tuple w_{\cB}
      )\be_{aa}}\,\norm{\Qzu[j_a]} \les
    \rho\sum_{a \in \cA} \norm{\Qzu[j_a]} \les \Oun\rho\sum_{j =0}^{\infty}
    M^{-j} =\Oun\rho
  \end{equation*}
  where we used that $\norm{X(\tuple w_{\cB})\be_{aa}}=1$ and again
  that all vertices $a\in\cB$ have different scales.

  Remarking that by the above upper bounds on $\Tr A_{2}$ and
  $\norm{A_{2}}$, for $\rho$ small enough, the series
  $\sum_{n=1}^\infty \tfrac 1n \rTrd[A^n]$ converges, we have
  \begin{align*}
    \det (\rIdirect - A_{2})^{-1} &= e^{-\rTrd[\log (\rIdirect -
      A_{2})]} = e^{\sum_{n=1}^\infty \tfrac
      1{n}\rTrd[A_{2}^n]}\nonumber\\
    &\les e^{\rTrd[A_{2}]\sum_{n=1}^\infty \norm{A_{2}}^{n-1}} =
    e^{\Oun\rho M^{j_1}}.
  \end{align*}
  The cases of $A_{0}$ and $A_{1}$ are very similar. For example,
  \begin{equation*}
    \Tr A^{2}_{0}= \rho^{2}\sum_{a,a'\in\cA}\Tr[X(\tuple
    w_{\cB})\be_{aa}X(\tuple w_{\cB})\be_{a'a'}\otimes Q_{0,j_{a}}Q_{0,j_{a'}}]=\rho^{2}\sum_{a,a'}\delta_{aa'}X_{aa}X_{aa}\Tr[Q_{0,j_{a}}^{2}]\les\Oun\rho^{2}\card{\cA}
  \end{equation*}
by \cref{thm-Qj}. Likewise,
  \begin{equation*}
    \Tr
    A^{2}_{1}\les\Oun\rho^{2},\quad\norm{A_{0}}\les\Oun\rho,\quad\norm{A_{1}}\les\Oun\rho.
  \end{equation*}
  Finally, using $\Det_{2}(\rIdirect-A)\les e^{\frac
    12\Tr[A^{2}]\sum_{n\ges 2}\norm{A}^{n-2}}$, we conclude the proof.
\end{proof}

We can now treat the easy parts of $I^{\mathit{NP}}_{\cB}$. It is obvious that 
\begin{equation*}
\vert e^{\delta_{\kN_2, \cB} (w)}  \vert  \les \Oun  e^{ \Oun  \rho^{2}\card{\cB}},
\end{equation*}
since the counterterm $\delta_{\kN_2}$ is {logarithmically} divergent,
hence it can be bounded by a constant per slice $j$ (times $\rho^2$,
see \cref{thm-wotau}).\\

\noindent
The piece $\Bigl(\int d \nu_\cB  \prod_a  e^{2
  ( \Re \lambda^2 )\wo{\sigmads^{a}\scalprod  Q_{0,j_a} \sigmads^{a}}}
\Bigr)^{\!1/4} $ can be bounded through an explicit computation:
\begin{equation*}
\int d \nu_\cB  \prod_a  e^{2 ( \Re \lambda^2 ) \wo{\sigmads^{a}\scalprod  Q_{0,j_a} \sigmads^{a}}}   =  \Det_{2}(\rIdirect - A^0_\cB)^{-1/2}
\end{equation*}
where $A^0_\cB$ equals $4  (\Re\lambda^2) X_{\cB}Q_{0,\cB}$. Using
\cref{thm-determinants}, we get 
\begin{equation*}
  \Det_{2}(\rIdirect - A^0_\cB)^{-1/2}\les e^{\Oun\rho^{2}\card\cB}
\end{equation*}
which reproduces the desired bound. Remark that the Wick-ordering here is
absolutely essential to suppress the $\rTrd[A^0_\cB]$ term, since that
term is \emph{linearly} divergent.\\

\noindent
The bound on $\int d \nu_\cB  \prod_a  e^{-4 \Re(\widetilde V^{\les 2}_{j_a} (\sigmads^{a}))}$ 
is similar. It consists in an exact Gaussian integration but this time
with a source term $\int_{0}^{1}dt_{j}\,\Tr[D'_{2,j}\Sigma_{\les j}]$,
see \cref{eq-Vtildej2}. Let us define $\cD_{2,j}$ as
$C^{1/2}D'_{2,j}C^{1/2}$, $\underline{\cD}_{2,j}$ as
$\tfrac{1}{\lambda^{5}}\int_{0}^{1}dt_{j}\,\cD_{2,j}$ and
$\Dvec_{2,j}$ such that
$(\Dvec_{2,j})_{c}\defi\Tr_{\hat
  c}\underline{\cD}_{2,j}$. Then,
\begin{equation*}
  \int d \nu_\cB  \prod_{a\in\cB}  e^{-4 \Re(\widetilde V^{\les
      2}_{j_a} (\sigmads^{a}))}=\Det_{2}(\rIdirect+4X_{\cB}\widetilde
  Q_{1,\cB})^{-1/2}\,\exp\lbt
  72[\Re(\lambda^{5})]^{2}\Big(\Dvec_{2,\cB},\frac{X_{\cB}}{\rIdirect+4X_{\cB}\widetilde
  Q_{1,\cB}}\Dvec_{2,\cB}\Big)\rbt
\end{equation*}
where $\Dvec_{2,\cB}$ is the vector of vectors such that
$(\Dvec_{2,\cB})_{a}\defi\Dvec_{2,j_{a}}$ for all $a\in\cB$ and $(\ ,\ )$ denotes the
natural scalar product on $V_{\cB}$ inherited from the one on
$\Hopdirect$. Using
\cref{thm-determinants} the determinant prefactor is bounded by
$\exp(\Oun\rho^{2})$. As the norm of $X_{\cB}\widetilde Q_{1,\cB}$ is
bounded above by $\Oun\rho$ and the one of $X_{\cB}$ is not greater
than $\card{\cB}$, we have, for $\rho$ small enough,
\begin{equation*}
  \Big|\Big (\Dvec_{2,\cB},\frac{X_{\cB}}{\rIdirect+4X_{\cB}\widetilde
  Q_{1,\cB}}\Dvec_{2,\cB}\Big)\Big|\les\Oun\card{\cB}\norm{\Dvec_{2,\cB}}^{2}=\Oun\card{\cB}\sum_{a\in\cB}\sum_{c=1}^{4}\Tr_{c}[\big((\Dvec_{2,j_{a}})_{c}\big)^{2}].
\end{equation*}
From the definition of $D_{2}$, see \cref{eq-defABDUR2}, and the bound
on $\Ar_{\cM_{2}}$ (\cref{thm-Aren}), one easily gets
$\norm{\Dvec_{2,\cB}}^{2}\les\Oun$ which implies
\begin{equation*}
  \int d \nu_\cB  \prod_{a\in\cB}  e^{-4 \Re(\widetilde V^{\les
      2}_{j_a} (\sigmads^{a}))}\les e^{\Oun\rho^{2}\card{\cB}}.
\end{equation*}

But by far the lengthiest and most difficult bound is the one for 
$ \int d \nu_\cB  \prod_a  e^{4 \abs{V^{\ges 3}_{j}}}$, which we treat
now. We will actually bound a slightly more general expression.
\begin{thm}\label{thm-GeneralnpBound}
 For all $\cB'\subset\cB$, for all real number $\alpha$, for $\rho $ small
 enough and for any value of the $w$ interpolating parameters, there
 exist positive numbers $\cstK 1_{\alpha}$ and $\cstK 2_{\alpha}$ depending on $\alpha$ such that
\begin{equation*}
\Itrois{\cB'}(\alpha)\defi\int d\nu_{\cB} \prod_{a\in\cB'}
e^{\alpha\abs{V^{\ges 3}_{j_a} (\sigmads^{a})}}   \les
\cstK 1_{\alpha}2^{\card{\cB'}}e^{\cstK 2_{\alpha}\rho^{3/2}\card{\cB'}}.
\end{equation*}
\end{thm}
\begin{cor}\label{BosonicIntegration} 
For $\rho $ small enough and for any value of the $w$ interpolating parameters, if $b \ges 2$
\begin{equation*}
\int d\nu_{\cB} \prod_{a\in\cB}   e^{ 4\abs{V^{\ges 3}_{j_a} (\sigmads^{a})}}   \les
\Oun ^{|\cB|} e^{\Oun  \rho^{3/2}|\cB|}
.
\end{equation*}
\end{cor}
From now on we fix a subset $\cB'$ of $\cB$. For any $j\in S_{\cB'}$ and any integer $p_j \ges 0$ we write
\begin{equation}
e^{\alpha\vert   V^{\ges 3}_{j} \vert }  =  \cP_j + \cR_j , \quad  \cP_j \defi
\sum_{k=0}^{p_j} \frac{ \alpha  \abs{V^{\ges 3}_{j}}^k}{k!} ,   \;\; \cR_j\defi \int_0^1 dt_j (1-t_j)^{p_j } 
\frac{\alpha \vert    V^{\ges 3}_{j}  \vert^{p_j +1}}{p_j !} e^{\alpha t_j \vert    V^{\ges 3}_{j}  \vert }. \label{eq-basictaylor}
\end{equation}
We choose $p_j = M^j$ (assuming $M$ integer for simplicity) and, in $\prod_{a\in\cB'}  e^{\alpha\vert   V^{\ges 3}_{j_a}  \vert }$, we distinguish the set $\cA$ of indices in which we choose the remainder term from its complement $\bar \cA = \cB'\setminus\cA$. The result is:
\begin{multline*}
  \prod_{a\in\cB'} e^{\alpha\vert V^{\ges 3}_{j_a} \vert }=\sum_{\cA
    \subset \cB'}\, \prod_{a\in \cA} \cR_{j_a} \prod_{a\in \bar
    \cA}\cP_{j_{a}} =\sum_{\cA \subset \cB'}\biggl(\prod_{a \in \cA}
  \frac{\alpha^{p_{j_a}+1}}{p_{j_a} !}\biggr)\\
  \times\sum_{\set{k_a\tqs a \in \bar \cA}=0}^{\set{p_{j_a}}} \,
  \biggl( \prod_{a\in \bar \cA} \frac{\alpha^{k_a}}{k_a !}  \biggr) \;
  \cI (\cA, \{ k_a\})
\end{multline*}
with 
\begin{equation*}
\cI (\cA, \{ k_a\}) =  \prod_{a\in \cA}  \biggl( \int_0^1 dt_{j_a} (1-t_{j_a})^{p_{j_a} } \vert   V^{\ges 3}_{j_a}  \vert^{p_{j_a} +1} 
e^{\alpha t_{j_{a}} \vert V_{j_{a}} \vert } \biggr) \prod_{a\in \bar \cA} \vert   V^{\ges 3}_{j_a} \vert^{k_a} .
\end{equation*}
To simplify the notations we put by convention $k_a \defi p_{j_a}+1$ for $a \in \cA$. Remember there is no sum over $k_a$ for such $a\in \cA$.  Hence we write
\begin{equation*}
\prod_{a\in \cB'}  e^{\alpha\vert   V^{\ges 3}_{j_a}  \vert }   
= \sum_{\cA \subset \cB'}\; \sum_{\set{k_a\tqs a \in \bar \cA}=0}^{\set{p_{j_a}}}   \; \bigl(\prod_{a \in \cB'} \frac{\alpha^{k_a}}{k_a !}  \bigr) \bigl( \prod_{a \in \cA} k_a  \bigr)    \; \cI (\cA, \{ k_a\}).
\end{equation*}

Let us fix from now on both the subset $\cA$ and the integers $\{k_a\}_{a\in\bar\cA}$ and bound the remaining integral of $\cI (\cA, \set{ k_a})$ with the measure $ d\nu_{\cB}$.
We bound trivially the $t_{j_a}$ integrals and separate again the perturbative from the non-perturbative terms through a Cauchy-Schwarz inequality:
\begin{equation}
\int d\nu_{\cB} \; \cI (\cA, \{ k_a\}) \les \biggl(\underbrace{\int d\nu_{\cB} \prod_{a \in \cA} e^{2\alpha\vert   V^{\ges 3}_{j_a} \vert }}_{\text{non-perturbative}}\biggr)^{\!\!1/2}
\biggl(\underbrace{\int d\nu_{\cB}  \prod_{a \in \cB'} \vert   V^{\ges 3}_{j_a} \vert^{2k_a}}_{\text{perturbative}}\biggr)^{\!\!1/2}. \label{eq-tobebou}
\end{equation}
Note that the non-perturbative term is $\Itrois{\cA}(2\alpha)$. Thus in order to get the bound
of \cref{thm-GeneralnpBound} on $\Itrois{\cB'}(\alpha)$, we need a
(fortunately cruder) bound on it. This is the object of
\cref{propnonpert}. This bound is actually much worse than in
\cite{Delepouve2014aa}, as it is growing with a power $M^{j_1}$ rather
than logarithmically. But ultimately it will be controlled by the expansion \eqref{eq-basictaylor}.
\begin{prop}\label{propnonpert}
  For all $\cB'\subset\cB$, let $j_{1}$ stand for $\sup_{a \in \cB'} j_a$. For all real number $\alpha$, for $\rho$ small enough and for any value of the $w$ interpolating parameters, there
  exists positive numbers $K$ and $K_{\alpha}$ (the latter depending on $\alpha$ solely)
  such that
\begin{equation*}
\Itrois{\cB'}(\alpha)=\int d\nu_{\cB} \prod_{a\in \cB'} e^{\alpha\abs{  V^{\ges 3}_{j_a}  (\vec\sigma^a)}}  \les K^{\card{\cB'}}e^{K_{\alpha}\rho M^{j_1}}.
\end{equation*}
\end{prop}
\begin{proof}
  We use the quadratic bound of \cref{thm-lemmaquadbound}. Note that
  $\Qj(\sigmad)=\sigmad\scalprod(Q_{0,j}+\cst
  Q{1,j}1)\sigmad\fide\sigmad\scalprod\Qzu[j]\sigmad$. Thus
\begin{equation*}
  \int d\nu_{\cB}   \prod_{a \in \cB'  }   e^{\alpha\abs{   V^{\ges 3}_{j_a} (\sigmads^a)}}  \les
  e^{k\alpha\rho\card\cB'}\int d\nu_{\cB}\, e^{k\alpha\rho \sum_{a \in
      \cB'}  \sigmads^{a}\scalprod \Qzu[j_{a}]\sigmads^{a}}\fide K^{\card\cB'}\int d\nu_{\cB}\, e^{k\alpha\rho(\rsigmads,\Qzu[\cB']\rsigmads)}
\end{equation*}
where $\Qzu[\cB']$ is now a linear operator on
$V_\cB$. Defining $A \defi k\alpha
\rho X_\cB \Qzu[\cB']$, we have
\begin{equation*}
\int d\nu_{\cB}   \, e^{k\alpha\rho(\rsigmads,\Qzu{\cB'}\rsigmads)} = [\det (\rIdirect  - A  )]^{-1/2},
\end{equation*}
and we conclude with \cref{thm-determinants}.
\end{proof}

We turn now to the second (perturbative) factor in \cref{eq-tobebou}, namely
$\int d\nu_{\cB} \prod_{a \in \cB'} \vert   V^{\ges 3}_{j_a} \vert^{2k_a}$.
We replace each $ \vert   V^{\ges 3}_{j_a}  \vert^2$ by its
\emph{quartic} bound (see \cref{thm-eighticbound})
\begin{equation*}
  \abs{  V^{\ges 3}_{j}}^{2}\les\Oun \rho^{3}(M^{-(2-\epsilon)j}+\sum_{i=1}^{4}\sum_{\alpha\in
    A_{i}}\abs{U_{j}^{i,\alpha}}^{2})
\end{equation*}
and Wick-contract the result. It is indexed by graphs of order
$2\sum_{a\in\cA}k_{a}$. More precisely, any such graph has, for all
$a\in\cB$ and all $i\in\set{1,2,3,4}$, $q_{a,i}$ pairs of loop vertices of the $U^i_{j_{a}}$
type (their subindex $\alpha$ will play no further role), and $q_{a,0}$ pairs of constants
$\rho^{10} M^{-(2-\epsilon)j_{a}}$, with 
\begin{equation*}
 q_{a,4}  +  q_{a,3} +  q_{a,2} +  q_{a,1} + q_{a,0}  =   k_a.
\end{equation*}
Let us put 
\begin{equation}
  \begin{alignedat}{2}
    q_{r}&\defi \sum_{a \in \cB} q_{a,r} \text{ for } r\in[4]_{0}\defi\set{0,
    1,\dotsc, 4},&  q&\defi \sum_{a \in \cB} k_a = \sum_{r=0}^4
    q_{r},\\
    \Qadm&\defi\Bigl\{q_{a,r}\in\N, a\in\cA, r\in[4]_{0}\tqs\forall
    a\in\cA, \sum_{r=0}^{4}q_{a,r}=k_{a}\Bigr\},&\qquad \varphi&\defi\sum_{r=0}^{4}2rq_{r}.
  \end{alignedat}\label{eq-notationsqQphi}
\end{equation}
$q=n/2$ is the total number of $\abs{  V^{\ges 3}}^2$ vertices in the second factor
of \cref{eq-tobebou} (and half the order $n$ of our graphs), $\varphi$ is
the number of $\sigma$-fields for a given choice of a sequence
$(q_{a,r})\in\Qadm$. Each Wick-contraction results in a graph $G$ equipped
with a scale attribution $\nu : V(G)\defi\set{\text{loop
    vertices}}\to[\jm]_{0}$ which associates to each (loop) vertex
$a\in\cB$ of $G$ an integer $j_{a}$ reminding us that exactly one of
the propagators $C$ of this vertex $a$ bears a cut-off
$\indic_{j_{a}}$. In the sequel such a contraction will be denoted $G^{\nu}$.

The quartic bound of \cref{thm-eighticbound} having exactly ten terms, developing
a product of $q$ such factors produces $10^q$ terms. The number of graphs obtained
by Wick contracting $2r$ fields is simply $(2 r)!!  \les \Oun ^r
r!$. But if these graphs have uniformly bounded
coordination at each vertex \emph{and} a certain number $t$ of
tadpoles (\textit{i.e.}\@ contractions of fields belonging to the same
vertex), the combinatorics is lower. Indeed the total number of Wick contractions with $2r$ fields
and vertices of maximal degree four leading to graphs with exactly
$t$ tadpoles is certainly bounded by $\Oun ^r  ( r - t )!$. 

Hence using these remarks we find:
\begin{equation}
  \int d\nu_{\cB} \prod_{a\in \cB'} \abs{  V^{\ges 3}_{j_a} }^{2k_a} \les(\Oun \rho^{3})^q
  \sup_{\substack{(q_{a,r})\in\Qadm,\\0\les t\les\varphi/2}} M^{-(2-\epsilon)\sum_{a}q_{a,0}j_{a}}
  (\varphi/2-t)! \sup_{G,\; t(G) =t} A_{G^{\nu}}
  \label{combibou1}
\end{equation}
where the supremum is taken over graphs $G$
with $q_{a,r}$ pairs of loop vertices of length $r$ and highest scale
$j_a$ for all $r\in[4]$ and all $a\in\cA$, and $t$ is the total
number of tadpoles of $G$. In the right-hand side of \cref{combibou1},
the scale attribution $\nu$ is fixed \ie the supremum is
not taken over it. The following \lcnamecref{thm-AGestimate} gives an estimate of $A_{G^{\nu}}$.
\begin{lemma}\label{thm-AGestimate}
  There exists $0<\epsilon\ll 1$ such that any intermediate field graph $G^{\nu}$ of order $n$, made of propagators joining $n_{r,j}$ loop vertices $U^r_j$ of length $r$ with $t_{r,j}$
  loop vertices $U^r_j$ bearing at least a tadpole, for $r$ in $[4]$, obeys the bound
  \begin{align}
    \abs{A_G}&\les \Oun ^n  \prod_{j\in\nu(V(G))} M^{-\frac 12j
        [n_{1,j} +  3n_{2,j} -(1+\epsilon)t_{2,j}+
        3n_{3,j}-t_{3,j} + 3n_{4,j} - t_{4,j} ] } .\label{graphbound34}
  \end{align}
\end{lemma}
\begin{proof}
  As usual such a power counting result is obtained thanks to
  multiscale analysis. Each graph $G^{\nu}$ is already equipped with
  one scale per loop vertex: for all vertex $a\in\cB$ there is exactly
  one $C$-propagator $C_{a}$ of scale $j_{a}=\nu(a)$ (namely in the
  trace represented by that vertex we have the combination
  $C_{a}\indic_{j_{a}}$). We further decompose all remaining
  $C$-propagators ($C\indic_{\les j}$) using
  $\uninfj=\sum_{k=0}^{j}\indic_{k}$. Each graph $G^{\nu}$ is now a
  sum over scale attributions $\mu$ (depending on $\nu$) of graphs
  $G^{\nu}_{\mu}$ which bear one scale per $C$-propagator. We will
  first estimate $A_{G^{\nu}_{\mu}}$ and then sum over $\mu$ to get \cref{graphbound34}.

  The intermediate-field graph $G^{\nu}_{\mu}$ is made of edges, of faces
  $f$ and of loop-vertex corners (in short LVC) $\ell$ which
  correspond to $C$-propagators, hence to the edges of the
  underlying ordinary graph in the standard representation. Each LVC
  $\ell$ has exactly one scale index $j(\ell)$, and we can assume that
  the $r$ LVCs of a loop vertex $v$ of order $r(v)=r$ (in short, a $r$LV) are
  labelled as $\ell_1, \ell_2,\dotsc, \ell_r$ so that
  $j(\ell_1)\fide j_1 \ges j(\ell_2) \fide j_2 \ges\dotsm\ges
  j(\ell_r)\fide j_r$. Each sum over a face
  index costs therefore $\Oun  M^{j_m(f)}$ where $j_m (f)$ is the
  minimum over indices of all the LVCs through which the face
  runs. Hence
  \begin{equation}\label{eq-pertbou}
A_{G^{\nu}_{\mu}} \les \Oun ^n \prod_{\ell} M^{-2j (\ell)} \prod_f
  M^{j_m (f)}.
\end{equation}
This bound is optimal but difficult to analyse. In particular it
depends on the topology of $G$, see
\cite{Ben-Geloun2011aa,Ousmane-Samary2012ab}. In our context of a
\emph{super}-renormalisable model, we can afford to weaken it and
consequently get a new bound which will be factorised over the loop
vertices of $G$. It will have the advantage of depending only on the
types and number of vertices of $G$, thus furnishing also an upper bound
for the $\sup_{G}$ in \cref{combibou1}.

We call a face $f$ local with respect to a loop vertex (hereafter LV)
$v$ if it runs only through
  corners of $v$. The set of faces and local faces of $G$ are
  denoted respectively $F(G)$ and $\Floc(G)$. The complement of
  $\Floc$ in $F$ is $\Fnl$, the set of non-local faces of $G$. Let $f$ be a face of $G$ and $v$ be one of
  the vertices of $G$. If $f$ is incident with $v$, we define $j_{m}^{v}(f)$ as the minimum over indices
  of all the LVCs of $v$ through which the face $f$ runs. Otherwise,
  $j_{m}^{v}(f)\defi 0$. If $f$ is
  non-local then it visits at least two LVs. In that case, we replace
  $j_{m}(f)$ by the bigger factor $\prod_{v\ot f} M^{j_m^{v}(f)/2}$ where the
  product runs over the vertices incident with $f$:
  \begin{align*}
    A_{G^{\nu}_{\mu}} &\les \Oun ^n
  \prod_{v \in V(G)}\prod_{i=1}^{r(v)}
  M^{-2j_i}\prod_{f\in\Floc(G)}M^{j_{m}(f)}
  \prod_{f\in\Fnl(G)}\prod_{v\ot f}M^{j^v_m (f)/2}\\
  &=\Oun ^n
  \prod_{v \in V(G)}\Bigl(\underbrace{\prod_{i=1}^{r(v)}
  M^{-2j_i}\prod_{\substack{f\in\Floc(G),\\f\to v}}M^{j_{m}^{v}(f)}
  \prod_{\substack{f\in\Fnl(G)\\f\to v}}M^{j^v_m (f)/2}}_{\fide W(v)}\Bigr).
\end{align*}
Our bound is now factorised over the loop vertices of $G$ and we can
simply bound the contribution $W(v)$ of each vertex $v$ according to its type.\\

  Consider a 3LV; it can be of type $c^{3}$, $c_{1}^{2}c_{2}$ or
  $c_{1}c_{2}c_{3}$, depending on whether the three lines hooked to it have the
  same colour $c$, two different colours $c_{1}, c_{2}$ or three different
  colours $c_{1},c_{2},c_{3}$, see \cref{f-U3}.
  \begin{figure}[!htp]
    \centering
    \begin{subfigure}[b]{.3\linewidth}
      \centering
      \includegraphics[scale=.8]{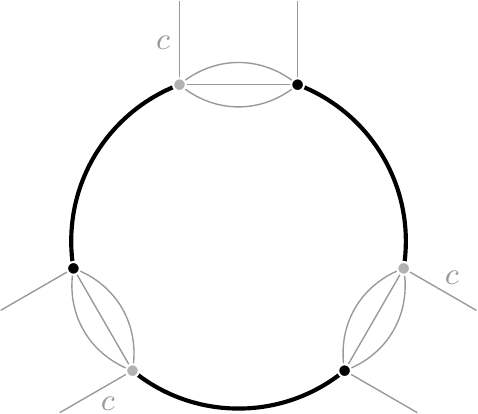}
      \caption{The $c^{3}$-case}\label{f-U3c3}
    \end{subfigure}\hfill
    \begin{subfigure}[b]{.3\linewidth}
      \centering
     \includegraphics[scale=.8]{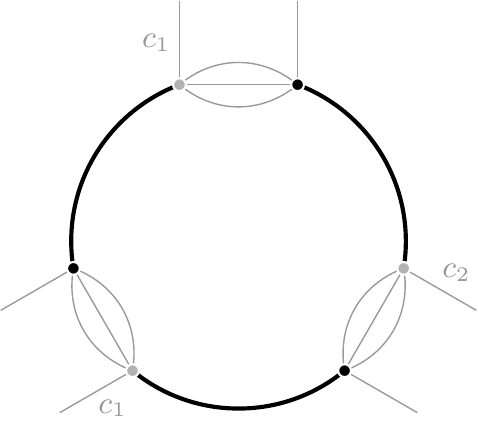}
     \caption{The $c_{1}^{2}c_{2}$-case}\label{f-U3c2c}
    \end{subfigure}\hfill
    \begin{subfigure}[b]{.3\linewidth}
      \centering
      \includegraphics[scale=.8]{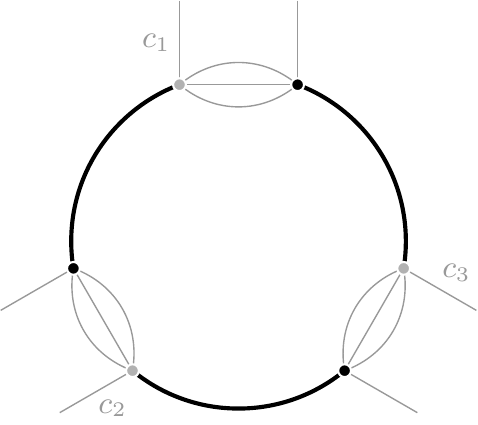}
      \caption{The $c_{1}c_{2}c_{3}$-case}\label{f-U3ccc}
    \end{subfigure}
    \caption{The three coloured versions of a $U^3$-loop vertex.}
    \label{f-U3}
  \end{figure}
 Only in the two first cases can it have a tadpole,
  and then one local face incident with a single LVC \ie of length one. Hence:
  \begin{itemize}
  \item In case $c^{3}$, the three faces of length $3$ and colour
    $c' \neq c$ are local, see \cref{f-U3c3local}, and their total cost is $M^{3j_3}$. In
    case there is a tadpole (of colour $c$ and LVC $t \in \{ 1,2,3\}$),
    its local face, see \cref{f-U3c3tadloc}, costs $M^{j_t}$ and the
    other (non-local) face of colour $c$, see \cref{f-U3c3tadnl}, costs at
    most $\inf_{t' \neq t} M^{j_{t'}/2}$. The worst case is when
    $t =1$, in which case the total cost of colour $c$ faces is
    $M^{j_1 + j_3/2}$. In case there is no tadpole, the faces of colour
    $c$ are non-local. There are at most three of them, so their cost
    is at worst $M^{j_1/2 + j_2/2 + j_3/2}$. The worst case is
    therefore the tadpole case with $t=1$, where the total face cost
    is $M^{j_1 + 7j_3/2}$.  Joining to the $ M^{-2(j_1 + j_2 + j_3)}$
    factor the vertex weight $W (v) $ is therefore bounded in the
    $c^{3}$ case by $M^{- j_1 - j_2/2 - 3(j_2 -j_3)/2 }$.

  \item In case $c_{1}^{2}c_{2}$, the two local faces of length three
    (and colour $c\neq c_{1},c_{2}$) cost $M^{2j_3}$ and the non-local
    face of colour $c_{2}$, see \cref{f-U3c2cnl}, costs $M^{j_3/2}$. In case there is a tadpole (of colour $c_{1}$ and LVC
    $t \in \{ 1,2,3\}$), its face costs $M^{j_t}$ and the other
    local face of colour $c_{1}$ (and length 2) costs
    $\inf_{t' \ne t} M^{j_{t'}}$; in case there is no tadpole, the
    single or the two non-local faces of colour $c_{1}$ cost at most
    $M^{j_1/2 + j_3/2}$. The worst case is therefore again the tadpole
    case with $t=1$, where the total face cost is again
    $M^{j_1 + 7j_3/2}$, and the vertex weight $W (v) $ is therefore
    again bounded in the $c_{1}^{2}c_{2}$ case by
    $M^{- j_1 - j_2/2 - 3(j_2 -j_3)/2 }$.

  \item Finally the case $c_{1}c_{2}c_{3}$ is simpler as there can be no tadpole.
    The three non-local faces cost in total $M^{3j_3/2}$, the local
    face costs $M^{j_3}$, and the vertex weight $W (v) $ is therefore
    bounded by the better factor
    $M^{- 2j_1 - 3 j_2/2 - (j_2 -j_3)/2 }$.
  \end{itemize}
  \begin{figure}[!htp]
    \centering
    \begin{subfigure}[b]{.45\linewidth}
      \centering
      \includegraphics[scale=.8]{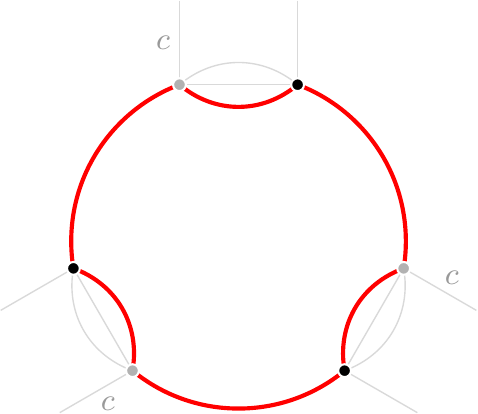}
      \caption{A local face of length $3$}\label{f-U3c3local}
    \end{subfigure}
    \begin{subfigure}[b]{.45\linewidth}
      \centering
      \includegraphics[scale=.8]{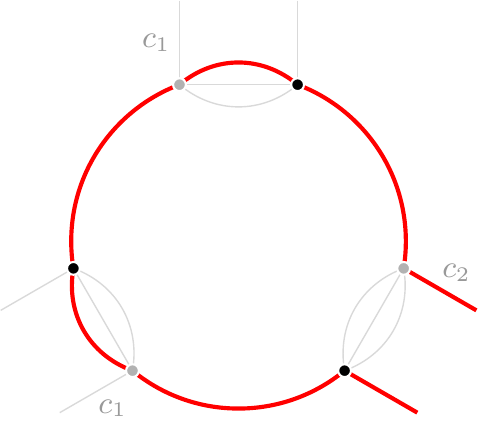}
      \caption{A non-local face of colour $c_{2}$}\label{f-U3c2cnl}
    \end{subfigure}\\
    \begin{subfigure}[c]{.45\linewidth}
      \centering
      \includegraphics[scale=.8]{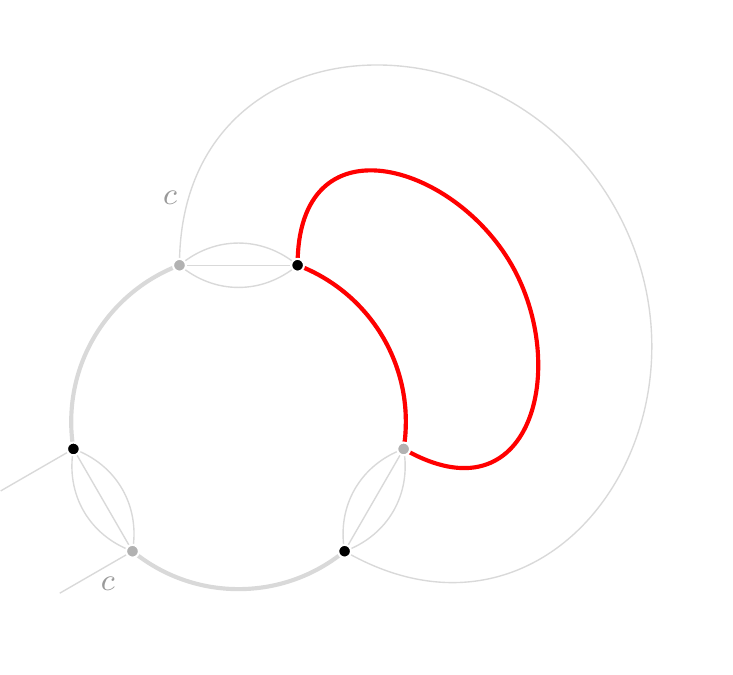}
      \caption{A local face of length $1$}\label{f-U3c3tadloc}
    \end{subfigure}
    \begin{subfigure}[c]{.45\linewidth}
      \centering
      \includegraphics[scale=.8]{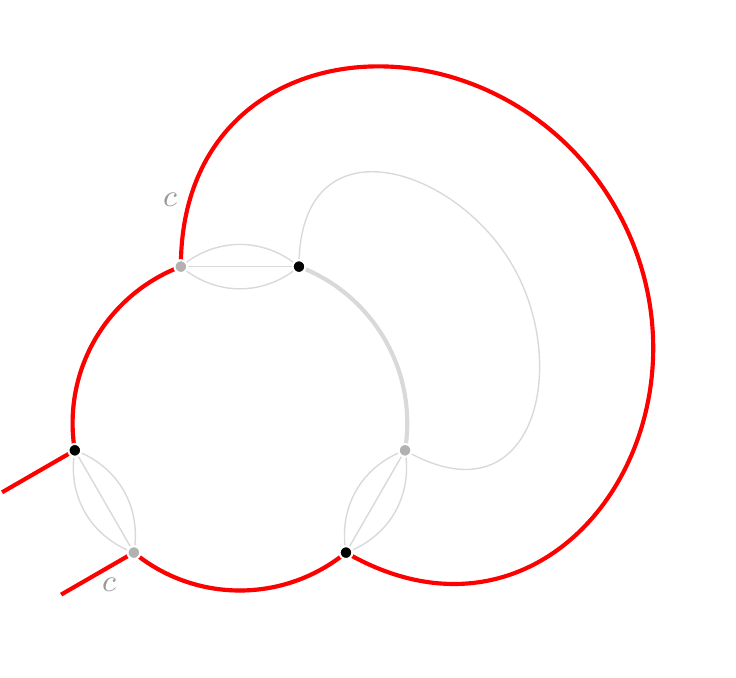}
      \caption{A non-local face of length $>3$}\label{f-U3c3tadnl}
    \end{subfigure}
    \caption{Some faces of a $U^3$-vertex}
    \label{f-U3faces}
  \end{figure}

  The same analysis can be repeated for 4LV's. As it is somewhat
  tedious, we postpone it to \cref{sec-quartic-loop-vertex}. There, it
  can be checked that the worst total face cost is:
  \begin{itemize}
  \item with two tadpoles, $M^{j_{1}+j_{2}+4j_{4}}$,
  \item with one tadpole, $M^{j_{1}+j_{2}/2+7j_{4}/2}$,
  \item without tadpole, $M^{(j_{1}+j_{2}+j_{3}+7j_{4})/2}$.
  \end{itemize}
  The vertex weight $W (v)$ is therefore, when tadpole(s) are
  present, at worst $M^{- j_1 - j_2 - 2(j_3- j_4 )}$, and when
  they are not $M^{- 3j_1/2 - 3 j_2/2 - 3 (j_3-j_4)/2 }$. The worst
  total face costs for loop vertices of degree one and two are
  available in \cref{sec-faces-loop-vertices}.\\

  With a bound on $A_{G^{\nu}_{\mu}}$, there remains to sum over $\mu$
  to get \cref{graphbound34}. We decompose this sum into two parts:
  first a sum over the relative positions of $j_{2},\dotsc,j_{r}$ at
  all vertices of degree $r\ges 2$. This costs at worst $3!^{n}$. Then
  a sum over $j_{2}\ges\dotsm\ges j_{r}$ at each loop vertex. The
  analysis above has shown that this is convergent and leads to the
  bound \eqref{graphbound34} and thus to \cref{thm-AGestimate}.
\end{proof}

Coming back to the notations of \cref{eq-notationsqQphi,combibou1} and
remembering that the $q_{a,r}$'s are meant for \emph{pairs} of vertices,
\begin{equation*}
  \sup_{G,\; t(G) =t} A_{G^{\nu}}\les \Oun ^{n}  \prod_{a \in \cB}
  M^{- j_a
    [q_{a,1}+3q_{a,2}-(\frac12+\epsilon)t_{a,2}+3q_{a,3}-\frac12 t_{a,3}+3q_{a,4}-\frac12 t_{a,4}]} ,
\end{equation*}
where $t_{a,r}\defi t_{r,j_{a}}$, $r=2,3,4$, is the total number of
vertices of length $r$ and scale $j_{a}$ in $G$ which bear at least
one tadpole. We put $\tau_{a,r} = t_{a,r}/2$ and $\tau_r = \sum_a
\tau_{a,r}/2$. In \cref{combibou1}  we remark that $t=\sum_{r\ges
  2}\sum_a t_{a,r}  = 2\sum_{r} \tau_r$. Since
$q_{1}+q_{2}+q_{3}+q_{4}\les q$, the factor
$(\varphi/2-t)!=(\sum_{r=1}^4 rq_{r}  -t)!$ in \cref{combibou1} is
bounded by $\Oun ^{q}\prod_{r}  (q_r !)^r (\tau_r !)^{-2}$ (we
put $\tau_{1}=0$ and interpret $n!$ for  $n$ not integer as $\Gamma (n)$).
Hence the perturbative factor of \cref{eq-tobebou} obeys ($\tau_{1}=0$)
\begin{multline*}
  \biggl(\int d\nu_{\cB} \prod_{a \in \cB'} \vert   V^{\ges 3}_{j_a} \vert^{2k_a}
  \biggr)^{\!\!1/2} \les
  (\Oun \rho^{3/2})^{q}\sup_{\substack{(q_{a,r})\in\Qadm,\\0\les\tau_{a,r}\les
      q_{a,r}}}\Bigl(\prod_{r=1}^{4}(q_{r}!)^{r/2}(\tau_{r}!)^{-1}\Bigr)\\
  \times\prod_{a\in\cB}M^{-\frac12j_{a}[(2-\epsilon)q_{a,0}+q_{a,1}+\sum_{r=2}^{4}(3q_{a,r}-\tau_{a,r})-\epsilon\tau_{a,2}]}. 
\end{multline*}
Joining this last estimate with \cref{propnonpert}, the
term to be bounded in \cref{thm-GeneralnpBound} obeys
\begin{multline*}
  \int d\nu_{\cB} \prod_{a\in\cB'} e^{\alpha\abs{ V^{\ges 3}_{j_a}
      (\sigmads^{a})}}\les \sum_{\cA \subset
    \cB'}K^{\card\cA}\,e^{\cstK{1}_{\alpha}\rho M^{j_1} } \,
  \sum_{\set{k_a , a \in \bar \cA}=0}^{\set{p_{j_a}}}\, (\Oun \rho^{3/2})^{q} \; \Bigl(\prod_{a \in \cB'}\frac{\alpha^{k_a}}{k_a !}  \Bigr) (\prod_{a \in \cA} k_a \bigr)\\
  \sup_{\substack{(q_{a,r})\in\Qadm,\\0\les\tau_{a,r}\les
      q_{a,r}}}\Bigl(\prod_{r=1}^{4}(q_{r}!)^{r/2}(\tau_{r}!)^{-1}\Bigr)\prod_{a\in\cB}M^{-\frac12j_{a}[(2-\epsilon)q_{a,0}+q_{a,1}+\sum_{r=2}^{4}(3q_{a,r}-\tau_{a,r})-\epsilon\tau_{a,2}]}
\end{multline*}
where again $j_{1}=\sup_{a\in\cA}j_{a}$. Note that we use, and will go on using, the symbols $K$, $K_{\alpha}$,
$\cstK{1}_{\alpha}$, $\cstK{2}_{\alpha}$ etc essentially the same way as we
do with $\Oun $ \ie to denote generic constants possibly depending on
$\alpha$. In the rest of this proof, our strategy will be to use the
power counting namely the powers of $M^{-j_{a}}$ to compensate both for
the large number of Wick contractions (the $q_{r}!$'s) and for the
crude bound of \cref{propnonpert}.\\

As $\tau_{r}=\sum_{a}\tau_{a,r}$,
$(\tau_{r}!)^{-1}\les\prod_{a}(\tau_{a,r}!)^{-1}$. Similarly, since $k_a= \sum_{r}q_{a,r}$, $(k_a ! )^{-1}  \les\prod_{r}(q_{a,r} !)^{-1}$. Moreover we remark
that $ \prod_{a \in \cB'}\alpha^{k_a}  \prod_{a \in \cA} k_a  \les
(\sup\set{2,\alpha})^q$. Hence
\begin{multline}
  \int d\nu_{\cB} \prod_{a\in\cB'} e^{\alpha\abs{  V^{\ges 3}_{j_a} }}\les
  \sum_{\cA \subset \cB'}K^{\card\cA}\,e^{\cstK 1_{\alpha}\rho M^{j_1} } \,
  \sum_{\set{k_a , a \in \bar \cA}=0}^{\set{p_{j_a}}}\,
  (\cstK 2_{\alpha}\rho^{3/2})^{q}
  \sup_{\substack{(q_{a,r})\in\Qadm,\\0\les\tau_{a,r}\les  q_{a,r}}}\\
  \prod_{r=1}^{4}\Bigl((q_{r}!)^{r/2}\prod_{a\in\cB'}(q_{a,r}!\,\tau_{a,r}!)^{-1}\Bigr)
  \prod_{a\in\cB'}M^{-\frac12j_{a}[(2-\epsilon)q_{a,0}+q_{a,1}+\sum_{r=2}^{4}(3q_{a,r}-\tau_{a,r})-\epsilon\tau_{a,2}]}.\label{eq-tbb1}
\end{multline}
For $r=2,3,4$ we remark that if $\tau_{a,r}\les
q_{a,r}/2$, we have
\begin{equation*}
  (\tau_{a,r}!)^{-1}M^{-\frac 12j_{a}(3q_{a,r}-\tau_{a,r})}\les
  M^{-\frac 54j_{a}q_{a,r}},
\end{equation*}
and if $\tau_{a,r}\ges q_{a,r}/2$ (and of course $\tau_{a,r}\les q_{a,r}$),
\begin{equation*}
    (\tau_{a,r}!)^{-1}M^{-\frac 12j_{a}(3q_{a,r}-\tau_{a,r})}\les
  2^{q_{a,r}}(q_{a,r}!)^{-1/2} M^{-j_{a}q_{a,r}}.
\end{equation*}
In the sequel we will use the following simple bound several times: for any
$\eta\in\R^{*}_{+}$,
\begin{equation}
  M^{-\eta j_{a}q_{a,r}}\les K^{\eta
    q_{a,r}}(q_{a,r}!)^{-\eta},\label{eq-pwtofactorial}
\end{equation}
This is an easy consequence of $q_{a,r}\les k_a
\les M^{j_a +1}$. Thus, using \cref{eq-pwtofactorial} with $\eta=1/4$,
we have that for all $\tau_{a,r}$,
\begin{equation*}
  (\tau_{a,r}!)^{-1}M^{-\frac 12j_{a}(3q_{a,r}-\tau_{a,r})}\les\Oun ^{q_{a,r}}(q_{a,r}!)^{-1/4}M^{-j_{a}q_{a,r}}.
\end{equation*}
Using $\tau_{a,2}\les q_{a,2}$, \cref{eq-tbb1} then becomes
\begin{multline*}
  \int d\nu_{\cB} \prod_{a\in\cB'} e^{\alpha\abs{  V^{\ges 3}_{j_a} }}\les
  \sum_{\cA \subset \cB'}K^{\card\cA}\,e^{\cstK 1_{\alpha}\rho M^{j_1} } \,
  \sum_{\set{k_a , a \in \bar \cA}=0}^{\set{p_{j_a}}}\,
  (\cstK 2_{\alpha}\rho^{3/2})^{q}
  \sup_{\substack{(q_{a,r})\in\Qadm}}(q_{1}!)^{1/2}\prod_{a\in\cB'}(q_{a,1}!)^{-1}\\
  \prod_{r=2}^{4}\Bigl((q_{r}!)^{r/2}\prod_{a\in\cB'}(q_{a,r}!)^{-5/4}\Bigr)\prod_{a\in\cB'}M^{-j_{a}[(1-\epsilon)q_{a,0}+\frac
    12q_{a,1}+(1-\epsilon)q_{a,2}+q_{a,3}+q_{a,4}]}.
\end{multline*}

\paragraph{Crude bound versus power counting}
We can now take care of the $e^{\cstK 1_{\alpha}\rho M^{j_{1}}}$ factor by
using a part of the power counting. Let $\eta$ be a real positive
number. Remembering that for all $a\in\cA$, $k_{a}=M^{j_{a}+1},$
\begin{align*}
  \prod_{a\in\cA}M^{-\eta
    j_{a}\sum_{r=0}^{4}q_{a,r}}&=\prod_{a\in\cA}M^{-\eta
    j_{a}k_{a}}\les \prod_{a\in\cA}M^{-\eta j_{a}M^{j_{a}}}\les M^{-\eta j_{1}M^{j_{1}}}.
\end{align*}
But
\begin{equation*}
  e^{\cstK 1_{\alpha}\rho M^{j_{1}}}\prod_{a\in\cA}M^{-\eta
    j_{a}\sum_{r=0}^{4}q_{a,r}}\les e^{\cstK 1_{\alpha}\rho M^{j_{1}}}
  M^{-\eta j_{1}M^{j_{1}}}\les K_{\alpha,\eta}
\end{equation*}
so that
\begin{multline}
  \int d\nu_{\cB} \prod_{a\in\cB'} e^{\alpha\abs{  V^{\ges 3}_{j_a} }}\les
  K_{\alpha,\eta}\sum_{\cA \subset \cB'}K^{\card\cA}\,
  \sum_{\set{k_a , a \in \bar \cA}=0}^{\set{p_{j_a}}}\,
  (\cstK 2_{\alpha}\rho^{3/2})^{q}
  \sup_{\substack{(q_{a,r})\in\Qadm}}(q_{1}!)^{1/2}\prod_{a\in\cB'}(q_{a,1}!)^{-1}\\
  \prod_{r=2}^{4}\Bigl((q_{r}!)^{r/2}\prod_{a\in\cB'}(q_{a,r}!)^{-5/4}\Bigr)\prod_{a\in\cB'}M^{-j_{a}[(1-\epsilon)q_{a,0}+\frac
    12q_{a,1}+(1-\epsilon)q_{a,2}+q_{a,3}+q_{a,4}-\eta k_{a}]}.\label{eq-tbb3}
\end{multline}

\paragraph{Combinatorics versus power counting}
In order to beat the $q_{r}!$'s, we need to boost the powers of some
of the $q_{a,r}!$'s. We use \cref{eq-pwtofactorial} for the couples
$(r,\eta)$ equal to $(3,1/4)$ and $(4,3/4)$. \Cref{eq-tbb3} becomes
\begin{multline*}
  \int d\nu_{\cB} \prod_{a\in\cB'} e^{\alpha\abs{  V^{\ges 3}_{j_a} }}\les
  K_{\alpha,\eta}\sum_{\cA \subset \cB'}K^{\card\cA}\,
  \sum_{\set{k_a , a \in \bar \cA}=0}^{\set{p_{j_a}}}\,
  (\cstK 2_{\alpha}\rho^{3/2})^{q}
  \sup_{\substack{(q_{a,r})\in\Qadm}}\\
  \prod_{r=1}^{4}\Bigl((q_{r}!)^{r/2}\prod_{a\in\cB'}(q_{a,r}!)^{-r/2}\Bigr)\prod_{a\in\cB'}M^{-j_{a}[(1-\epsilon)q_{a,0}+\frac
    12q_{a,1}+(1-\epsilon)q_{a,2}+\frac 34q_{a,3}+\frac 14q_{a,4}-\eta k_{a}]}.
\end{multline*}
Then for $\epsilon\les 3/4$ and $\eta<1/4$,
\begin{multline*}
  \int d\nu_{\cB} \prod_{a\in\cB'} e^{\alpha\abs{  V^{\ges 3}_{j_a} }}\les
  K_{\alpha,\eta}\sum_{\cA \subset \cB'}K^{\card\cA}\,
  \sum_{\set{k_a , a \in \bar \cA}=0}^{\set{p_{j_a}}}\,
  (\cstK 2_{\alpha}\rho^{3/2})^{q}
  \sup_{\substack{(q_{a,r})\in\Qadm}}\\
  \prod_{r=1}^{4}\Bigl(q_{r}!\prod_{a\in\cB'}(q_{a,r}!)^{-1}M^{-\frac
    2r(\frac
    14-\eta)j_{a}q_{a,r}}\Bigr)^{\! r/2}.
\end{multline*}
Now we remark that for all $r$, by the multinomial theorem,
$q_{r}!\prod_{a\in\cB'}(q_{a,r}!)^{-1}M^{-\frac 2r(\frac
  14-\eta)j_{a}q_{a,r}}$ is one of the terms in the multinomial
expansion of $(\sum_{a \in \cB'} M^{-\frac
  2r(\frac14-\eta)j_{a}})^{q_r} $. Since the 
$j_a$'s are all distinct, $(q_{r}!\prod_{a\in\cB'}(q_{a,r}!)^{-1}M^{-\frac 2r(\frac
  14-\eta)j_{a}q_{a,r}})^{r/2}\les (\sum_{j\ges 0}M^{-\frac
  2r(\frac14-\eta)j})^{q_{r}r/2}=(K_{r,\eta})^{q_{r}}\les\Oun ^{q}$.
Hence
\begin{equation*}
    \int d\nu_{\cB} \prod_{a\in\cB'} e^{\alpha\abs{  V^{\ges 3}_{j_a} }}\les
  K_{\alpha,\eta}\sum_{\cA \subset \cB'}K^{\card\cA}\,
  \sum_{\set{k_a , a \in \bar \cA}=0}^{\set{p_{j_a}}}\,
  (\cstK 2_{\alpha}\rho^{3/2})^{q}.
\end{equation*}
Let $q_{\cA}$ denote $\sum_{a\in\cA}k_{a}$. Then we have
\begin{align*}
    \int d\nu_{\cB} \prod_{a\in\cB'} e^{\alpha\abs{  V^{\ges 3}_{j_a} }}&\les
  K_{\alpha,\eta}\sum_{\cA \subset \cB'}K^{\card\cA}(\cstK 2_{\alpha}\rho^{3/2})^{q_{\cA}}\,
  \prod_{a\in\bar\cA}\,\sum_{k_a =0}^{p_{j_a}}\,(\cstK
  2_{\alpha}\rho^{3/2})^{k_{a}}\\
  &\les K_{\alpha,\eta}\sum_{\cA \subset \cB'}K^{\card\cA}(\cstK 2_{\alpha}\rho^{3/2})^{q_{\cA}}\,
  \prod_{a\in\bar\cA}\frac{1-(\cstK
    2_{\alpha}\rho^{3/2})^{p_{j_{a}}+1}}{1-\cstK
    2_{\alpha}\rho^{3/2}}\\
  &\les K_{\alpha,\eta}\sum_{\cA \subset \cB'}K^{\card\cA}(\cstK
  2_{\alpha}\rho^{3/2})^{q_{\cA}}2^{\card{\bar\cA}}&&\text{for $\rho$
    small enough}\\
  &\les K_{\alpha,\eta}\sum_{\cA \subset \cB'}K^{\card\cA}(\cstK
  2_{\alpha}\rho^{3/2})^{\card{\cA}}2^{\card{\bar\cA}}&&k_{a}\ges 1,\
  a\in\cA\\
  &\les K_{\alpha,\eta}(2+K_{\alpha}\rho^{3/2})^{\card{\cB'}}\\
  &\les K_{\alpha,\eta}2^{\card{\cB'}}e^{K_{\alpha}\rho^{3/2}\card{\cB'}}.
\end{align*}
This completes the proof  of \cref{thm-GeneralnpBound}.\hfill$\square$\\

To conclude this section, let us briefly comment on the case of a block $\cB$
with a single node ($\card\cB =1$) in \cref{BosonicIntegration}. The
proof of the single node case in \cref{thm-npBound} is very similar, even easier, than the proof of
\cref{thm-GeneralnpBound} but we need to remember that there is no term
with $k=0$ vertices, because we are dealing with $e^{-  V^{\ges 3}_{j_a}  (  \vec \sigma^a ) } -1$ rather than
$e^{-  V^{\ges 3}_{j_a}  (  \vec \sigma^a ) }$.

\section{Perturbative functional integral bounds}
\label{sec-pert-funct-integr}

We still have to bound the fourth ``perturbative'' factor in
\cref{eq-CS-Pert-NonPert}, namely
\begin{equation*}
  I_{4}=\Bigl( \int d \nu_\cB\,\vert  \wo{A_{\skG} (\sigmad)}  \vert^4\Bigr)^{\!1/4}.
\end{equation*}
It is not fully perturbative though because of the resolvents still
present in $A_{\skG}$. If $|\cB|=1$, we recall that the graphs $\skG$
are either one-vertex maps or one-edge trees. For $|\cB|\ges 2$, they
are forests with $e(\skG)= |\cB|-1$ (coloured) edges joining $n(\skG) = c(\skG) + e(\skG)$ (effective) vertices, each of which has a weight given by 
\cref{eq-DerivSigmak1,eq-DerivSigmak2,eq-developcycles}. The number of
connected components $c(G)$ is bounded by $ |\cB|-1$, hence
$n(\skG)\les 2(\abs{\cB}-1)$, see \cref{eq-foreboun}. $I_{4}^{4}$ can be reexpressed as $ \int d\nu_\cB\,
A_{\skG''}(\sigmad)$ where $\skG''$ is the (disjoint) union of two
copies of the graph $\skG$ and two copies of its mirror conjugate graph $\skG'$ of
identical structure but on which each operator has been replaced by
its Hermitian conjugate. This overall graph $\skG''$ has thus four times as many vertices, edges, resolvents,  $\sigmad^a$ insertions and connected components than the initial graph $\skG$.\\

\subsection{Contraction process}
\label{sec-contraction-process}

To evaluate the amplitude $A_{\skG"}= \int d \nu_\cB \vert  \wo{A_{\skG}
  (\sigmad)}  \vert^4$, we first replace any isolated vertex of type
$V_{j}^{\ges 3}$ by its quartic bound, \cref{lemmaquarticbound}, and
then contract every $\sigmad^a$ insertion,
which means using repeatedly integration by parts until there are no $\sigmad^a$ numerators left, thanks to the formula
\begin{equation}\label{eq-intbyparts}
\int (\sigmad^a)_{c;mn} F(\sigmad)\, d\nu(\sigmad)=-\sum_{k,l}\int
\delta_{ml}\delta_{nk}\frac{\partial F(\sigmad)}{\partial (\sigmad^{a})_{c;kl}}\, d\nu(\sigmad),
\end{equation}
where $d\nu (\sigmad)$ is the standard Gaussian measure of covariance
$\Idirect$. We call this procedure the contraction process. The derivatives $\frac{\partial}{\partial(\sigmads^a)_{c}}$ will act on any resolvent $\fres_{\les j_a}$ or any remaining $\sigmad^a$ insertion of $\skG"$,
creating a new contraction edge\footnote{The combinatorics for these contractions will be paid by the small factors earned from the explicit
$j$-th scale propagators, see \cref{sec-final-sums}.}. 
When such a derivative acts on a resolvent,
\begin{equation}\label{eq-DerivationOfSigma}
\partial_{\sigma_s} \fres^{(\dagger)}_{\lj} =\fres^{(\dagger)}_{\lj}\dU{\lj}\fres^{(\dagger)}_{\lj},
\end{equation}
it creates two new corners representing $\sqrt C_{\les j_a} \fres_{\les j_a} \sqrt C_{\les j_a}$ or
$\sqrt C_{\les j_a}\fres_{\les j_a}^{\dagger} \sqrt C_{\les j_a}$ product
of operators. Remark that at the end of this process we have therefore obtained a
sum over new \emph{resolvent graphs} $\resG$, the amplitudes of which
no longer contain any $\sigmad^a$ insertion. Nevertheless the number
of edges, resolvents and connected components at the end of this
contraction process typically has changed. However we have a bound on the number of new edges generated by the contraction process. Since
each vertex of $\skG$ contains at most three $\sigmad^a$
insertions\footnote{We focus here on Bosonic blocks with more more
  than one vertex. The case of isolated vertices will only lead to
  $\Oun^{|\cB|}$ combinatorial factors which will be easily
  compensated by powers of the coupling constant $g$.}
, $\skG"$ contains at most $12 n(\skG)$, hence using
\cref{eq-foreboun} at most $24(|\cB|-1)$ insertions to contract. Each
such contraction creates at most one new edge. Therefore each
resolvent graph $\resG$ contains the initial $4(|\cB|-1)$ coloured edges of $\skG"$ decorated with up to at most $24(|\cB|-1)$ additional new edges.\\

Until now, the amplitude $A_{\resG}$ contains $\sqrt C_{\les
  j}=\sum_{j'< j}\sqrt C_{j'}+t_j \sqrt C_j$ operators. We now develop
the product of all such factors as a sum over scale
assignments $\mu$, as in \cite{Riv1}. It means that each former $\sqrt
C_{\les j}$ is replaced by a fixed scale $\sqrt C_{j'}$  operator with
scale attribution $j'\les j$ (the $t_j$ factor being bounded by $1$). The amplitude
at fixed scale attribution $\mu$ is noted $A_{\resG_{\mu}}$. The sum over $\mu$ will be standard to bound
after the key estimate of \cref{thm-PowCountSpare} is established. Similarly the sums over $\skG$ and over $\resG$ only generate
a finite power of $\vert \cB \vert !$, hence will be no problem using the huge
decay factors of \cref{thm-BoundI4}, see \cref{sec-final-sums}.\\

We shall now bound each amplitude $A_{\resG_{\mu}}$. Were it not for
the presence of resolvents, the graph $\resG$, which is convergent,
would certainly obey the standard bound on convergent amplitudes in
super-renormalisable theories. A precise statement can be found in \cref{thm-PowCountSpare}. The only problem is therefore to get rid of these
resolvents, using that their norm is bounded by a constant in the
cardioid domain. This can be done through the technique of \icst
bounds or \ics, introduced for the first time in a similar tensor field theoretic 
context in \cite{Magnen2009ab}.

\subsection{Iterated Cauchy-Schwarz estimates}
\label{sec-iter-cauchy-schw}

Let us first give a crude description of the steps necessary to bound
the amplitude of a (connected) graph $\resG$ by a product of amplitudes freed of
resolvents.

\subsubsection{ICS algorithm 1.0b}
  Let $\resG$ be a connected graph in the \ifrt obtained after the contraction
  process \ie a connected component of a resolvent graph. The following steps constitute the core of the \ics method:
  \begin{enumerate}
  \item\label{item-StepSingleTrace} Write the amplitude $A_{\resG}$ of $\resG$ as a single trace over $L(\Htens)$ times a product of
    Kronecker deltas. This trace contains some resolvents.
  \item\label{item-StepScalProd} Write $A_{\resG}$ as a scalar product
    of the form $\scalprodtens{\alpha}{(\fres\otimes S\otimes\Itens)\beta}$ or
    $\scalprodtens{\alpha}{(\fres\otimes S\otimes\fres^{\transpose})\beta}$ where
    $\alpha$ and $\beta$ are vectors of an inner product space and $S$
    is a permutation operator.
  \item\label{item-StepApplyCS} Apply Cauchy-Schwarz inequality to the
    previous expression to get
    \begin{equation*}
|A_{\resG}|\les
    \norm{\fres}^{(2)}\sqrt{\normtenssq{\alpha}}\sqrt{\normtenssq{\beta}}.
  \end{equation*}

  \item\label{item-StepIterate} Notice that $\normtenssq{\alpha}$ and
    $\normtenssq{\beta}$ are also amplitudes of some graphs. If they
    still contain some resolvents, iterate the process by going back
    to step \ref{item-StepSingleTrace}.
  \end{enumerate}
In the rest of this section, we give a bound on the number of
iterations of this algorithm before it stops. We also refine it
in order to avoid pathological situations. But before that, to give
the reader a more concrete idea of the method, we illustrate it now
with examples. It will be the occasion to go through all steps of the \icst
method, and understand why the rough algorithm given above needs to be modified.

\subsubsection{Concrete examples}
\label{sec-concrete-examples}

Let us consider the convergent graph $\resG$ of
\cref{f-cvGraphExample}, in \ifrt, obtained after the contraction
process. \latin{Stricto sensu} it represents a sum of different amplitudes. As any spanning tree of it contains a single edge, the
possible vertices associated to this graph can be found in
\cref{eq-DerivSigmak1}. Let us choose to study the following
expression
\begin{multline}
  \label{eq-cvGraphAmplExample}
  A_{\resG}=\Big(\prod_{i=1}^{3}\sum_{m_{i},n_{i},m'_{i},n'_{i}\in\Z}\Big)\Tr\big[(\be^{c_{1}}_{m_{1}n_{1}}\otimes\Itens_{\hat
    c_{1}})C(\be^{c_{2}}_{m_{2}n_{2}}\otimes\Itens_{\hat c_{2}})C]\\
  \times\Tr[\sqrt C(\be^{c_{2}}_{m'_{2}n'_{2}}\otimes\Itens_{\hat
    c_{2}})C(\be^{c_{3}}_{m_{3}n_{3}}\otimes\Itens_{\hat
    c_{3}})C(\be^{c_{3}}_{m'_{3}n'_{3}}\otimes\Itens_{\hat
    c_{3}})\sqrt C\fres\sqrt
  C(\be^{c_{1}}_{m'_{1}n'_{1}}\otimes\Itens_{\hat c_{1}})\sqrt C\fres\big]\\
  \times\delta_{m_{1}n'_{1}}\delta_{n_{1}m'_{1}}\delta_{m_{2}n'_{2}}\delta_{n_{2}m'_{2}}\delta_{m_{3}n'_{3}}\delta_{n_{3}m'_{3}}.
\end{multline}
\begin{figure}[!htp]
  \centering
  \includegraphics[scale=1.3]{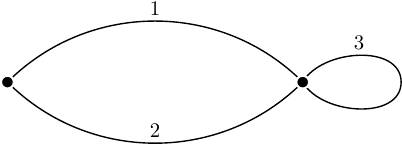}
  \caption{A convergent graph with resolvents}
  \label{f-cvGraphExample}
\end{figure}

\noindent
Vertices of $\resG$ correspond to traces and edges to pairs of
Kronecker deltas, \eg $\delta_{m_{1}n'_{1}}\delta_{n_{1}m'_{1}}$ is
represented by edge number $1$.

The first step consists in writing $A_{\resG}$ as a single trace. To
this aim, we apply the following identity twice (for a general
graph, we need to apply it several times): let $\mathbf{c}$ be
any non empty proper subset of $\set{1,2,3,4}$ and
$\be^{\mathbf{c}}_{\mtup\ntup}$ be the tensor product
$\bigotimes_{c\in \mathbf{c}}\be^{c}_{m_{c}n_{c}}$. Then
\begin{equation}
  \label{eq-PartialDualityEdge}
  \sum_{\mtup,\ntup,\mtup',\ntup'\in\Z^{|\mathbf{c}|}}(\be^{\mathbf{c}}_{\mtup\ntup})_{\tuple a\tuple
    b}(\be^{\mathbf{c}}_{\mtup'\ntup'})_{\tuple d\tuple
    e}\,\delta^{|\mathbf{c}|}_{\mtup\ntup'}\delta^{|\mathbf{c}|}_{\ntup\mtup'}=\delta^{|\mathbf{c}|}_{\tuple
    a\tuple e}\delta^{|\mathbf{c}|}_{\tuple b\tuple d}.
\end{equation}
We apply it first to
$\be^{c_{1}}_{m_{1}n_{1}}\be^{c_{1}}_{m'_{1}n'_{1}}$ in $A_{\resG}$
then to the two remaining $\Itens_{\hat c_{1}}$ factors but in the
reverse direction (\ie from right to left in
\cref{eq-PartialDualityEdge}). We get
\begin{multline}
  \label{eq-DualAmplitudeExample}
  A_{\resG}=\sum_{\tupm_{1},\tupn_{1},\tupm'_{1}\tupn'_{1}\in\Z^{3}}\Big(\prod_{i=2}^{3}\sum_{m_{i},n_{i},m'_{i},n'_{i}\in\Z}\Big)\Tr\big[(\be^{\hat
    c_{1}}_{\mtup_{1}\ntup_{1}}\otimes\Itens_{c_{1}})\sqrt C\fres \sqrt C(\be^{c_{2}}_{m'_{2}n'_{2}}\otimes\Itens_{\hat
    c_{2}})C\\
  (\be^{c_{3}}_{m_{3}n_{3}}\otimes\Itens_{\hat
    c_{3}})C(\be^{c_{3}}_{m'_{3}n'_{3}}\otimes\Itens_{\hat
    c_{3}})\sqrt C\fres\sqrt
  C
  (\be^{\hat
    c_{1}}_{\mtup'_{1}\ntup'_{1}}\otimes\Itens_{c_{1}})C(\be^{c_{2}}_{m_{2}n_{2}}\otimes\Itens_{\hat
    c_{2}})C\big]\\
  \times\delta^{3}_{\tupm_{1}\tupn'_{1}}\delta^{3}_{\tupn_{1}\tupm'_{1}}\delta_{m_{2}n'_{2}}\delta_{n_{2}m'_{2}}\delta_{m_{3}n'_{3}}\delta_{n_{3}m'_{3}}.
\end{multline}
As usual in quantum field theory, we would like to represent this new
expression by a graph $\resG'$, a map in fact. It would allow us to
understand how to proceed with Step $1$ in the case of a general
graph. Given that \cref{eq-DualAmplitudeExample} contains only one
trace, it is natural to guess that $\resG'$ has only one vertex, but still three edges. What is the
relationship between $\resG$ and $\resG'$? To understand it, we must
come back to the Feynman graphs of our original tensor model. Each
edge of a graph in the \ifrt corresponds to a melonic quartic vertex,
somehow stretched in the direction of its distinguished
colour, see \cref{f-Edges} left. Applying twice identity \eqref{eq-PartialDualityEdge} to a
given edge $\ell$, we
first contract it and then re-expand it in the orthogonal
direction. This operation bears the name of partial duality with
respect to $\ell$, see \cite{Chmutov2007aa} where
\fabciteauthorinits{Chmutov2007aa} introduced that duality
relation. It is a generalization of the natural duality of maps which
exchanges vertices and faces. Partial duality can be applied \wrt any
spanning submap of a map. Natural duality corresponds to partial
duality \wrt the full map. The number of vertices of the partial
dual $\resG^{E'}$ of $\resG$ \wrt the spanning submap $\skel F_{E'}$ of edge-set $E'$
equals the number of faces of $\skel F_{E'}$. In our example, we
performed partial duality of $\resG$ \wrt edge $1$. Its spanning submap
of edge-set $\set 1$ has only one face. $\resG'$ has consequently only
one vertex, which is confirmed by expression
\eqref{eq-DualAmplitudeExample} containing only one trace. Note also
that if a direct edge bears a single colour index $c$, its dual edge
has the three colours $\hat c$. This can be seen on the amplitudes
themselves: in \cref{eq-DualAmplitudeExample} edge $1$ corresponds to
the two \emph{three-dimensionnal} deltas
$\delta^{3}_{\tupm_{1}\tupn'_{1}}\delta^{3}_{\tupn_{1}\tupm'_{1}}$
whereas edge $1$ in \cref{eq-cvGraphAmplExample} represents the two
\emph{one-dimensionnal} deltas $\delta_{m_{1}n'_{1}}\delta_{n_{1}m'_{1}}$.
\begin{figure}[!htp]
  \centering
  \begin{minipage}[c]{.4\linewidth}
    \centering
    \includegraphics[scale=.7]{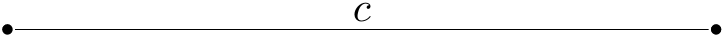}\\
    \medskip
    =\\
    \includegraphics[scale=.7]{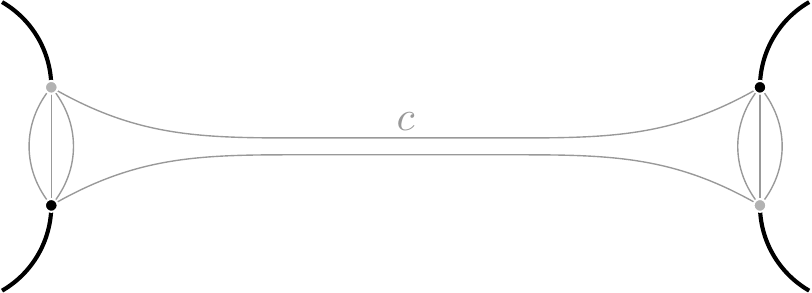}
  \end{minipage}\hspace{1.5cm}
  \begin{minipage}[c]{.4\linewidth}
    \centering
    \includegraphics[align=c,scale=.7]{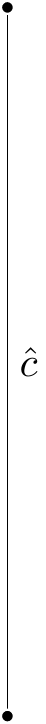}\qquad
    =\qquad\includegraphics[align=c,scale=.7]{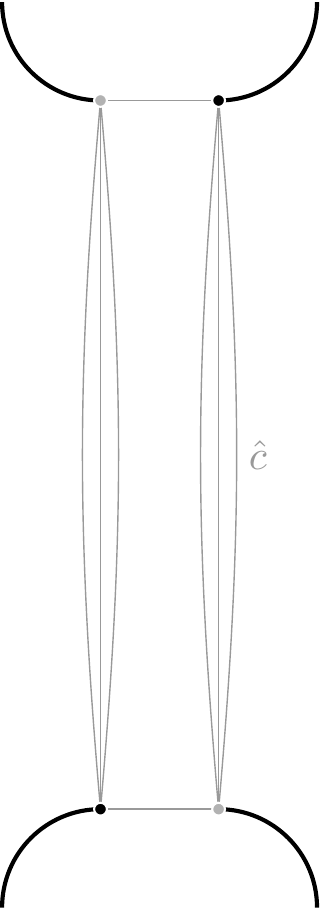}
   \end{minipage}
  \caption{Edges (on the left) and dual edges (on the right) both in
    the intermediate field and the coloured tensor representations.}
  \label{f-Edges}
\end{figure}

Given a map $\resG$, how to draw its dual $\resG^{E'}$ \wrt the
spanning submap of edge-set $E'\subseteq E(\resG)$? Cut the edges of
$\resG$ not in $E'$, making them half-edges. Turning around the faces of
$\skel F_{E'}$, one (partial) orders all the half-edges of $\resG$, \ie
including those in $E(\resG)\setminus E'$. The cycles of half-edges thus obtained constitute the vertices of
$\resG^{E'}$. Finally, connect in $\resG^{E'}$ the half-edges which
formed an edge in $\resG$. The result of this construction in the case
of the example of \cref{f-cvGraphExample} with $E'=\set{1}$ is given
in \cref{f-chorddiagex}. Note that we will always
represent one-vertex maps as chord diagrams.
\begin{figure}[!htp]
  \centering
  \includegraphics[scale=.8]{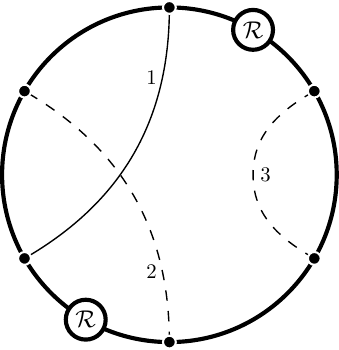}
  \caption{The partial dual $\resG^{\set 1}$ of the map $\resG$ of
    \cref{f-cvGraphExample}, as a chord diagram. In general \ie in the
    case of the partial dual of $\resG$ \wrt $E'$, edges in $E'$ will be depicted as solid lines and those in $E(\resG)\setminus
    E'$ as dashed lines. Resolvent insertions are
    explicitely represented. Bold solid line segments on the external
    circle correspond to propagators (or square roots of propagators
    around resolvents).}
  \label{f-chorddiagex}
\end{figure}
\clearpage

The advantage of writing the amplitude of $\resG$ as a single trace is
that it allows us to easily identify it with a scalar product. Let us
indeed rewrite the amplitude of $\resG$ as
\begin{multline*}
  A_{\resG}=\sum_{\substack{\tupm,\tuple
      l\in\Z^{4}\\m_{2},n_{2},m'_{2},n'_{2}\in\Z}}\delta_{m_{2}n'_{2}}\delta_{n_{2}m'_{2}}\\
  \Big(\sum_{\tupn,\tuple
    k\in\Z^{4}}\fres_{\tuple{m}\tuple{n}}\fres^{\transpose}_{\tuple l\tuple k}\,\sum_{m_{3},n_{3},m'_{3},n'_{3}\in\Z}\delta_{m_{3}n'_{3}}\delta_{n_{3}m'_{3}}\big(\sqrt C(\be^{c_{2}}_{m'_{2}n'_{2}}\otimes\Itens_{\hat
    c_{2}})C(\be^{c_{3}}_{m_{3}n_{3}}\otimes\Itens_{\hat
    c_{3}})C(\be^{c_{3}}_{m'_{3}n'_{3}}\otimes\Itens_{\hat
    c_{3}})\sqrt C\big)_{\tupn\tuple k}\Big)\\
  \times\Big(\sum_{\tupm_{1},\tupn_{1},\tupm'_{1},\tupn'_{1}\in\Z^{3}}\delta^{3}_{\tupm_{1}\tupn'_{1}}\delta^{3}_{\tupn_{1}\tupm'_{1}}\big(\sqrt C(\be^{\hat
    c_{1}}_{\mtup'_{1}\ntup'_{1}}\otimes\Itens_{c_{1}})C(\be^{c_{2}}_{m_{2}n_{2}}\otimes\Itens_{\hat
    c_{2}})C (\be^{\hat
    c_{1}}_{\mtup_{1}\ntup_{1}}\otimes\Itens_{c_{1}})\sqrt
  C\big)_{\tuple l\tupm}\Big).
\end{multline*}
Then the amplitude takes the form of a scalar product in
$\Htens\otimes\Hop[2]\otimes\Htens$:
\begin{align}
  A_{\resG}&=\scalprodtens{\alpha}{(\fres\otimes\fres^{\transpose})\beta},\label{eq-scalprodex}\\
  \alpha&=\sum_{\tupm_{1},\tupn_{1},\tupm'_{1},\tupn'_{1}\in\Z^{3}}\delta^{3}_{\tupm_{1}\tupn'_{1}}\delta^{3}_{\tupn_{1}\tupm'_{1}}\big(\sqrt C(\be^{\hat
    c_{1}}_{\mtup'_{1}\ntup'_{1}}\otimes\Itens_{c_{1}})C(\be^{c_{2}}_{m_{2}n_{2}}\otimes\Itens_{\hat
    c_{2}})C (\be^{\hat
    c_{1}}_{\mtup_{1}\ntup_{1}}\otimes\Itens_{c_{1}})\sqrt
  C\big)^{\dagger},\nonumber\\
  \beta&=\sum_{m_{3},n_{3},m'_{3},n'_{3}\in\Z}\delta_{m_{3}n'_{3}}\delta_{n_{3}m'_{3}}\big(\sqrt C(\be^{c_{2}}_{m'_{2}n'_{2}}\otimes\Itens_{\hat
    c_{2}})C(\be^{c_{3}}_{m_{3}n_{3}}\otimes\Itens_{\hat
    c_{3}})C(\be^{c_{3}}_{m'_{3}n'_{3}}\otimes\Itens_{\hat
    c_{3}})\sqrt C\big).\nonumber
\end{align}
The vectors $\alpha$ and $\beta$ can be pictorially identified: from
the graph of \cref{f-chorddiagex}, one first detaches the two
resolvents and then cut along a line joining their former positions,
see \cref{f-cuttingDiagramsEx}.
\begin{figure}[!htp]
  \centering
  \includegraphics[scale=.8]{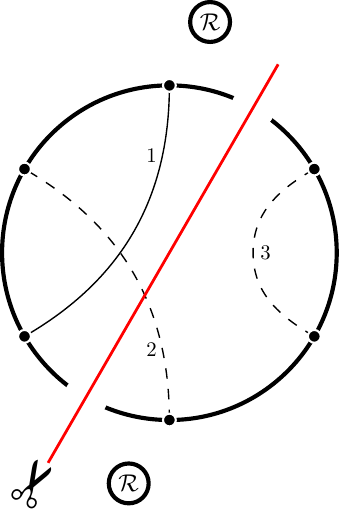}
  \caption{Amplitudes as scalar products.}
  \label{f-cuttingDiagramsEx}
\end{figure}

As can be seen in \cref{eq-scalprodex}, the amplitude of $\resG$ does
not exhibit any permutation operator. This is due to the fact that the
(red) cut of this example crosses only one edge, see
\cref{f-cuttingDiagramsEx}. A permutation operator appears \ifft there
are some crossings among the cut edges. Let us now give a second example, $\resgraph{H}$,
the amplitude of which contains such a permutation, see
\cref{f-PermutationOpEx} (left). On the right of $\resgraph{H}$ we have its
partial dual \wrt edges $1$ and $2$. Cutting this diagram through
both resolvents, one identifies the two vectors $\alpha$ and $\beta$ in
$\Htens\otimes\Hop[c_{3}]\otimes\Hop[c_{2}]\otimes\Hop[c_{1}]\otimes\Htens$
(reading counterclockwise) and the permutation operator $S$ (see \cref{f-PermutationOpEx} right)
from $\Hop[c_{2}]\otimes\Hop[c_{1}]\otimes\Hop[c_{3}]$ to
$\Hop[c_{3}]\otimes\Hop[c_{2}]\otimes\Hop[c_{1}]$ such that $A_{\resgraph{H}}=\scalprodtens{\alpha}{(\fres\otimes S\otimes\fres^{\transpose})\beta}$.
\begin{figure}[!htp]
  \centering
  \begin{tikzpicture}
    \node (H) {\includegraphics[scale=1,align=c]{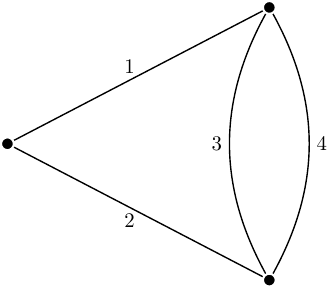}};
    \node (Hdual) [right=of H] {\includegraphics[scale=.8,align=c]{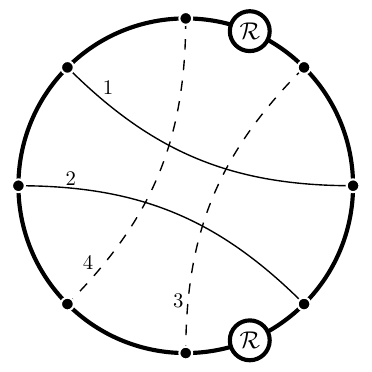}};
    \node (S) [right=of Hdual] {\includegraphics[scale=.8,align=c]{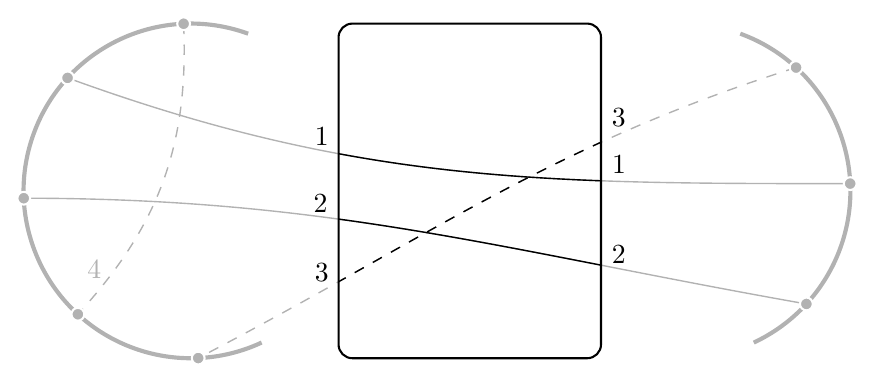}};
    \node (Htag) [below=2.4cm of H.center,anchor=south] {\large $\mathcal H$};
    \node (Hdualtag) [below=2.4cm of Hdual.center,anchor=south] {\large $\mathcal H^{\set{1,2}}$};
    \node (Stag) [below=2.4cm of S.center, xshift=.2cm,anchor=south] {\large $S$};
  \end{tikzpicture}
  \caption{Example of a graph $\resgraph{H}$ (left) the amplitude of which, written as a
    scalar product, exhibits a permutation operator $S$ (right). The
    picture in the middle is the partial dual $\resgraph{H}^{\set{1,2}}$ of
  $\resgraph{H}$ \wrt edges $1$ and $2$. The vectors whose scalar
  product equals $A_{\resgraph{H}}$ are identified by cutting the chord
  diagram of $\resgraph{H}^{\set{1,2}}$ through both resolvents.}
  \label{f-PermutationOpEx}
\end{figure}
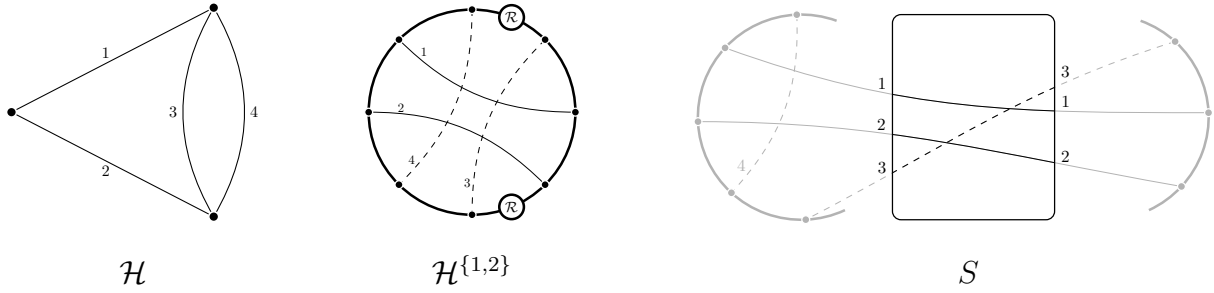

After having written the amplitude of a graph as a scalar product, we
can apply \CSi which corresponds to Step \ref{item-StepApplyCS} in the \ics
algorithm. Finally there only remains to identify the squares of the
norms of $\alpha$ and $\beta$ as amplitudes of some definite maps. It
simply consists in duplicating each half of the cut diagram and glue each
piece to its mirror symmetric one \ie its Hermitian conjugate. In the
case of graph $\resG$ of \cref{f-cuttingDiagramsEx}, we get the two
chord diagrams of \cref{f-DuplicatingHalfDiagrams}.
\begin{figure}[!htp]
    \hfil\includegraphics[scale=.8]{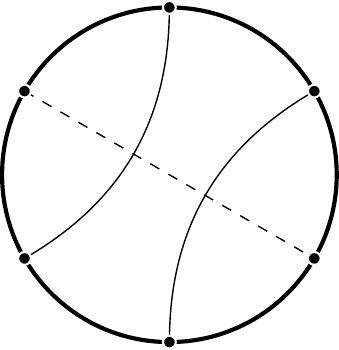}\hfil\includegraphics[scale=.8]{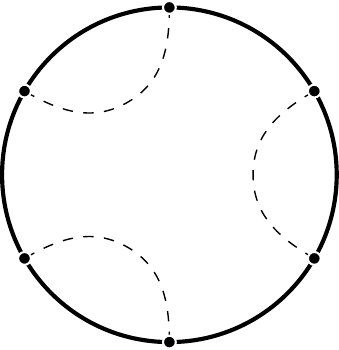}\hfil
    \caption{$\normtenssq{\alpha}$ (left) and $\normtenssq{\beta}$
      (right) in the case of \cref{f-cuttingDiagramsEx}.}
    \label{f-DuplicatingHalfDiagrams}
  \end{figure}
But in general it could happen that $\normtenssq{\alpha}$ (or $\normtenssq{\beta}$) is
infinite that is to say its corresponding chord diagram is dual to a
divergent graph. To conclude this section of examples, let us exhibit
a graph such that any cut of its chord diagram leads to divergent
graphs. Let $\resG$ be the graph of \cref{f-divergentCutsEx} (above left), in the \ifrt. The gray
parts represent renormalized subgraphs. Let us perform partial
duality \wrt all its edges and get the chord diagram of
\cref{f-divergentCutsEx} (above right). All of its four possible cuts
(we \emph{never} cut inside a renormalized block) lead to divergent
upper bounds by \CSi.
\begin{figure}[!htp]
  \centering
  \begin{tikzpicture}
    \node (G) at (0,0) {\includegraphics[scale=1]{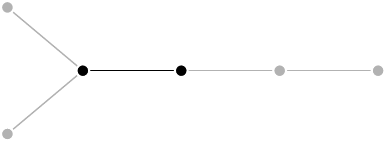}};
    \node (dG) [right=3cm of G] {\includegraphics[scale=.8]{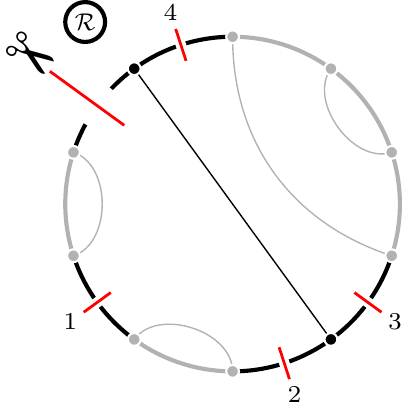}};
    \node (Gtag) [below=2.5cm of G.center,anchor=south] {\large $\resG$};
    \node (dGtag) [below=2.5cm of dG.center,anchor=south] {\large
      $\resG^{E(\resG)}$};
  \end{tikzpicture}\\
  \vspace{1.0cm}
  \begin{tikzpicture}
    \node (cutun) at (0,0) {\includegraphics[scale=.6]{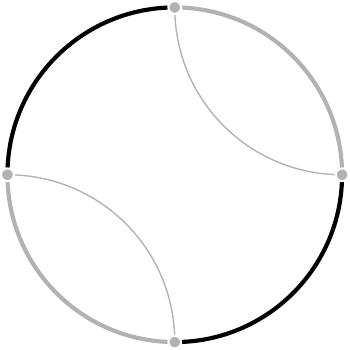}};
    \node (cutuntag) [below=1.9cm of cutun.center,anchor=south]
    {\large $1$};
    \node (cutdeux) [right=of cutun]
    {\includegraphics[scale=.6]{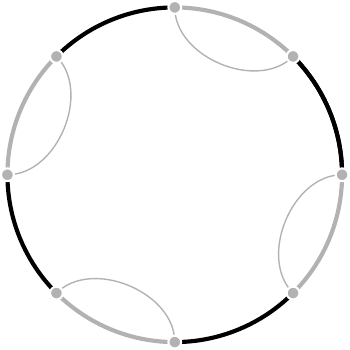}};
    \node (cutdeuxtag) [below=1.9cm of cutdeux.center,anchor=south]
    {\large $2$};
    \node (cuttrois) [right=of cutdeux]
    {\includegraphics[scale=.6]{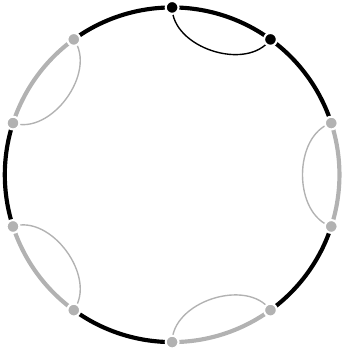}\quad \includegraphics[scale=.6]{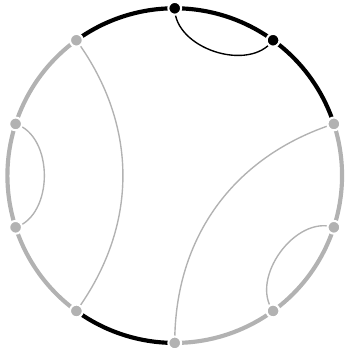}};
    \node (cuttroistag) [below=1.9cm of cuttrois.center,anchor=south]
    {\large $3$};
    \node (cutquatre) [right=of cuttrois]
    {\includegraphics[scale=.6]{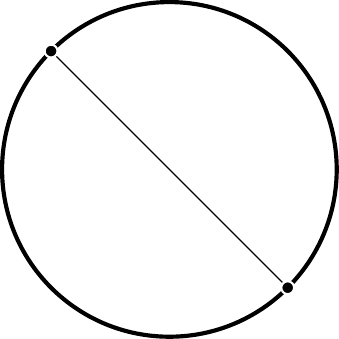}};
    \node (cutquatretag) [below=1.9cm of cutquatre.center,anchor=south]
    {\large $4$};
  \end{tikzpicture}
  \caption{A graph $\resG$ with divergent cuts. Gray parts represent
    renormalized subgraphs. The four possible cuts of
    $\resG^{E(\resG)}$ are indicated by numbered red segments. On the
    second line, we display the divergent factors of
    $\normtenssq{\alpha}\normtenssq{\beta}$ for the different cuts.}
  \label{f-divergentCutsEx}
\end{figure}

\subsubsection{ICS algorithm 1.0}
\label{sec-ics-algorithm-1.0}

Thus there exist chord diagrams with only divergent cuts. How do we
get rid of their resolvents using Cauchy-Schwarz inequality? We can in
fact expand some of the resolvents, $\fres=\Itens+U\fres$, and get new
graphs. In the sequel we will show that for all resolvent graph
$\resG$, there is a systematic way of expanding its resolvents such
that, for any newly created graph, there exists an iterative cutting
scheme which converges itself to a collection of graphs without
resolvents.\\

A more precise (but still not enough) ICS algorithm can be written as follows:
\begin{algorithm}[H]
  \caption{ICS 1.0}
  \label{algo-ICS}
  \begin{algorithmic}[1]
    \Require $\resG$ a resolvent graph.
    \State \textbf{Partial duality}: Write $A_{\resG}$ as $c(\resG)$ traces (times Kronecker
    deltas)
    \State \textbf{Preparation step}: Expand (some of) the resolvents of $A_{\resG}$ conveniently
    and get a collection $S$ of new resolvent graphs
    \For{$\resgraph{S}$}{$S$}
    \State\label{algo-stepCut} \textbf{Cutting scheme}: choose a cut and thus write $A_{\resgraph{S}}$ as a scalar product
    \State \textbf{Cauchy-Schwarz inequality}: apply it to
    $A_{\resgraph{S}}$
    \State Go back to step \ref{algo-stepCut} and iterate sufficiently.
    \EndFor
  \end{algorithmic}
\end{algorithm}
\noindent
The first step of \cref{algo-ICS} consists in writing the amplitude
$A_{\resG}$ of a resolvent graph $\resG$ as a product of $c(\resG)$ traces. To this aim, we choose arbitrarily a spanning tree in each
connected component and perform partial duality \wrt this set $\resgraph{F}$ of
edges. The amplitude of each connected component of $\resG$ is then
represented by a one-vertex map that we will draw as a chord
diagram. The disjoint union of all these chord diagrams form the
partial dual $\resG^{\resgraph{F}}$ of $\resG$. An edge of colour $c$
in $\resG$ still bears colour $c$ in $\resG^{\resgraph{F}}$ if it does
not belong to $\resgraph{F}$ and bears colours $\hat
c=\set{1,2,3,4}\setminus\set{c}$ if it is in $\resgraph{F}$. Tree
edges will be represented as plain lines and loop edges as dashed
lines in the following pictures.

\subsubsection{The preparation step}

In order to write the amplitude of (each connected component of)
$\resG$ as a scalar product we need to choose a cut in the corresponding chord
diagram. But as we have seen previously, there exist resolvent
graphs such that any Cauchy-Schwartz cut results in divergent
amplitudes $\normtenssq{\alpha}$ and/or
$\normtenssq{\beta}$. Nevertheless we can see on \cref{f-dvirdual}
that divergent vacuum graphs (which have essentially only one spanning tree and
thus a canonical associated chord diagram) have either less than four
tree lines and no loops, or one loop line and less than one tree line,
or two loops but no tree lines. Thus if a diagram has enough edges, so to
speak, between the two resolvents of a cut, the Cauchy-Schwarz bound
will be \emph{superficially} convergent. We will ensure it by suitably
expanding some resolvents as $\fres=\Itens+\fres U$ or $\Itens+U\fres$.
\begin{figure}[!htp]
  \centering
  \begin{tikzpicture}[node distance=.9cm and 1cm,label distance=-.1cm]
    \def\repdist{.2cm};
    \matrix[row sep={between origins,3cm}, column sep=1cm]{
    \node[label=below:$\cV_{1}$] (un) {\includegraphics[scale=1,align=c]{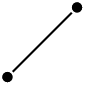}\hspace{\repdist}\includegraphics[scale=\scadvdual,align=c]{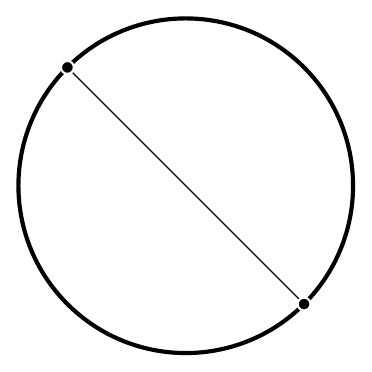}};&
    \node[label=below:$\cV_{2}$] (deux) {\includegraphics[scale=1,align=c]{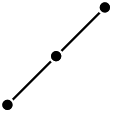}\hspace{\repdist}\includegraphics[scale=\scadvdual,align=c]{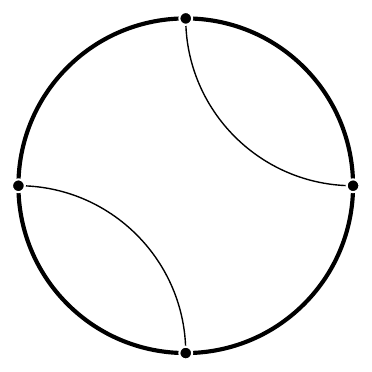}};&
    \node[label=below:$\cV_{3}$] (trois) {\includegraphics[scale=1,align=c]{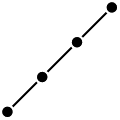}\hspace{\repdist}\includegraphics[scale=\scadvdual,align=c]{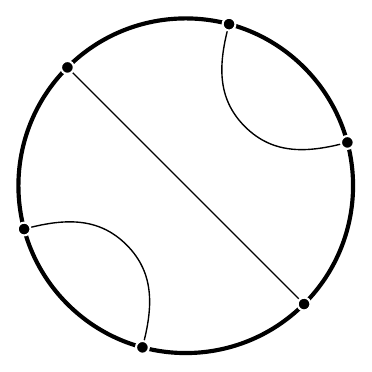}};\\
    \node[label=below:$\cV_{4}$] (quatre) {\includegraphics[scale=1,align=c]{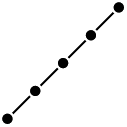}\hspace{\repdist}\includegraphics[scale=\scadvdual,align=c]{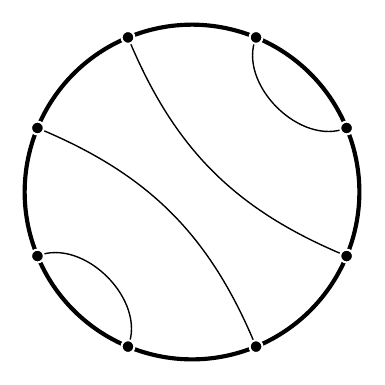}};&
    \node[label=below:$\cV_{5}$] (cinq) {\includegraphics[scale=1,align=c]{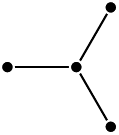}\hspace{\repdist}\includegraphics[scale=\scadvdual,align=c]{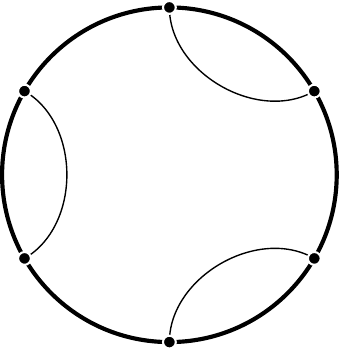}};&
    \node[label=below:$\cV_{6}$] (six) {\includegraphics[scale=1,align=c]{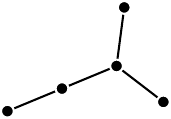}\hspace{\repdist}\includegraphics[scale=\scadvdual,align=c]{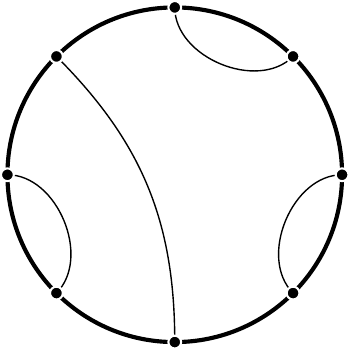}};\\
    \node[label=below:$\cV_{7}$] (sept) {\includegraphics[scale=1,align=c]{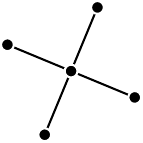}\hspace{\repdist}\includegraphics[scale=\scadvdual,align=c]{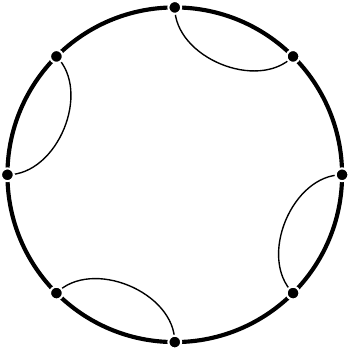}};&&\\
    \node[label=below:$\kN_{1}$] (Nun) {\includegraphics[scale=1,align=c]{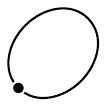}\hspace{\repdist}\includegraphics[scale=\scadvdual,align=c]{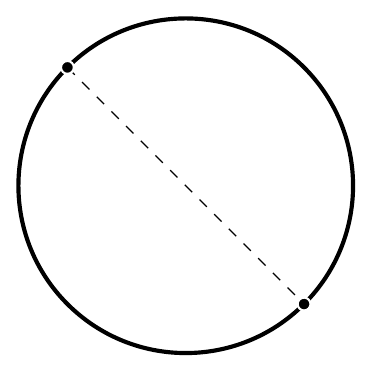}};&
   \node[label=below:$\kN_{2}$] (Ndeux) {\includegraphics[scale=1,align=c]{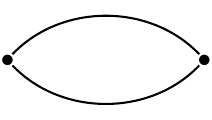}\hspace{\repdist}\includegraphics[scale=\scadvdual,align=c]{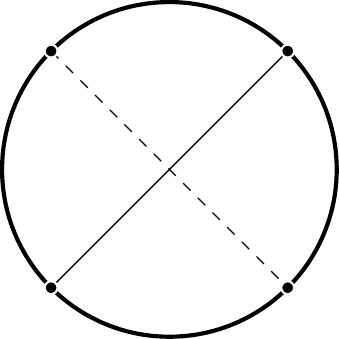}};&
 \node[label=below:$\kN_{3}$] (Ntrois) {\includegraphics[scale=1,align=c]{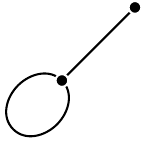}\hspace{\repdist}\includegraphics[scale=\scadvdual,align=c]{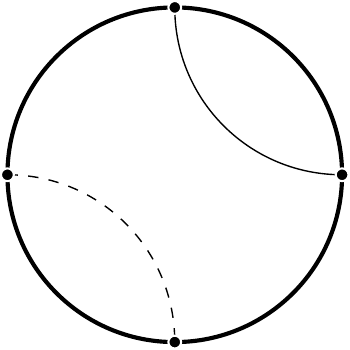}};&\\
};
  \end{tikzpicture}
  \caption{The divergent vacuum graphs in the intermediate field
    (left) and dual (right) representations.}
  \label{f-dvirdual}
\end{figure}

But to ensure finiteness, we also need to find a cut such that no divergent subgraphs pop
up in $\normtenssq{\alpha}$ and/or $\normtenssq{\beta}$. Divergent
($2$-point) subgraphs appear in chord diagrams as represented in
\cref{fig-dvsubgraphs}. Note that they are absent from resolvent
graphs (and from their partial duals) because \MLVE produced only
renormalized amplitudes. It is easy to convince oneself that if there is no
tree line next to corners of cut, there will be no divergent subgraphs
in $\normtenssq{\alpha}$ and $\normtenssq{\beta}$.
\begin{figure}[!htp]
  \centering
  \begin{tikzpicture}
    \node[label=below:$\cM_{1}$,anchor=south] at (0,0)
    {\includegraphics[scale=1]{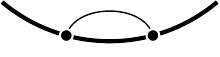}};
    \node[label=below:$\cM_{2}$,anchor=south] at (6cm,0)
    {\includegraphics[scale=1]{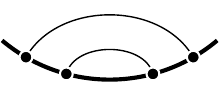}};
  \end{tikzpicture}
  \caption{Divergent subgraphs in the dual representation.}
  \label{fig-dvsubgraphs}
\end{figure}\\

We now explain precisely which resolvents will be
expanded and how many times. Later on, we will prove that after such
expansions there exists a sequence of iterated Cauchy-Schwarz
cuts which bounds the amplitude of any resolvent graph by the
geometric mean of finite amplitudes, most of them freed of resolvents.

\clearpage
First of all, we need to define when resolvent expansions should
stop \ie when we consider a diagram as secured or said differently when a
diagram is ready for the cut process to be defined in the next
section.

In the following we will always read a chord diagram counterclockwise. Thus if $O_{1}$ and $O_{2}$ are operators in $\OpHtens$
and appear in the amplitude of $\gls{resC}$, we will consider that $O_{2}$ is \emph{on the
right} of $O_{1}$ if $O_{2}$ is met just after $O_{1}$ counterclockwise around
$\resgraph{C}$ or equivalently if $A_{\resgraph{C}}$ contains the
product $O_{1}O_{2}$. We will say symmetrically that $O_{2}$ is \emph{on the
left} of $O_{1}$ if the product $O_{2}O_{1}$ appears in
$A_{\resC}$. We will denote $r(\resC)$ the number of resolvents in $A_{\resC}$.
\begin{defn}[Safeness]\label{def-Safeness}
  Let us consider a chord diagram representing the partial dual
  of a resolvent graph. A \emph{safe element} is either a half loop edge or a
  renormalized $D$-block.
\end{defn}
\begin{defn}[Tree-resolvents]\label{def-treeRes}
  We say that a resolvent $\fres$ is a
  right (\resp left) \emph{tree-resolvent} if
  \begin{itemize}
  \item the product $\dU{}S\fres$ (\resp $\fres S\dU{}$), where $S$ is
    itself a possibly empty product of safe elements and $s$ labels a
    half tree line, appears in $A_{\resC}$
  \item and the number of safe elements in $S$ is less than or equal to six.
  \end{itemize}
A tree-resolvent is a resolvent which is either a right
  or a left tree-resolvent (or both). Tree-resolvents are the resolvents ``closest'' to the
  tree of $\resC$. We also let $t(\resC)$ be the number of tree-resolvents in
  $A_{\resC}$.
\end{defn}
We will need to order the tree-resolvents of a diagram amplitude. In
the following if $\resC$ is a connected chord diagram, we will write
$\resgraph{C}_{\bullet}$ for a pair made of $\resC$ and a
distinguished tree-resolvent (called root resolvent
hereafter). We consider all of its tree-resolvents as ordered counterclockwise
starting with the root one and denote them $\fres_{1},\fres_{2},\dotsc,\fres_{t(\resC)}$. If
$\resC=\sqcup_{i=1}^{c(\resC)}\resC_{i}$ is a disjoint union of
chord diagrams (and $c(\resC)$ is the number of connected
components of $\resC$), $\rC$ stands for a choice of one root resolvent per connected
component. In each $\rC[i]$, resolvents are ordered from
$1$ to $t(\resC_{i})$.
\begin{defn}[Distance to tree]\label{def-distTree}
  Let $\resC$ be a connected Feynman chord diagram. Let $s$ be a half tree edge
  and $j$ an element of $\set{1,2,\dotsc,t(\resC)}$. The pair $(s,j)$
  is \emph{admissible} if $\fres_{j}$ is a tree-resolvent and $s$ is
  separated from $\fres_{j}$ only by safe elements. Said differently,
  from $\fres_{j}$ to $s$ we meet neither half tree edges nor
  resolvents. For any admissible pair $p=(s,j)$, let $d_{p}$ be the number of
  safe elements in $A_{\resC}$ between $\dU{}$ and
  $\fres_{j}$. $d_{p}$ is the \emph{distance} between $s$ and
  $\fres_{j}$ and is, by \cref{def-treeRes}, less than or equal to
  six.
\end{defn}
\begin{defn}[Secured diagrams]\label{def-SecuredDiagrams}
  A connected chord diagram $\resC$ is \emph{secured} if either
  $r(\resC)=0$ or for any admissible pair $p$, $d_{p}$ equals six. A possibly disconnected diagram is secured if all its connected components are secured.
\end{defn}

We now explain which resolvents of a diagram we expand, and how, in
order to reach only secured graphs. \Cref{algo-ExpandRes} simply expands on its
right a given resolvent of a graph. More precisely it returns the list
of graphs representing the various terms of the expansion. A
symmetrical algorithm, named $\ExpandL$, does the same on the left.
\begin{algorithm}[H]
  \caption{Right expansion}
  \label{algo-ExpandRes}
  \begin{algorithmic}[5]
    \Require $\rC$ a rooted chord diagram, $1\les i\les c(\resC)$ and $1\les j\les r(\resC_{i})$.
    \Procedure{ExpandR}{$\resgraph{C}_{\bullet}$, $i$, $j$}
    \Comment{Expands once $\fres_{j}$ on its right in $A_{\resC_{i}}$.}
    \State $L\defi [\ ]$\Comment{an empty list}
    \State Expand $\fres_{j}$ as $\Itens+\fres_{j}(D+\Sigma)$
    \Statex
    \State $\resC^{(0)}_{i}\defi\resC_{i}$
    with \includegraphics[scale=.7,align=c]{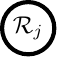} replaced
    by \includegraphics[scale=.7,align=c]{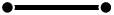}
    \State $\resC^{(0)}\defi\resC\sqcup\resC_{i}^{(0)}\setminus\resC_{i}$
    \State $L$.append($\resC^{(0)}$)
    \Statex
    \State $\resC^{(1)}_{i}\defi\resC_{i}$
    with \includegraphics[scale=.7,align=c]{dualres} replaced
    by \includegraphics[scale=.7,align=c]{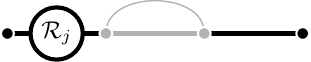}
    \State $\resC^{(1)}\defi\resC\sqcup\resC_{i}^{(1)}\setminus\resC_{i}$
    \State $L$.append($\resC^{(1)}$)
    \Statex
    \State Integrate by parts the $\Sigma$-term (\cref{eq-intbyparts})
    \Comment{$r(\resC)$ new graphs.}
    \For[k]{2}{$r(\resC)+1$}
    \State $e_{k}\defi\text{the new additional edge}$
    \If{$e_{k}$ is a loop}
    \State $\resC_{i}^{(k)}\defi\resC_{i}\cup\set{e_{k}}$
    \State $\resC^{(k)}\defi\resC\sqcup\resC_{i}^{(k)}\setminus\resC_{i}$
    \Else\Comment{$e_{k}$ connects $\resC_{i}$ to $\resC_{i'}$, $i\neq
      i'$.}
    \State $\resC_{i}^{(k)}\defi(\resC_{i}\cup\resC_{i'}\cup\set{e_{k}})^{\set{e_{k}}}$
    \State
    $\resC^{(k)}\defi\resC\sqcup\resC_{i}^{(k)}\setminus\set{\resC_{i},\resC_{i'}}$\vspace{1pt}
    \EndIf
    \State $L$.append($\resC^{(k)}$)
    \EndFor
    \State \textbf{return} $L$
    \EndProcedure
  \end{algorithmic}
\end{algorithm}
Given a non secured connected component $\resC_{i}$ of a Feynman chord diagram,
\cref{algo-ChooseExpand} decides which resolvent to expand and how
many times. Before giving its
\href{https://en.wikipedia.org/wiki/Pseudocode}{pseudocode}, we need
to introduce a few more definitions. Let $j$ be an element of $\set{1,2,\dotsc,t(\resC_{i})}$. We define
$\Right_{\rC[i]}(\fres_{j})$ as the number of consecutive safe
elements at the right of $\fres_{j}$. We define
$\Left_{\rC[i]}(\fres_{j})$ symmetrically. We let
$\RightTree_{\rC[i]}(\fres_{j})$ (\resp $\LeftTree_{\rC[i]}(\fres_{j})$) be True if
$\fres_{j}$ is a right (\resp left) tree-resolvent and False otherwise. $\Root(\resC_{i})$ chooses a root resolvent
among the tree-resolvents, randomly say.
\begin{algorithm}
  \caption{Choose \& expand}
  \label{algo-ChooseExpand}
  \begin{algorithmic}[5]
    \Require $\resC$ a Feynman chord diagram and $1\les i\les
    c(\resC)$ such that $\resC_{i}$ not secured.
    \Procedure{ChooseExpand}{$\resC$,$i$}
    \State $\rC[i]\defi (\resC_{i},\Root(\resC_{i}))$
    \State $j\defi 1$
    \While{$j\les t(\resC_{i})$}
    \If{$\RightTree_{\rC[i]}(\fres_{j})$ \algoand
      $\Left_{\rC[i]}(\fres_{j})\les 5$}
    \State \textbf{return} $\ExpandL(\rC,i,j)$
    \ElsIf{$\LeftTree_{\rC[i]}(\fres_{j})$ \algoand
      $\Right_{\rC[i]}(\fres_{j})\les 5$}
    \State \textbf{return} $\ExpandR(\rC,i,j)$
    \Else
    \State $j\defi j+1$
    \EndIf
    \EndWhile
    \EndProcedure
  \end{algorithmic}
\end{algorithm}

Finally \cref{algo-SecuringResolvents} secures all the resolvents of a given
diagram $\resC$. More precisely it returns the list of secured
diagrams obtained from $\resC$ by successive expansions of its
resolvents. \Cref{algo-SecuringResolvents} can be thought of as
building a rooted tree $T_{\resC}$ inductively. At each of the nodes of that tree,
there is an associated chord diagram. The root of $T_{\resC}$ consists in the input diagram $\resC$. The children of a given
node $\resC'$ correspond to the $r(\resC')+2$ new graphs obtained by
expanding one resolvent of $\resC'$, the one chosen by
$\ChooseExpand$. \Cref{algo-SecuringResolvents} returns the list of
totally secured graphs. They correspond to the leaves of $T_{\resC}$.
\begin{algorithm}
  \caption{Securing resolvents}
  \label{algo-SecuringResolvents}
  \begin{algorithmic}[5]
    \Require $\resC$ a Feynman chord diagram.
    \State $L\defi[\resC]$
    \State $S\defi [\ ]$
    \While{$L$ not empty}
    \State $D\defi [\ ]$
    \For[$k$]{0}{$\pythlen(L)-1$}\Comment{$k$ indexes the graphs in
      $L$.}
    \If{$L[k]$ secured}
    \State $S.\pythappend(L[k])$
    \State $D.\pythappend(L[k])$
    \Else
    \State Pick a non secured connected component $L[k]_{i}$ of $L[k]$
    \State $L\defi L+\ChooseExpand(L[k],i)$
    \State $D.\pythappend(L[k])$
    \EndIf
    \EndFor
    \For{$\resG$}{$D$}
    \State $L.\pythremove(\resG)$
    \EndFor
    \EndWhile
    \State \textbf{return} $S$
  \end{algorithmic}
\end{algorithm}

We now prove that \cref{algo-SecuringResolvents} stops after a finite
number of steps and give an upper bound on the number of elements of
the list it returns.
\begin{lemma}\label{thm-NumberOfLeaves}
  Let $\cB$ be Bosonic block with $n+1$ vertices. Let $\resC$ be one
  of the resolvent graphs obtained from $\cB$ by the contraction process. After a finite number
  of steps, \cref{algo-SecuringResolvents} applied to $\resC$ stops
  and returns a list of at most $(98n-28)^{42n-30}$ secured diagrams.
\end{lemma}
\begin{proof}
  In the computation tree $T_{\resC}$ representing \cref{algo-SecuringResolvents},
  each new generation corresponds to the expansion of a resolvent and
  each child of a given node to a term of this expansion (plus
  integration by parts). Along the branches of $T_{\resC}$, from a
  given node to one of its children, the number of
  connected components is constant except in case of a new tree
  edge where it decreases by one. In order to control the maximal
  number of steps taken by \cref{algo-SecuringResolvents} we now
  introduce one more parameter $m(\resC)$ namely the number of missing safe
  elements to get $\resC$ secured:
  \begin{displaymath}
    m(\resC)\defi\sum_{p\text{ admissible}}6-d_{p}.
  \end{displaymath}
\Cref{algo-SecuringResolvents} stops when $m=0$.

  Let us now inspect the evolution of $m$ along the branches of
  $T_{\resC}$. As \cref{algo-SecuringResolvents} only expands
  tree-resolvents, let us consider such an operator $\fres$. Locally,
  around $\fres$ in $A_{\resC}$, we have the following situation:
  $A_{1}S_{1}\fres S_{2}A_{2}$ where both $A_{1}$ and $A_{2}$ are
  either half tree edges or resolvents but at least one of them is a
  half tree edge and $S_{1},S_{2}$ are possibly empty products
  of safe elements. If the expansion term of $\fres$ is:
  \begin{itemize}
  \item $\Itens$ and
    \begin{itemize}
    \item both $A_{1}$ and $A_{2}$ are half tree edges then $m$
      decreases by $12-|S_{1}|-|S_{2}|\ges 1$ if both $|S_{1}|$ and
      $|S_{2}|$ are less than or equal to six, and by $6-|S_{2}|\ges
      1$ (\resp $6-|S_{1}|$) if $|S_{1}|$ (\resp $|S_{2}|$) is
      strictly greater than six,
    \item $A_{1}$ (\resp $A_{2}$) is a resolvent then $m$ decreases by
      $|S_{1}|$ (\resp $|S_{2}|$) if $|S_{1}|+|S_{2}|\les 6$ and by
      $6-|S_{2}|$ (\resp $6-|S_{1}|$) otherwise,
    \end{itemize}
  \item $D$, $m$ decreases by one,
  \item a new loop edge, $m$ decreases by one,
  \item a new tree edge, $m$ increases by $12+|S_{1}|$ (\resp
    $12+|S_{2}|$) if $\fres$ is
    left- (\resp right-)expanded.
  \end{itemize}
Thus at each generation, in all cases, the non-negative
  integer valued linear combination
  \begin{displaymath}
 \psi\defi 18(c-1)+m
\end{displaymath}
strictly decreases. As it is bounded
  above (at fixed $n$), \cref{algo-SecuringResolvents} stops
  after a finite number of steps.

In order to determinate an upper bound on the number of leaves of
$T_{\resC}$, we need a bound on its number of generations. As
$\psi\ges 0$, the length of a branch of $T_{\resC}$ is certainly
bounded by $\psi(\resC)$. The number of children of a node $\resC'$ is
$r(\resC')+2$. As the number of resolvents increases by $1$ with each
new added edge, the maximal total number of
  resolvents over all the nodes of $T_{\resC}$ is
  $r(\resC)+\psi(\resC)$. In conclusion, the number of leaves of
  $T_{\resC}$ is bounded by
  \begin{equation*}
    (r(\resC)+\psi(\resC)+2)^{\psi(\resC)}.
  \end{equation*}
As already discussed at the beginning of  \cref{sec-pert-funct-integr}, a resolvent
graph coming from a Bosonic block with $n+1$ vertices has at most $n$
connected components, $2n-1$ tree edges thus at most $4n-2$ admissible
pairs and less than $56n$ resolvents. We get $m(\resC)\les 24n-12$ and
$\psi(\resC)\les 42n-30$. Consequently, as a function of $n$, the number of new graphs
created by \cref{algo-SecuringResolvents} is bounded above by $(98n-28)^{42n-30}$.
\end{proof}

\subsubsection{Iterative cutting process}
\label{sec-iter-cutt-proc}

The preparation step has expressed the amplitude of any resolvent
graph $\resG$ as the sum over the leaves of $T_{\resG}$ of the
amplitudes of the corresponding secured graphs. Thus, from now on we
consider a secured Feynman chord diagram $\resC$, together with a scale
attribution $\mu$. We apply 
Cauchy-Schwarz inequalities to $A_{\resC_{\mu}}$ iteratively until we
bound $|A_{\resC_{\mu}}|$ by a geometric mean of convergent
resolvent-free amplitudes. 

First of all, note that an iterative cutting process can be represented as a
rooted binary tree. Its root corresponds to $\resC$ and the two children
of each node are the result of a Cauchy-Schwarz inequality. It will be
convenient to use the Ulam-Harris encoding of rooted plane trees \cite{Miermont2014aa}. It
identifies the set of vertices of a rooted tree with a subset of the
set
\begin{displaymath}
  \cU=\bigcup_{n\ges 0}\N^{n}
\end{displaymath}
of integer words, where $\N^{0}=\set{\varnothing}$ consists only in
the empty word. The root vertex is the word $\varnothing$. The
children of a node represented by a word $w$ are labelled, in our
binary case, $w0$ and $w1$.
\begin{defn}[Odd cut]\label{def-oddcut}
  Let $\resC$ be a secured Feynman chord diagram. Note that all the
  secured diagrams obtained after the preparation step contain at
  least one tree line and thus at least one tree-resolvent
  $\fres_{0}$. Thanks to the preparation step, there are at least six safe
  elements between $\fres_{0}$ and a half tree edge. An \emph{odd Cauchy-Schwarz cut} starts at $\fres_{0}$ and ends
 between the third and fourth safe element situated between $\fres_{0}$
 and the tree in $\resC$. See \cref{fig-CSoddstep} for a graphical representation.
\end{defn}
\begin{figure}[!htp]
  \centering
  \begin{tikzpicture}
    \node[label=10:$w$] (w) at (0,0)
    {\includegraphics[scale=.8]{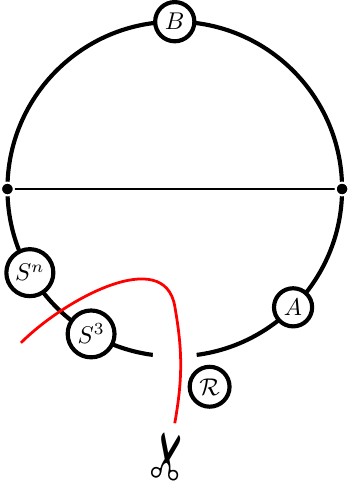}};
    \node[below left=of w, label=below:$w0$] (w0) {\includegraphics[scale=.8]{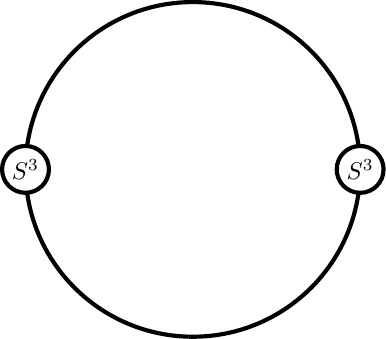}};
    \draw[->] (w) -- (w0);
    \node[below right=of w, label=below:$w1$] (w1) {\includegraphics[scale=.8]{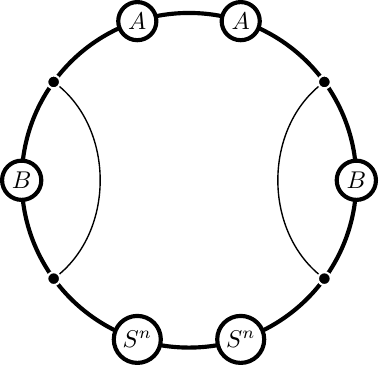}};
    \draw[->] (w) -- (w1);
  \end{tikzpicture}
  \caption{One odd Cauchy-Schwarz iteration ($n\ges 3$). For all $p\ges 0$, $S^{p}$
    represents a product of $p$ safe elements. $A$ and $B$ are
    (almost) any operators.}
  \label{fig-CSoddstep}
\end{figure}
\begin{defn}[Even cut]\label{def-evencut}
  Let $\resC$ be a secured Feynman chord diagram with an even number,
  $2k$, of resolvents. An \emph{even Cauchy-Schwarz cut} consists in
  \begin{enumerate}
  \item choosing any of the resolvents in $A_{\resC}$, calling it
    $\fres_{1}$ and labelling the other ones
    $\fres_{2},\dotsc,\fres_{2k}$ (counter)clockwise around the unique
    vertex of $\resC$,
  \item cutting through $\fres_{1}$ and $\fres_{k+1}$.
  \end{enumerate}
\end{defn}
\begin{defn}[Cutting scheme]\label{def-ICP}
  Let $\secC$ be a secured Feynman chord diagram. We apply
  Cauchy-Schwarz inequalities iteratively as follows:
  \begin{enumerate}
      \setcounter{enumi}{-1}
    \item\label{item-oddstep} if $r(\secC)=2$ or $2k+1$, apply an odd cut. $|A_{\secC}|$ is then
      bounded by the product of (the square roots of the amplitudes
      of) a convergent diagram and a secured diagram with an even
      number ($2$ or $4k$) of resolvents.
    \item\label{item-evenstep} For any diagram with an even number of
      resolvents, perform an even cut and iterate until getting only
      resolvent-free graphs.
    \end{enumerate}
    In the following, graphs obtained from secured ones by such a
  cutting scheme will simply be called \emph{resolvent-free graphs}.
\end{defn}
Now, let $B_{k}$ be the set of binary words (\ie formed from the
alphabet $\set{0,1}$) of length $k$. According to \cref{def-ICP} the
amplitude of a secured chord diagram is bounded above by the following
expressions
\begin{equation}
  \label{eq-CSresult}
  |A_{\resC}|\fide |A_{\varnothing}|\les \norm{\fres}^{2e(\resC)}
  \begin{cases}
    \prod_{w\in B_{k}}|A_{w}|^{2^{-k}}&\text{if $r(\resC)=2k$, $k\ges
      2$,}\\
    |A_{0}|^{1/2}|A_{10}|^{1/4}|A_{11}|^{1/4}&\text{if $r(\resC)=2$},\\
    |A_{0}|^{1/2}\,\prod_{w\in B_{2k}}|A_{1w}|^{2^{-2k-1}}&\text{if $r(\resC)=2k+1$}.
  \end{cases}
\end{equation}
The only (slightly) non trivial factor to explain is
$\norm{\fres}^{2e(\resC)}$. Each cutting step delivers a factor
$\norm{\fres}^{2}$ and the number of steps is bounded above by half the
maximal possible number of resolvents in $A_{\resC}$, namely $2e(\resC)$.\\

Our aim is now to get an upper bound on the amplitude of any secured
graph. Next \namecref{thm-ConvSecuredGraphs} is a first step in this
direction as it proves that such amplitudes are finite.
\begin{lemma}[Convergence of secured graphs]\label{thm-ConvSecuredGraphs}
  Secured graphs are convergent: let $\gls{secG}$ be a secured graph then $|A_{\secG}|<\infty$.
\end{lemma}
To prove it we will need the following
\begin{lemma}[Between resolvents]\label{thm-betweenRes}
  Between two consecutive resolvents of a secured graph amplitude,
  there are either at least three $D$-blocks or at least one half loop
  edge.
\end{lemma}
\begin{proof}
  Remark that between two consecutive resolvents of a skeleton graph
  amplitude, there are either three safe elements or at least one half
  tree edge. This is obvious from
  \crefrange{eq-DerivSigmak1}{eq-developcycles}. During the
  contraction process, unmarked half edges can contract to resolvents
  and thus create graphs such that two consecutive resolvents are only
  separated by one half loop edge. Thus between two consecutive
  resolvents of a resolvent graph amplitude, there are either at least
  three $D$-blocks or at least one half (tree or loop) edge. Let us
  now have a look at the preparation step. When a (tree-)resolvent
  is expanded, it can either merge two intervals between resolvents
  (if the expansion term is $\Itens$) or increase the number of safe
  elements in an interval if the expansion term is a $D$ operator or
  half loop edge or create a new tree edge. In consequence, between two consecutive resolvents of a secured graph amplitude, there are either at least three $D$-blocks or at least one half loop
  edge or at least one half tree edge. In this last case, as
  the graph considered is secured, there are at least six safe
  elements between the two resolvents.
\end{proof}
\begin{proof}[of \cref{thm-ConvSecuredGraphs}]
  We prove that for any word $w$ in $B_{k}$ if $r(\secG)=2k$ or in
  $B_{2k}$ if $r(\secG)=2k+1$, $|A_{w}|<\infty$. Indeed, note first
  that the products of a cut of a secured graph, even or odd, are still secured. Thus the  cutting scheme of \cref{def-ICP} cannot create divergent subgraphs
  as we never cut through a corner adjacent to a tree edge. Then it is
  enough to check that each resolvent-free map $w$ either contains at
  least five tree edges or at least two tree edges and one loop edge
  or at least two loop lines, see \cref{f-dvirdual}.

  If $r(\secG)=1$, we proceed to an odd cut. The resulting
  resolvent-free graphs, denoted $0$ and $1$, contain at least six
  safe elements (see \cref{def-Safeness}) and are thus convergent.

  If $r(\secG)=2$, we split our analysis into two subcases. If the two
  resolvents in $A_{\secG}$ are separated by a tree line, and as
  $\secG$ is secured, an even cut will produce two resolvent-free
  graphs the amplitudes of which contain at least twelve safe elements
  each. They are thus convergent. If one of the two intervals between
  the two resolvents does not contain half tree edges, it must contain
  at least one half loop edge or at least three $D$ operators (by
  \cref{thm-betweenRes}). In this case, we first perform an odd
  cut. It results in two secured graphs. One of them is resolvent-free
  and convergent (see \cref{fig-CSoddstep}). The other one has two
  resolvents separated either by tree edges (thus at least twelve safe
  elements) or by at least two half loop edges. An even cut now
  produces only resolvent-free convergent graphs.

  If $r(\secG)\ges 3$, a resolvent-free graph $w$ necessarily originates
  from the application of an even cut on a secured graph $\sbe{\secG}{1}$ with
  two resolvents. And $\sbe{\secG}1$ itself is the product of an even cut on
  another secured graph $\sbe{\secG}0$ with four resolvents. By
  \cref{thm-betweenRes}, resolvents in $A_{\sbe{\secG}0}$ are separated by
  at least one half loop edge or at least three $D$ operators. Then
  resolvents in $A_{\sbe{\secG}1}$ are separated by at least two half loop
  edges or at least six $D$ operators. An even cut on $\sbe{\secG}1$ thus
  produces only convergent resolvent-free graphs.
\end{proof}

\subsection{Bounds on secured graphs}
\label{sec-bounds-secur-graphs}

Our next task is to get a better upper bound on the amplitude of a secured
graph, in terms of the loop vertex scales. Remember indeed that each
node $a$ of a tree in the \LVEac representation of $\log\cZ$,
see \cref{eq-treerep}, is equipped with a scale $j_{a}$ \ie an integer
between $0$ and $\jm$. Analytically it means that each $V_{j_{a}}$ in
$W_{j_{a}}=e^{-V_{j_{a}}}-1$ contains exactly one $\indic_{j_{a}}$
cutoff adjacent to a $\sqrt C$ operator (and all other propagators
bear $\indic_{\les j_{a}}$ cutoffs), see \cref{eq-nicevj}. Moreover the
scales of the nodes of a Bosonic block are all distinct. After applying the derivatives (situated at both ends of each
tree edge of a Bosonic block) to the $W_{j_{a}}$'s one gets skeleton
graphs which are forests with generically more vertices than their
corresponding abstract tree. Each duplicated vertex is a derivative of
some $W_{j_{a}}$ and bears consequently a $\indic_{j_{a}}$ cutoff. Thus
the (loop) vertices of the skeleton graphs do not have distinct scales
but contain at least as many $(\sqrt C)_{j_{a}}$'s as the
underlying tree.

During the contraction process (\ie integration by parts of the $\sigma$
fields not contained in the resolvents) no $(\sqrt C)_{j}$ operator
are created nor destroyed. When two sigmas contract to each other,
corners (\ie places where square roots of propagators are situated) do
not change. When a sigma field contracts to a resolvent, two new corners
are created but both with a $\indic_{\lj}$ cutoff. The potentially adjacent
$\indic_{j}$ cutoff is left unchanged. Secured graphs bear thus at
least as many $(\sqrt C)_{j_{a}}$'s as their original skeleton
graphs.
\begin{lemma}\label{thm-BoundSecGraphs}
  Let $\cB$ be a Bosonic block and $\secG$ be a secured graph
  originating from $\cB$. Then, there exist $K,\rho\in\R_{+}^{*}$
  such that for any coupling constant $g$ in the cardioid domain $\Card_{\rho}$,
  \begin{equation}
    \label{eq-BoundSecGraphs}
    |A_{\secG}|\les K^{|\cB|}\rho^{e(\secG)}\,\prod_{a\in\cB}M^{-\frac{1}{12}j_{a}}.
  \end{equation}
\end{lemma}
\begin{proof}
  To facilitate the argument we first need to introduce some more
  notation. We let $\ktilde$ be the number of Cauchy-Schwarz iterations
  in the cutting process of \cref{def-ICP}. Explicitly,
  \begin{equation*}
    \ktilde(\secG)\defi
    \begin{cases}
      k&\text{if $r(\secG)=2k$ and $k\ges2$,}\\
      2&\text{if $r(\secG)=2$,}\\
      2k+1&\text{if $r(\secG)=2k+1$.}
    \end{cases}
  \end{equation*}
  We often drop the dependence on $\secG$. In order to track
  corners which bear loop vertex scales, we also introduce the following: let $w$ be either a secured graph or a
  resolvent-free graph. For all $a\in\cB$, we let $c_{a}(w)$ be the
  number of corners of $w$ which bear integer $a$:
  \begin{equation*}
    c_{a}(w)\defi |\set{c\in s(w)\tqs i_{c}=a}|.
  \end{equation*}
  For all $k'\in\set{0}\cup[\ktilde(\secG)]$, let us note
  $F_{k'}(\secG)$ for the set of maps obtained from $\secG$ after $k'$
  steps of the cutting process of \cref{def-ICP}. For example, if $r(\secG)$ is
even and greater than four, $F_{k}(\secG)$ is the set of binary words of
length $r/2$. For all $\gls{mapm}\in F_{k'}(\secG)$, let $\alpha_{k'}(\mapm)$ be
the exponent of $|A_{\mapm}|$ in the corresponding Cauchy-Schwarz bound. Then, according to \cref{eq-CSresult}, for all $a\in\cB$, we define $m_{a,k'}$ as follows:
\begin{equation*}
  m_{a,k'}\defi\sum_{\mapm\in F_{k'}(\secG)}\alpha_{k'}(\mapm)c_{a}(\mapm).
\end{equation*}
We shall now bound the amplitude of $\secG$ by a
multiscale analysis. It means that for all $\mapm$ in $F(\secG)$, we expand each $(\sqrt C)_{\lj}$ operator as
$\sum_{i=0}^{j}(\sqrt C)_{i}$. Each map $\mapm$ is then equipped with a
scale attribution, namely a given integer per corner of $\mapm$. These
attributions correspond to the usual scale attributions on edges in
the tensor graph representation. Nevertheless, they are here
constrained: there exist (marked) corners with a fixed scale $j_{a}$ (these
are the loop vertex scales) and for each corner $c$ of $\mapm$, $i_{c}$ is
bounded by some $j_{a}$. Let $s(\mapm)$ be the set of marked corners of $\mapm$. Using \cref{thm-ConvSecuredGraphs,thm-PowCountSpare}, there exists
a positive real number $K$ such that
\begin{equation}\label{eq-BoundResFreeGraph}
  |A_{\mapm}|=|\sum_{\mu}A_{\mapm_{\mu}}|\les (K|g|)^{e(\mapm)}\,\prod_{c\in s(\mapm)}M^{-\frac{1}{24}i_{c}}.
\end{equation}
Remember that \cref{thm-PowCountSpare} is formulated in the tensor
graph representation. Here the edges of a chord diagram correspond to
the vertices of a tensor Feynman graph and edges of the latter are the
corners in the \rfgs. Moreover, looking at
\crefrange{eq-DerivSigmak1}{eq-developcycles}, one notices that each
loop vertex bears one factor $\lambda=g^{1/2}$ per corner (in a
\rfg). This explains the term $g^{e(\mapm)}$ in
\cref{eq-BoundResFreeGraph}. From \cref{eq-CSresult}, we deduce
\begin{equation}\label{eq-BoundSecuredOne}
  |A_{\secG}|\les
  \norm{\fres}^{2e(\secG)}(K|g|)^{\sum_{\mapm\in
      F_{\tilde k}(\secG)}\alpha_{\tilde
      k}(\mapm)e(\mapm)}\prod_{a\in\cB}M^{-\frac{1}{24}m_{a,\tilde k}j_{a}}.
\end{equation}
Remark that
\begin{equation}
  \label{eq-TrackEdges}
  \sum_{\mapm\in F_{\tilde k}(\secG)}\alpha_{\tilde k}(\mapm)e(\mapm)=e(\secG).
\end{equation}
Let us indeed consider $w\in F_{k'}(\secG)$ with
$0\les k'\les\ktilde$ and any edge $\ell$ of $w$. If the
$(k'+1)^{\text{th}}$ \CS iteration cuts $\ell$, then it appears
exactly once both in $w0$ and $w1$. If $\ell$ is not cut, it appears
twice in $w0$ or $w1$ but not in both graphs. In the two cases,
$e(w)=\frac 12(e(w0)+e(w1))$. Induction on $k'$ proves
\cref{eq-TrackEdges} as $\sum_{\mapm\in F_{0}(\secG)}\alpha_{0}(\mapm)e(\mapm)=e(\secG)$.

Let us now prove that for all $a\in\cB$, $m_{a,\tilde k}\ges 2$. Let us consider a fixed $a$ in $\cB$ and $k'$ between $0$ and
$\ktilde$. Let $w$ be a map in $F_{k'}(\secG)$. We define
$c_{a,r}(w)$ as the number of marked corners of $w$ of scale $a$ which
are adjacent to a resolvent. We also let $c_{a,f}(w)$ be
$c_{a}(w)-c_{a,r}(w)$. We further decompose $c_{a,r}(w)$ as
$c_{a,c}(w)+c_{a,s}(w)$ where $c_{a,c}(w)$ is the number of corners,
adjacent to a resolvent, and adjacent to the cut at step $k'$. Let
now $c$ be a marked corner in $s(w)$ such that $i_{c}=a$. If $c$ is adjacent to a resolvent cut
at the $(k')^{\text{th}}$ step, it appears in exactly one graph
among $w0$ and $w1$. If $c$ is not adjacent to a resolvent but
nevertheless cut (thus by an odd cut), it belongs to both $w0$ and
$w1$. If $c$ is not cut, it appears twice either in $w0$ or in $w1$
but not in both. Then
\begin{equation*}
  \left .
  \begin{aligned}
    c_{a,f}(w0)+c_{a,f}(w1)&=2c_{a,f}(w)+c_{a,c}(w)\\
    c_{a,r}(w0)+c_{a,r}(w1)&=2c_{a,s}(w)
  \end{aligned}
  \rb\Rightarrow c_{a}(w0)+c_{a}(w1)=2c_{a}(w)-c_{a,c}(w).
\end{equation*}
As $\alpha_{k'+1}(w0)=\alpha_{k'+1}(w1)=\frac 12\alpha_{k'}(w)$, we have
\begin{equation*}
  \alpha_{k'+1}(w0)c_{a}(w0)+\alpha_{k'+1}(w1)c_{a}(w1)=\alpha_{k'}(w)c_{a}(w)-\tfrac
12c_{a,c}(w).
\end{equation*}
Then, viewing the cutting process of \cref{def-ICP} as a
computation tree $T$, and resumming $m_{a,\tilde k}$ from the leaves to
the root of $T$, one gets
\begin{equation*}
  m_{a,\tilde k}=m_{a,0}-\frac 12\sum_{k'=0}^{\tilde k-1}\,\sum_{w\in
    F_{k'}(\secG)}c_{a,c}(w)=c_{a}(\secG)-\tfrac 12c_{a,r}(\secG).
\end{equation*}
As $c_{a,r}(\secG)\les c_{a}(\secG)$, $m_{a,\tilde k}\ges\frac12
c_{a}(\secG)$. Remembering that, as discussed at the begining of
\cref{sec-pert-funct-integr}, any resolvent graph has at least four
marked corners of each loop vertex scale (said differently
$c_{a}(\secG)\ges 4$ for all $a\in\cB$), $m_{a,\tilde k}\ges 2$. To
conclude the proof, we use
\begin{itemize}
\item this bound on $m_{a,\tilde k}$ as well as \cref{eq-TrackEdges} in
  \cref{eq-BoundSecuredOne},
\item the fact that $e(\secG)$ grows at most linearly with $|\cB|$,
\item the fact that $e(\skG'')\ges 4$,
\item the resolvent bound of \cref{thm-lemmaresbounded} and the
  definition of the cardioid domain $\Card_{\rho}$.
\end{itemize}
\end{proof}

The main goal of \cref{sec-pert-funct-integr} was to give an upper
bound on the perturbative term $I_{4}$ of
\cref{eq-CS-Pert-NonPert}. Here it is:
\begin{thm}[Perturbative factor $I_{4}$]\label{thm-BoundI4}
  Let $\cB$ be a Bosonic block and define $n$ by $|\cB|\fide n+1$, $n\ges
  0$. Then there exists $K\in\R_{+}^{*}$ such that the perturbative factor $I_{4}$ of \cref{eq-CS-Pert-NonPert} obeys
  \begin{equation*}
    I_{4}(\cB;\skG)\les K^{n}(n!)^{37/2}\rho^{x(\skG)}\,\prod_{a\in\cB}M^{-\frac
    1{48}j_{a}},\quad x(\skG)=
  \begin{cases}
    e(\skG)&\text{if $e(\skG)\ges 1$}\\
    2&\text{otherwise}.
  \end{cases}
  \end{equation*}
\end{thm}
\begin{proof}
  We concentrate here on the case of Bosonic blocks with more than one
  node. Summing up what we have done in this Section, we have
\begin{equation*}
  I_{4}^{4}=\int d\nu_{\cB}\,\sum_{\resG(\skG'')}\sum_{\secG(\resG)}A_{\secG}.
\end{equation*}
The functional integration \wrt the measure $\nu_{\cB}$ equals $1$ as
the integrand does not depend on $\sigmad$
anymore. \Cref{thm-BoundSecGraphs} gives a bound on
$A_{\secG}$ and \Cref{thm-NumberOfLeaves} a bound on the number of
terms in the sum over $\secG$. There remains to bound the number of
resolvent graphs $\resG$ obtained from the contraction process applied to a given
skeleton graph $\skG''$. Then, as already discussed at the begining of this Section,
$n(\skG'')\les 8n$ and $e(\skG'')\les 4n$. Thus $r(\skG'')\les 8n$ and
as loop vertices bear at most three sigmas, the total number of sigma
fields to be integrated by parts in the contraction process is bounded
above by $24n$. We deduce that the number of terms in the sum over
$\resG(\skG'')$ is bounded by $K^{n}(n!)^{32}$. All in all, we get
\begin{equation*}
  I_{4}^{4}\les
  K^{n}(n!)^{74}\rho^{e(\skG'')}\,\prod_{a\in\cB}M^{-\frac
    1{12}j_{a}}\  \Rightarrow\  I_{4}\les K^{n}(n!)^{37/2}\rho^{e(\skG)}\,\prod_{a\in\cB}M^{-\frac
    1{48}j_{a}}
\end{equation*}
where we used that $e(\secG)\ges e(\skG'')=4e(\skG)$. The final bound
is obtained by noticing that the possible vertices of a single node
Bosonic block bear at least two powers of $\rho$.
\end{proof}

\section{The final sums}
\label{sec-final-sums}

We are now ready to gather the perturbative and non perturbative
bounds of Sections \ref{sec-pert-funct-integr} and \ref{sec-funct-integr-bounds} into a
unique result on $\log\cZ_{\les\jm}$. Our starting point is the
expression of $\log\cZ_{\les\jm}$ obtained after application of the \MLVE:
\begin{multline}\tag{\ref{eq-treerep}}
  \cW_{\les\jm}(g)=\log \cZ_{\les\jm}(g)=
  \sum_{n=1}^\infty \frac{1}{n!}  \sum_{\cJ\text{ tree}}
  \,\sum_{j_1=1}^{\jm} 
  \dotsm\sum_{j_n=1}^{\jm}\\
  \int d\tuple{w_{\ccJ}} \int d\nu_{ \ccJ}  
  \,\partial_{\ccJ}   \Bigl[ \prod_{\cB} \prod_{a\in \cB}   \bigl(   -\bar \chi^{\cB}_{j_a}  W_{j_a}   (\sigmad^a , \taud^a )  
  \chi^{ \cB }_{j_a}   \bigr)  \Bigr].
\end{multline}
Then we need to remember that the functional derivative
$\partial_{\cJ}$, see \cref{eq-partialJ}, is the product of Fermionic
and Bosonic derivatives, \cref{eq-partialJ-product}, and that the
latter factor out over the Bosonic components of the tree $\cJ$. Then,
as in \cite{Gurau2014ab}, we start by estimating the functional
integration over the Grassmann variables to get:
  \begin{align*}
 \vert    \log\cZ_{\les\jm}  \vert   &\les\sum_{n=1}^\infty \frac{2^{n}}{n!} \sum_{\cJ\text{
        tree}}\,\sum_{\set{j_{a}}} \Bigl( \prod_{\cB}
  \prod_{\substack{a,b\in \cB\\a\neq b}} (1-\delta_{j_aj_b})
  \Bigr) \Bigl( \prod_{\substack{\ell_F \in
      \cF_F\\\ell_F=(a,b)}} \delta_{j_{a } j_{b } } \Bigr)\
  \prod_{\cB}|I_{\cB}|,\\
  I_{\cB}&= \int d\tuple{w_{\cB}}\int d\nu_{\cB}\,\partial_{\cT_{\cB}}\prod_{a\in \cB}  ( e^{-V_{j_a}} -1  ) (\sigmad^a, \taud^a).
  \end{align*}
  Using the language of skeleton graphs, applying Hölder inequality
  \eqref{eq-CS-Pert-NonPert} and using notation of
  \cref{eq-def-IBNP,thm-BoundI4}, we get
  \begin{align*}
   \vert    \log\cZ_{\les\jm}   \vert   &\les\sum_{n=1}^\infty \frac{\Oun^{n}}{n!} \sum_{\cJ\text{
        tree}}\,\sum_{\set{j_{a}}} \Bigl( \prod_{\cB}
    \prod_{\substack{a,b\in \cB\\a\neq b}} (1-\delta_{j_aj_b})
    \Bigr)\
    \prod_{\cB}\sum_{\skG(\cB)}I_{\cB}^{\mathit{NP}}I_{4}(\cB;\skG).
  \end{align*}
Let us introduce $n_\cB := \vert \cB \vert \ges 1$, which is therefore 1 plus the integer called $n$ in \cref{thm-BoundI4}.
The sum over skeleton graphs $\skG(\cB)$ can be decomposed into two
parts. Due to Faà di Bruno formula, there is a first sum over
partitions of the sets of edges incident to each vertex of
$\cT_{\cB}$. 
The total number of such partitions is bounded above by
$\Oun^{n_\cB}(n_\cB!)^{2}$. Given such partitions, there remains to
choose appropriate loop vertices for each vertex of $\skG$. As the
number of terms in the $k^{\text{th}}$ $\sigma$-derivative of
$V_{j}^{\ges 3}$ is bounded by $25\,k!$ the number of possible choices
of loop vertices is bounded by $\Oun^{n_\cB}(n_\cB!)^{2}$. Then,
\begin{align*}
  \vert     \log\cZ_{\les\jm}  \vert  &\les\sum_{n=1}^\infty \frac{\Oun^{n}}{n!} \sum_{\cJ\text{
        tree}}\,\sum_{\set{j_{a}}} \Bigl( \prod_{\cB} (n_\cB!)^{4}
    \prod_{\substack{a,b\in \cB\\a\neq b}} (1-\delta_{j_aj_b})
    \Bigr)\
    \prod_{\cB}    I_{\cB}^{\mathit{NP}}\sup_{\skG}I_{4}(\cB;\skG)
  \end{align*}
and using \cref{thm-npBound,thm-BoundI4} we have
\begin{equation}  \label{mainbound}
    \vert   \log\cZ_{\les\jm}  \vert  \les\sum_{n=1}^\infty \frac{\Oun^{n}}{n!}   \rho^{X} \sum_{\cJ\text{
        tree}}\,\sum_{\set{j_{a}}} \Bigl( \prod_{\cB} (n_\cB!)^{4+37/2}
    \prod_{\substack{a,b\in \cB\\a\neq b}} (1-\delta_{j_aj_b})
    \Bigr)\ \prod_{a=1}^{n}M^{-\frac 1{48}j_{a}}
  \end{equation} 
where $X\defi\sum_{\cB}\sup_{\skG}x(\skG)$. As $x(\skG)=|\cB|-1$ if
$|\cB|\ges 2$ and $x(\skG)=2$ if $|\cB|=1$, we have $X\ges\ceil{\tfrac
n2}$.\\

The factor $ \prod_{\cB}    \prod_{\substack{a,b\in \cB\\a\neq b}} (1-\delta_{j_aj_b})$ in \cref{mainbound} ensures that
slice indices $j_a$ are all \emph{different} in each block $\cB$.
Therefore
\begin{equation*}
  \sum_{a \in \cB}  j_a \ges 1 + 2 + \cdots + n_\cB  = \frac{n_\cB(n_\cB+1)}{2},
\end{equation*}
so that
\begin{equation*}
\prod_{a =1}^n M^{-j_a / 96} \les   \prod_{\cB}    e^{- \Oun n_\cB^2}.
\end{equation*}

The number of labeled trees on $n$ vertices is $n^{n-2}$ (the complexity of the complete graph 
$K_n$ on $n$ vertices), hence the 
number of two-level trees $\cJ$ in \cref{mainbound} is exactly $2^{n-1}n^{n-2}$. Since
$\sum_\cB n_\cB = n$, for $\rho$ small enough we have
\begin{align*}  
   \vert    \log\cZ_{\les\jm} \vert  &\les\sum_{n=1}^\infty\Oun^{n} \rho^{n/2} \sup_{\cJ\text{
        tree}}\, \Bigl( \prod_{\cB}  (n_\cB!)^{4+37/2}  e^{- \Oun n_\cB^2}
    \Bigr)\ \sum_{\set{j_{a}}} \prod_{a=1}^{n}M^{-\frac 1{96}j_{a}}\\
&\les\sum_{n=1}^\infty\Oun^{n}\rho^{n/2} < + \infty .
  \end{align*}
Hence for $\rho$ small enough the series \eqref{eq-treerep} is
absolutely and uniformly convergent in the cardioid domain
$\Card_\rho$. Analyticity, Taylor remainder bounds and Borel
summability follows for each $\cW_{\jm}$ (uniformly in $\jm$) from
standard arguments based on Morera's theorem. Similarly, since the
sequence $\cW_{\jm}$ is easily shown uniformly Cauchy in the cardioid
(from the geometric convergence of our bounds in $\jm$, the limit
$\cW_{\infty}$ exists and its analyticity, uniform Taylor remainder bounds and Borel summability follow again
from similar standard application of Morera's theorem. This completes the proof of our
main result, \cref{thetheorem}.

\section*{Conclusion}
\label{sec-conclusion}
\etoctoccontentsline*{section}{Conclusion}{1}

Uniform Taylor remainder estimates at order $p$ are required to complete 
the proof of Borel summability \cite{Sokal1980aa} in Theorem \ref{thetheorem}. 
They correspond to further Taylor expanding beyond trees up to graphs with \emph{excess} (\ie number of cycles) at most $p$. The corresponding 
\emph{mixed expansion} is described in detail in \cite{Gurau2013ac}. 
The main change is to force for an additional
$p!$ factor to bound the cycle edges combinatorics, 
as expected in the Taylor uniform remainders estimates of a Borel summable function. 

The main theorem of this paper clearly also extends to cumulants of the theory,
introducing \emph{ciliated} trees and graphs as in 
\cite{Gurau2013ac}. 
This is left to the reader. 

The next tasks in constructive tensor field theories would be to treat
the $T^4_5$ \cite{BenGeloun2013aa} and the $U(1)-T^4_6$ group field theory \cite{Ousmane-Samary2012ab}. They are both just renormalizable 
and asymptotically free \cite{Ben-Geloun2012aa,Rivasseau2015aa}.
Their full construction clearly requires more precise estimates,
but at this stage we do not foresee any reason it cannot be done via 
the strategy of the MLVE.

\newpage
\appendix
\section*{Appendices}
\etoctoccontentsline*{section}{Appendices}{1}

\renewcommand{\thesection}{A}
\renewcommand{\thesubsection}{A.\arabic{subsection}}
\renewcommand{\theequation}{A.\arabic{equation}}
\setcounter{equation}{0}

\etocsetnexttocdepth{2}
\renewcommand{\contentsname}{}
\etocsetstyle{subsection}%
{\vspace{-.5cm}\begin{changemargin}{1cm}{1cm}}
{\vspace*{-4pt}\leavevmode\leftskip 0cm \rightskip .6cm\relax}%
{\normalfont\normalsize%
\makebox[24pt][l]{\etocifnumbered{\etocnumber.}{\etocnumber}}%
 \etocname\dotfill{}\rlap{\makebox[.6cm][r]{\etocpage}}\par}%
{\end{changemargin}}
\localtableofcontents

\subsection{Bare and renormalized amplitudes}
\label{sec-bare-renorm-ampl}

Let $\cM_{1}^{c}$ denote the graph $\cM_{1}$ with a vertex of colour
$c$. Its regularized bare amplitude $A_{\cM_{1}^{c}}$ is a function of
the incoming index $n_{c}$ (and appears at first order of the Taylor
expansion in $g$)\footnote{By convention, we do not include the powers
  of the coupling constant $g$ in the amplitudes of the Feynman graphs.}:
\begin{equation*}
A_{\cM_1^{c}} =  -\sum_{\tuple{p}\in  [-N, N]^4}  \frac{\delta (p_c -n_c) }{\tuple{p}^2 + 1}.
\end{equation*}
The counterterm is minus the value at $n_c=0$, hence 
\begin{equation*}
\delta_{\cM_{1}^{c}} =  \sum_{\tuple{p}  \in  [-N, N]^4}  \frac{\delta(p_c) }{\tuple{p}^2 + 1} = \sum_{\tuple{p}
  \in  [-N, N]^3}  \frac{1}{\tuple{p}^2 + 1}.
\end{equation*}
Let us define $A_{\cM_{1}}$ to be $\sum_{c}A_{\cM_{1}^{c}}$. Remark that $A_{\cM_{1}^c}$ is independent of $c$, so that in fact 
\begin{equation*}
\delta_{\cM_{1}}\defi\sum_{c}\delta_{\cM_{1}^{c}} = 4\sum_{\tuple{p}  \in  [-N, N]^3}  \frac{1}{\tuple{p}^2 + 1} .
\end{equation*}
We can calculate the renormalized amplitude of $\cM_1$ as
\begin{equation*}
\Ar(\cM_1) = - \sum_c  \sum_{\tuple{p}   \in {\mathbb {Z}}^4}   \frac{\delta (p_c -n_c)  - \delta (p_c) } 
{\tuple{p}^2+ 1} = \sum_c  \sum_{\tuple{p}   \in {\mathbb{Z}}^3}   \frac{ n_c^2} {(n_c^2 + \tuple{p}^2+ 1)(\tuple{p}^2 +1)}.
\end{equation*}
It is now a convergent sum, hence no longer requires the cutoff $N$.\\

We now compute the counterterm to the graph $\cM_{2}$. Remark that
this $\log$ divergent mass graph has the tadpole $\cM_{1}$ as
a subgraph, hence its counterterm has to include the
subrenormalization of that tadpole. The bare amplitude of $\cM_2$ is 
\begin{equation*}
A_{\cM_2} =  \sum_c   \sum_{\tuple{p}  \in  [-N, N]^4 }
\frac{\delta (p_c -n_c) }{(\tuple{p}^2 + 1)^2}  \sum_{c'} \sum_{\bq
  \in  [-N, N]^4 }  \frac{\delta (q_{c'} -p_{c'})}{\bq^2 + 1}
.
\end{equation*}
The partly renormalized amplitude of $\cM_2$ (with only the inner tadpole subtraction) is
\begin{align*}
  A^{\text{p.\!\! ren}}_{\cM_2} &=   \sum_c   \sum_{\tuple{p}  \in  [-N, N]^4 }  \frac{\delta (p_c -n_c) }{(\tuple{p}^2 + 1)^2}  \sum_{c'} \sum_{\bq  \in  {{\mathbb Z}}^4 }  \frac{\delta (q_{c'} -p_{c'})  -  \delta (q_{c'} )   }    {\bq^2 + 1} \nonumber\\
  &= - \sum_c \sum_{\tuple{p} \in [-N, N]^4 } \frac{\delta (p_c
    -n_c) }{(\tuple{p}^2 + 1)^2} \sum_{c'} \sum_{\bq \in {\mathbb
      Z}^3} \frac{ p_{c'}^2 } {(p_{c'}^2 + \bq^2 + 1)(\bq^2 + 1)},
\end{align*}
where we relaxed the cutoff constraint on the inner tadpole to better
show that it is now a convergent sum. Hence the counterterm for  $\cM_2$ is $\delta_{\cM_{2}}
\defi \sum_{c} \delta_{\cM_{2}^{c}}$ and
\begin{align*}
  \delta_{\cM_{2}^{c}} &=-\sum_{\tuple{p} \in [-N, N]^4 } \frac{\delta
    (p_c ) }{(\tuple{p}^2 + 1)^2} \sum_{c'} \sum_{\bq \in {\mathbb
      Z}^4 } \frac{\delta (q_{c'} -p_{c'}) - \delta (q_{c'} ) } {\bq^2
    + 1}
  \nonumber\\
  &=  \sum_{\tuple{p}  \in  [-N, N]^4 }  \frac{\delta (p_c ) }{(\tuple{p}^2 + 1)^2}   \sum_{c'} \sum_{\bq  \in  {{\mathbb Z}}^3}  \frac{ p_{c'}^2 }    {(p_{c'}^2 + \bq^2 + 1)(\bq^2 + 1)} \nonumber\\
  &= 3 \sum_{\tuple{p} \in [-N, N]^3} \frac{p_1^2} {(\tuple{p}^2+
    1)^2} \sum_{\bq \in {{\mathbb Z}}^3} \frac{1} {(p_1^2 + \bq^2 +1)
    (\bq^2+ 1)}
\end{align*}
where in the last line we used that the value vanishes if $c=c'$, plus
again the colour symmetry. The renormalized amplitude for $\cM_2$ is the now fully convergent double sum
\begin{align*}
  \Ar_{\cM_2} ={} & \sum_c \sum_{\tuple{p}
    \in{\mathbb Z}^4 } \frac{\delta (p_c -n_c) - \delta (p_c )
  }{(\tuple{p}^2 + 1)^2} \sum_{c'} \sum_{\bq \in {{\mathbb Z}}^4 }
  \frac{\delta (q_{c'} -p_{c'}) - \delta (q_{c'} ) } {\bq^2 + 1}
  \nonumber\\
  ={}& - \sum_c \sum_{\tuple{p} \in{\mathbb Z}^4 } \frac{\delta (p_c
    -n_c) }{(\tuple{p}^2 + 1)^2} \sum_{\bq \in {{\mathbb Z}}^3} \frac{
    p_{c}^2 } {(p_{c}^2 + \bq^2 + 1)(\bq^2 + 1)} \nonumber\\ 
  &-\sum_c \sum_{\tuple{p} \in{\mathbb Z}^4 } \frac{\delta (p_c -n_c) -
    \delta (p_c ) }{(\tuple{p}^2 + 1)^2} \sum_{c' \ne c} \sum_{\bq \in
    {{\mathbb Z}}^3} \frac{ p_{c'}^2 } {(p_{c'}^2 + \bq^2 + 1)(\bq^2 +
    1)}\nonumber\\
  ={} & \sum_c n_c^2 \sum_{\tuple{p} \in {\mathbb
      Z}^3} \sum_{\bq \in {{\mathbb Z}}^3} \biggl( \frac{ 3 p_1^2 [
    n_c^2 + 2 (\tuple{p}^2 + 1) ] } {(n_c^2 + \tuple{p}^2 + 1)^2
    (\tuple{p}^2 + 1)^2 (p_1^2 + \bq^2 + 1)(\bq^2 +
    1)}\nonumber\\
  &- \frac{ 1 }{(n_c^2 + \tuple{p}^2 + 1)^2(n_c^2 + \bq^2 + 1)(\bq^2
    + 1)} \biggr).
\end{align*}

\subsection{The \texorpdfstring{$Q$}{Q} operator}
\label{sec-Q-operators}

We gather here some easy but useful results over the operators
$\gls{Q0}$ and $\gls{Q1}$, which are well-defined bounded operators on $L^2 (\{1,2,3,4\} \times
\Z^2)$. From their definitions in the momentum basis (which has been used throughout this paper), see \cref{eq-Q0expr,eq-Qexpr}, it is easy to bound the coefficients of these operators:
\begin{align*}
  (\gls{Q0})_{c,c';mn,m'n'}&\les\delta_{cc'}\delta_{mm'} \delta_{nn'} \frac{\Oun }{(m^2 + n^2 +1)^{1/2}},
  \\
  (\cst Q11)_{c,c';mn,m'n'}&\les(1-\delta_{cc'})\delta_{mn}
  \delta_{m'n'}\frac{\Oun }{m^2 + m'^2 +1},
  \\
  |(\cst Q12)_{c,c';mn,m'n'}|&\les\delta_{cc'}\delta_{mm'}
  \delta_{nn'} \frac{\Oun }{(m^2 + n^2 +1)^{3/2}}
  \\
  &\phantom{\les
    {}}+(1-\delta_{cc'})\delta_{mn}\delta_{m'n'}\frac{\Oun }{(m^2 +
    m'^2 +1)^{2}}.\nonumber
\end{align*}
Hence $Q_{0}$ is a bounded operator and \eg
$\norm{gQ_{0}} \les\tfrac 14$ for $g$ in the
cardioid and $\rho$ small enough. $Q_{1}$ is trace class therefore
$Q=Q_{0}+Q_{1}$ is bounded with $\norm{gQ}\les\tfrac 12$ for
$\rho$ small enough. Remark that $Q$ itself, without ultraviolet cutoff, is not a trace class operator, since $Q_{0}$ is not trace class: indeed 
$\sum_{m,n} \frac{\Oun }{\sqrt{m^2 + n^2 +1}} $ diverges linearly.
\begin{lemma}\label{thm-Qj}
  In the cardioid,
  \begin{alignat*}{3}
    \norm{Q_{j}}&\les\Oun M^{-j},&\qquad\Tr Q_{j}&\les\Oun M^{j},\\
    \norm{Q_{0,j}}&\les\Oun M^{-j},&\Tr Q_{0,j}&\les\Oun M^{j},&\qquad\Tr[Q_{0,j}^{2}]&\les\Oun,\\
    \norm{\cst{Q}{1,j}1}&\les\Oun,&\Tr\cst{Q}{1,j}1&\les\Oun,&\Tr[(\cst{Q}{1,j}1)^{2}]&\les\Oun
    M^{-j},\\
    \norm{\cst{Q}{1,j}2}&\les\Oun \rho M^{-j},&\Tr\cst{Q}{1,j}2&\les\Oun \rho M^{-j}.
  \end{alignat*}
\end{lemma}
\begin{proof}
  Simple exercise by noting that $Q_{0}$ is diagonal and that for any
  bounded operator $P$, $\norm P\les \Tr P$.
\end{proof}

\subsection{Proof of \cref{thm-finitevacuum}}
\label{sec-vacuum-contrib}

  Recollecting the definitions of $\cN,\cN_{1},\cN_{2},\cN_{3},\cN_{4}$ and
  $\cN_{5}$, we get
  \begin{multline*}
    \log\cN_{5}=\sum_{G\in\cV'}\tfrac{(-g)^{|G|}}{S_{G}}\delta_{G}+\delta_{t}+\tfrac
    12\direct{B}_{1}^{2}-i\lambda^{3}\sum_{c}\delta_{\cM_{2}^{c}}\Tr_{c}(\direct{B}_{1})_{c}+\tfrac
    12\Tr D_{1}^{2}\\
    +\tfrac
    12\direct{B}_{2}^{2}+\tfrac 12\Tr[D_{2}^{2}]+\tfrac13\Tr[D_{1}^{3}]+\Tr[D_{1}^{2}D_{2}]+\tfrac14\Tr[D_{1}^{4}]
  \end{multline*}
  where $\cV'=\cV\setminus\set{\kN_{1},\kN_{2},\kN_{3}}$ and $\delta_{t}=\tfrac g2(2N+1)\sum_{c}(\delta_{m}^{c})^{2}$. We need to prove
  that $\log\cN_{5}=0$.\\

Let $\Fprep(G)\defi\set{\text{divergent forests }\cF\text{ of }G\tqs G\notin\cF}$ and $\textit{\gls{Aprep}}_{G}\defi \sum_{\cF\in\Fprep(G)}\big(\prod_{g\in\cF}-\tau_{g}\big)A_{G}$.
Recall that for $G$ a vacuum graph, $\delta_{G}=-\Aprep_{G}$. We will prove that
  $\log\cN_{5}-\sum_{G\in\cV'}\tfrac{(-g)^{|G|}}{S_{G}}\delta_{G}$
  equals
  $\sum_{G\in\cV'}\tfrac{(-g)^{|G|}}{S_{G}}\Aprep_{G}$. In the
  following, in order to lighten notations, we will denote $\cM_{1}$
  (resp.\@ $\cM_{2}$) by $\scalebox{.8}{\nuonefig}$ (\resp $\scalebox{.8}{\nutwofig}$) or any rotated version of it. Thus we have
  \begin{align*}
    \Fprep(\cV_{1})&=\set{\varnothing,\set{\nuonefig},\set{\nuonefig[180]}},\qquad\Aprep_{\cV_{1}}=(1-\tau_{\nuonefigtau}-\tau_{\nuonefigtau[180]})A_{G},
  \end{align*}
    \begin{align*}
      \Fprep(\cV_{2})&=\lb\varnothing,\set{\nuonefig},\set{\nuonefig[180]},\set{\nutwofig},\set{\nutwofig[180]},\set{\nuonefig,\nuonefig[180]},\set{\nuonefig,\nutwofig},\set{\nuonefig[180],\nutwofig[180]}\right\},\\
    \Aprep_{\cV_{2}}&=(1-\tau_{\nuonefigtau})(1-\tau_{\nuonefigtau[180]})A_{\cV_{1}}-\tau_{\nutwofigtau}(1-\tau_{\nuonefigtau})A_{\cV_{2}}-\tau_{\nutwofigtau[180]}(1-\tau_{\nuonefigtau[180]})A_{\cV_{2}},
  \end{align*}
  \begin{align*}
    \Fprep(\cV_{3})&=\lb\varnothing,
    \set{\nuonefig},\set{\nuonefig[180]},\set{\nutwofig},\set{\nutwofig[180]},\set{\nuonefig,\nuonefig[180]},\set{\nuonefig,\nutwofig},\set{\nuonefig[180],\nutwofig[180]},\set{\nuonefig,\nutwofig[180]},\set{\nuonefig[180],\nutwofig},\set{\nuonefig,\nutwofig,\nuonefig[180]},\set{\nuonefig,\nutwofig[180],\nuonefig[180]}\rb,\\
    \Aprep_{\cV_{3}}&=(1-\tau_{\nutwofigtau})(1-\tau_{\nuonefigtau})(1-\tau_{\nutwofigtau[180]})(1-\tau_{\nuonefigtau[180]})A_{\cV_{3}}-\tau_{\nutwofigtau}\tau_{\nutwofigtau[180]}(1-\tau_{\nuonefigtau}) (1-\tau_{\nuonefigtau[180]})A_{\cV_{3}}.
  \end{align*}
The remaining vacuum graphs are easier to handle since they do not
have any overlapping divergences:
\begin{align*}
  \Aprep_{\cV_{4}}&=(1-\tau_{\nutwofigtau})(1-\tau_{\nuonefigtau})
  (1-\tau_{\nutwofigtau[180]})(1-\tau_{\nuonefigtau[180]})A_{\cV_{4}},\\
  \Aprep_{\cV_{5}}&=(1-\tau_{\nuonefigtau})(1-\tau_{\nuonefigtau[120]})
  (1-\tau_{\nuonefigtau[240]})A_{\cV_{5}},\\
  \Aprep_{\cV_{6}}&=(1-\tau_{\nutwofigtau})(1-\tau_{\nuonefigtau})(1-\tau_{\nuonefigtau[120]})
  (1-\tau_{\nuonefigtau[240]})A_{\cV_{6}},\\
  \Aprep_{\cV_{7}}&=(1-\tau_{\nuonefigtau[-135]}) (1-\tau_{\nuonefigtau[135]}) (1-\tau_{\nuonefigtau[45]})(1-\tau_{\nuonefigtau[-45]})A_{\cV_{7}}.
\end{align*}
Using $\delta_{m}^{c}=-\delta_{\cM^{c}_{1}}+\lambda^{2} \delta_{\cM^{c}_{2}}$, it is then easy to check that:
\begin{align*}
  -\tfrac{\lambda^{2}}{2}\Aprep_{\cV_{1}}&=\tfrac
  12\direct{B}_{1}^{2}+\tfrac{\lambda^{2}}{2}
  (2N+1)\sum_{c}(\delta_{\cM_{1}^{c}})^{2},\\
  \tfrac{\lambda^{4}}{2}\Aprep_{\cV_{2}}&=\tfrac 12\Tr
  D_{1}^{2}-i\lambda^{3}\sum_{c}\delta_{\cM_{2}^{c}}\Tr_{c}(\direct{B}_{1})_{c}-\lambda^{4}(2N+1)\sum_{c}\delta_{\cM_{1}^{c}}\delta_{\cM_{2}^{c}},\\
  -\tfrac{\lambda^{6}}{2}\Aprep_{\cV_{3}}&=-\tfrac{\lambda^{6}}{2}\direct{\Ar}_{\cM_{2}}\scalprod
  \direct{\Ar}_{\cM_{2}}+\tfrac{\lambda^{6}}{2}(2N+1)\sum_{c}(\delta_{\cM_{2}^{c}})^{2},\\
  \tfrac{\lambda^{8}}{2}\Aprep_{\cV_{4}}&=\tfrac12\Tr[D_{2}^{2}]
  ,\qquad -\tfrac{\lambda^{6}}{3}\Aprep_{\cV_{5}}=\tfrac 13\Tr
  D_{1}^{3},\\
  \lambda^{8}\Aprep_{\cV_{6}}&=\Tr(D_{1}^{2}D_{2}),\qquad
    \tfrac{\lambda^{8}}{4}\Aprep_{\cV_{7}}=\tfrac 14\Tr
  D_{1}^{4}.
\end{align*}
In other words,
\begin{equation*}
  \log\cN_{5}=\sum_{G\in\cV'}\tfrac{(-g)^{|G|}}{S_{G}}\delta_{G}+\sum_{G\in\cV'}\tfrac{(-g)^{|G|}}{S_{G}}\Aprep_{G}=0 
\end{equation*}
which is the desired result.

\subsection{Proof of \cref{lemmaquarticbound}}
\label{sec-proof-quartic}

  We start from expression \eqref{eq-startvjeqnice} of $v_{j}$ and
  first expand $\Itens-\fres$ as $-U\fres$ or $-\fres U$:
  \begin{align*}
    v_{j}^{(1)}&=-\Tr\bigl[(U\unj D\unj\Sigma+\Sigma\unj D\unj
    U)\fres+D^{2}\unj\Sigma\fres+\Sigma\unj
    D^{2}\fres+D^{3}\unj\Sigma\fres+D^{4}\unj\Sigma\fres\bigr],\\
    v_{j}^{(2)}&=-\Tr\bigl[2\Sigma\unj\Sigma U\fres+\Sigma\unj
    D\unj\Sigma
    U\fres+D_{2}\unj\Sigma^{2}\unj-2D_{1}(\Sigma\unj\Sigma+\unj\Sigma^{2}\unj)\bigr].
  \end{align*}
  Using $D=D_{1}+D_{2}$, we gather in $v_{j}^{(1)} + v_{j}^{(2)}$ the
  six terms quadratic in $\Sigma$ and linear in $D$ which combine
  (using trace cyclicity) as (again, ordering is carefully chosen!)
  \begin{align*}
    &-\Tr\bigl[2\Sigma\unj D\unj\Sigma\fres+2\Sigma\unj\Sigma
    D\fres+D_{2}\unj\Sigma^{2}\unj-2D_{1}(\Sigma\unj\Sigma+\unj\Sigma^{2}\unj)\bigr]\nonumber\\
    ={}&\Tr\bigl[2\Sigma\unj
    D\unj\Sigma(\Itens-\fres)+2\Sigma\unj\Sigma
    D(\Itens-\fres)-D_{2}\unj\Sigma^{2}\unj-2D_{2}(\Sigma\unj\Sigma+\unj\Sigma^{2}\unj)\bigr]\nonumber\\
    ={}&-\Tr\bigl[2\Sigma\unj D\unj\Sigma U\fres+2U\Sigma\unj\Sigma
    D\fres+D_{2}\unj\Sigma^{2}\unj+2D_{2}(\Sigma\unj\Sigma+\unj\Sigma^{2}\unj)\bigr].
  \end{align*}
  Then $v_{j}^{(1)}+v_{j}^{(2)}$ rewrites as
  \begin{multline*}
    v_{j}^{(1)}+v_{j}^{(2)}=-\Tr\bigl[2(D^{2}\unj\Sigma+\Sigma\unj
    D^{2})\fres+D^{3}\unj\Sigma\fres+2\Sigma\unj\Sigma^{2}\fres\\
    +(D^{4}\unj\Sigma+3\Sigma\unj D\unj\Sigma U+2U\Sigma\unj\Sigma
    D)\fres\\
    +3D_{2}\unj\Sigma^{2}\unj+2D_{2}\Sigma\unj\Sigma\bigr].
  \end{multline*}
The third line contains only convergent
loop vertices and is free of any resolvent. The terms on the second
line are ready for a HS bound (and a $L^{1}/L^{\infty}$ bound)
\ie they will lead to convergent loop vertices. But the
first line needs some further expansion of the resolvent factors
($\fres=\Itens+U\fres=\Itens+\fres U$):
\begin{align*}
  \Tr\bigl[D^{2}\unj\Sigma\fres\bigr]&=\Tr\bigl[D^{2}\unj\Sigma\bigr]+\Trsb{(\Sigma+D)D^{2}\unj\Sigma\fres}\\
  &=\Trsb{D^{2}\unj\Sigma}+\Trsb{(\Sigma\unj
    D)(D\unj\Sigma\fres)}+\Trsb{D^{3}\unj\Sigma\fres}\nonumber\\
  &=\Trsb{D^{2}\unj\Sigma}+\Trsb{(\Sigma\unj
    D)(D\unj\Sigma\fres)}+\Trsb{D^{3}\unj\Sigma}+\Trsb{(\Sigma+D)D^{3}\unj\Sigma\fres},\nonumber\\
  \nonumber\\
  \Trsb{\Sigma\unj D^{2}}&=\Trsb{D^{2}\unj\Sigma}+\Trsb{(\Sigma\unj
    D)(D\unj\Sigma\fres)}+\Trsb{D^{3}\unj\Sigma}\nonumber\\
  &\hspace{7cm}+\Trsb{\Sigma\unj
    D^{3}(\Sigma+D)\fres},\\
  \nonumber\\
  \Trsb{D^{3}\unj\Sigma\fres}&=\Trsb{D^{3}\unj\Sigma}+\Trsb{(\Sigma+D)D^{3}\unj\Sigma\fres},\\
  \nonumber\\
  \Trsb{\Sigma\unj\Sigma^{2}\fres}&=\Trsb{\Sigma^{3}\unj}+\Trsb{(\Sigma+D)\Sigma\unj\Sigma^{2}\fres}.
\end{align*}
We finally get an expression for $v_{j}$ which is completely ready for
our bounds:
\begin{equation}
  \begin{aligned}
    v_{j}=-\Tr\bigl[&2\Sigma^{2}\unj\Sigma^{2}\fres+2D\Sigma\unj\Sigma^{2}\fres+2\Sigma^{2}\unj\Sigma
    D\fres+3\Sigma D\unj\Sigma^{2}\fres\\
    &+4(\Sigma\unj D)(D\unj\Sigma\fres)+3(\Sigma\unj D)(\unj\Sigma
    D\fres)+2(D\Sigma\unj)(\unj\Sigma D\fres)+5\Sigma
    D^{3}\unj\Sigma\fres\\
    &+4D^{4}\unj\Sigma\fres+2\Sigma\unj D^{4}\fres\\
    &+2\Sigma^{3}\unj+3D_{2}\unj\Sigma^{2}\unj+2D_{2}\Sigma\unj\Sigma+4D^{2}\unj\Sigma\\
    &+D^{5}\unj\fres\bigr]+\cD_{conv,j}.
  \end{aligned}\label{eq-vjready}
\end{equation}
We need to bound
$\abs{V_{j}}=\abs{\int_{0}^{1}dt_{j}\,v_{j}(t_{j})}\les\int_{0}^{1}dt_{j}\,\abs{v_{j}}$. The
bound on $\abs{v_{j}}$ will be uniform in $t_{j}$ so that the integral
is simply bounded by $1$. Let us now show that (the module of) each
term of \cref{eq-vjready} is bounded by a sum of (modules of) allowed
loop vertices, see \cref{def-cvLoopVertices}.
\begin{align*}
  \abs{\Trsb{\Sigma^{2}\unj\Sigma^{2}\fres}}&\les\Trsb{\abs{\Sigma}^{4}\unj}+\Trsb{(\Sigma^{2})^{*}\unj\Sigma^{2}\fres\fres^{*}}&&\text{by
    \cref{eq-HSdef} and $\unj=\unj^{2}$}\\
  &\les 2\Trsb{\abs{\Sigma}^{4}\unj}=2\rho^{2}U_{j}^{4}&&\text{by \cref{eq-LunLinftyDef} and
    $\abs{g}\norm{\fres}\les\rho$.}\\
  \\
  \abs{\Trsb{D\Sigma\unj\Sigma^{2}\fres}}&\les\Trsb{D^{2}\Sigma^{*}\unj\Sigma}+\Trsb{(\Sigma^{2})^{*}\unj\Sigma^{2}\fres\fres^{*}}&&\text{by
    \cref{eq-HSdef} and $\unj=\unj^{2}$}\\
  &\les\abs g^{3}U_{j}^{2,b}+\rho^{2}U_{j}^{4}&&\text{by \cref{eq-LunLinftyDef} and
    $\abs{g}\norm{\fres}\les\rho$}
  \end{align*}
  and similarly for the two other terms of the first line
    of \cref{eq-vjready}.
    \begin{align*}
      \abs{\Trsb{(\Sigma\unj
          D)(D\unj\Sigma\fres)}}&\les\rho^{3}U_{j}^{2,a}+\Trsb{\Sigma^{*}\unj
        D^{2}\unj\Sigma\fres\fres^{*}}\les 2\rho^{3}U_{j}^{2,a}.\\
      \intertext{Similarly one gets}
      \abs{\Trsb{(\Sigma\unj D)(\unj\Sigma
          D\fres)}}&\les\rho^{3}(U_{j}^{2,a}+U_{j}^{2,b}),\\
      \abs{\Trsb{(D\Sigma\unj)(\unj\Sigma D\fres)}}&\les 2\rho^{3}U_{j}^{2,b}\\
      \abs{\Trsb{\Sigma D^{3}\unj\Sigma\fres}}&\les\rho^{3}U_{j}^{2,a}+\rho^{5}U_{j}^{2,c},\\
      \abs{\Trsb{D^{4}\unj\Sigma\fres}}&\les\rho^{3}U_{j}^{2,a}+\rho^{6}U_{j}^{0,a}\ges\abs{\Trsb{\Sigma\unj D^{4}\fres}}.
    \end{align*}
    The remaining terms of $v_{j}$, see \cref{eq-vjready}, already belong to the list of convergent loop
    vertices.

\subsection{Faces and loop vertices}
\label{sec-faces-loop-vertices}

We gather here the missing details of the proof of
\cref{thm-AGestimate}. 

\subsubsection{Quartic loop vertex}
\label{sec-quartic-loop-vertex}

We start by studying the incidence relations
between faces and a vertex of type $U_{j_{1}}^{4}=1/\abs
  g^{2}\Tr[\abs{\Sigma}^{4}\indic_{j_{1}}]$, see \cref{fig-U4} for a
graphical representation of it. Making explicit the
dependency of $U_{j_{1}}^{4}$ on the $\sigma$-field, we decompose our
analysis into four main cases corresponding to the number of different
colours around the vertex.

\paragraph{The $c^4$-case}

Here all four $\sigma$-fields bear the same colour, \cref{fig-U4c4}. A Wick-contraction
can create zero, one or two tadpoles. We gather data about
possible faces (colour, local or not, length, worst cost) in \cref{tab-c4}. For
example, in case of two ``planar'' tadpoles, there is one local face
of length $2$ and colour $c$, depicted in red in \cref{fig-c42plantad}. In case of
one ``non-planar'' tadpole, there are two non-local faces of length at
least $3$, in red in \cref{fig-c41nptad}. Also, in any case, there are three local
faces of respective colours $c'\neq c$ and length $4$, in red in \cref{fig-U4c4local}. All in all, the worst cost with tapole(s) is
$M^{j_{1}+j_{2}+4j_{4}}$ and $M^{(j_{1}+j_{2}+j_{3}+7j_{4})/2}$
without tadpole.
\setlength{\extrarowheight}{3pt}
\begin{table}[!htp]
  \begin{displaymath}
    \begin{array}{|c|c|c|c|c||c|}
      \hline
      \multicolumn{2}{|c|}{\thead{Tadpoles}}&\thead{Colour}&\thead{Locality}&\thead{Length}&\multicolumn{1}{c|}{\thead{Worst cost}}\\
      \hline
      \hline
      \multicolumn{2}{|c}{}&c'\neq c&\text{local $\times 3$}&4&M^{3j_{4}}\\
      \hline
      \multirow{3}{*}{2}&\multirow{2}{*}{\text{planar}}&\multirow{8}{*}{$c$}&\text{local
        $\times
      2$}&1&M^{j_{1}+j_{2}}\\
      &&&\text{local}&2&M^{j_{4}}\\
      \cline{2-2}\cline{4-6}
      &\text{non-planar}&&\text{local}&4&M^{j_{4}}\\
      \cline{1-2}\cline{4-6}
      \multirow{4}{*}{1}&\multirow{3}{*}{planar}&&\text{local}&1&M^{j_{1}}\\
      &&&\text{non-local}&>1&M^{j_{2}/2}\\
      &&&\text{non-local}&>2&M^{j_{4}/2}\\
      \cline{2-2}\cline{4-6}
      &\text{non-planar}&&\text{non-local $\times
      2$}&>2&M^{(j_{2}+j_{4})/2}\\
      \cline{1-2}\cline{4-6}
      \multicolumn{2}{|c|}{0}&&\text{non-local $\times
      4$}&>1&M^{(j_{1}+j_{2}+j_{3}+j_{4})/2}\\
      \hline
    \end{array}
  \end{displaymath}
  \caption{The $c^{4}$-case}
  \label{tab-c4}
\end{table}
\begin{figure}
  \centering
  \begin{subfigure}[c]{.5\linewidth}
    \centering
    \includegraphics[scale=.8]{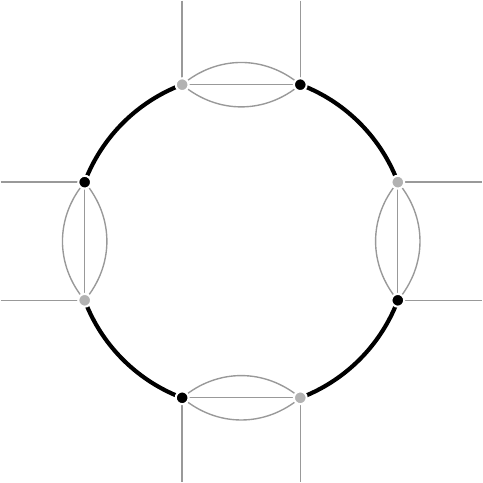}
    \caption{The $U^{4}$-vertex}
    \label{fig-U4}
  \end{subfigure}\\%
  \begin{subfigure}[c]{.5\linewidth}
    \centering
    \includegraphics[scale=.8]{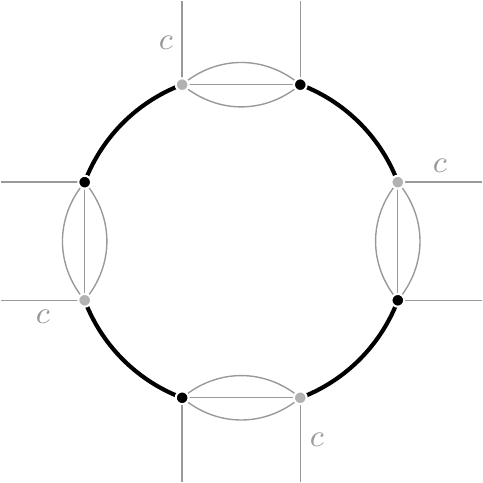}
    \caption{The $c^{4}$-case}
    \label{fig-U4c4}
  \end{subfigure}
  \begin{subfigure}[c]{.5\linewidth}
    \centering
    \includegraphics[scale=.8]{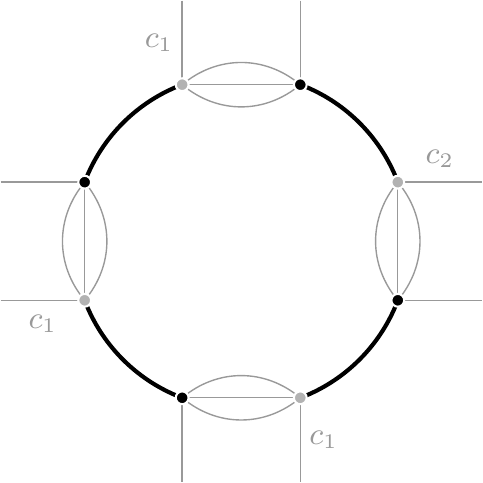}
    \caption{The $c_{1}^{3}c_{2}$-case}
    \label{fig-c3c}
  \end{subfigure}\\%
  \bigskip
  \begin{subfigure}[c]{.5\linewidth}
    \centering
    \includegraphics[scale=.8]{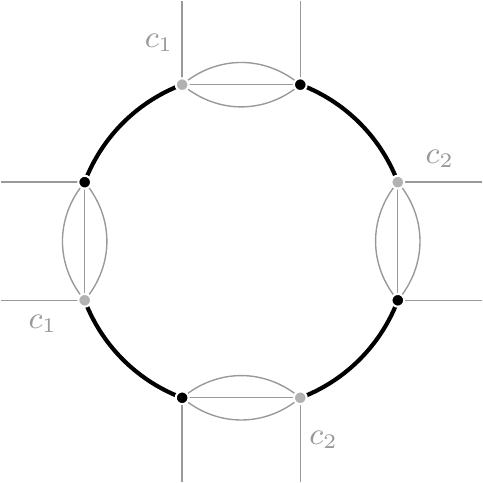}
    \caption{The contiguous $c_{1}^{2}c_{2}^{2}$-case}
    \label{fig-c2c2a}
  \end{subfigure}
  \begin{subfigure}[c]{.5\linewidth}
    \centering
    \includegraphics[scale=.8]{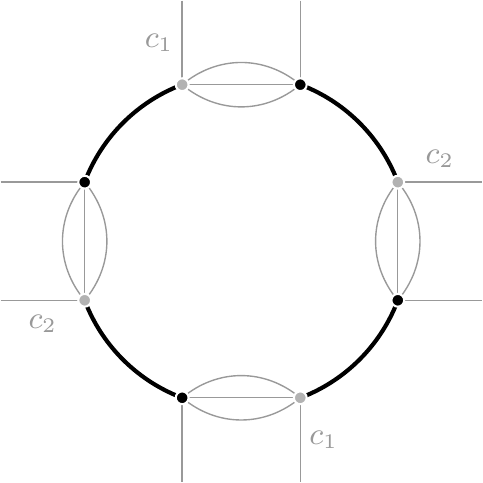}
    \caption{The alternating $c_{1}^{2}c_{2}^{2}$-case}
    \label{fig-c2c2b}
  \end{subfigure}
  \caption{The $U^{4}$-loop vertex and some of its coloured versions.}
  \label{fig-U4etal}
\end{figure}
\begin{figure}
  \centering
  \begin{subfigure}[c]{.5\linewidth}
    \centering
    \includegraphics[scale=.8]{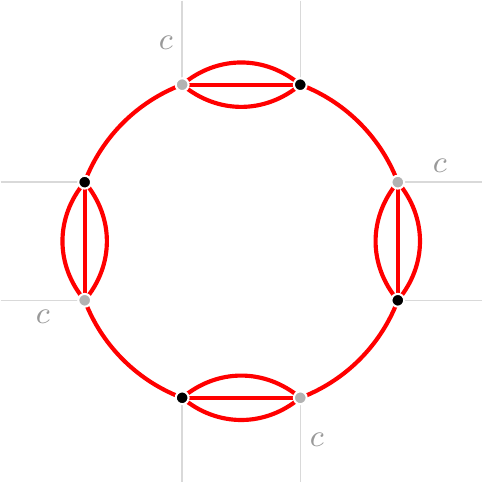}
    \caption{The $3$ local faces of colours $c'\neq c$}
    \label{fig-U4c4local}
  \end{subfigure}\\
    \bigskip
  \begin{subfigure}[c]{.5\linewidth}
    \centering
    \includegraphics[scale=.8]{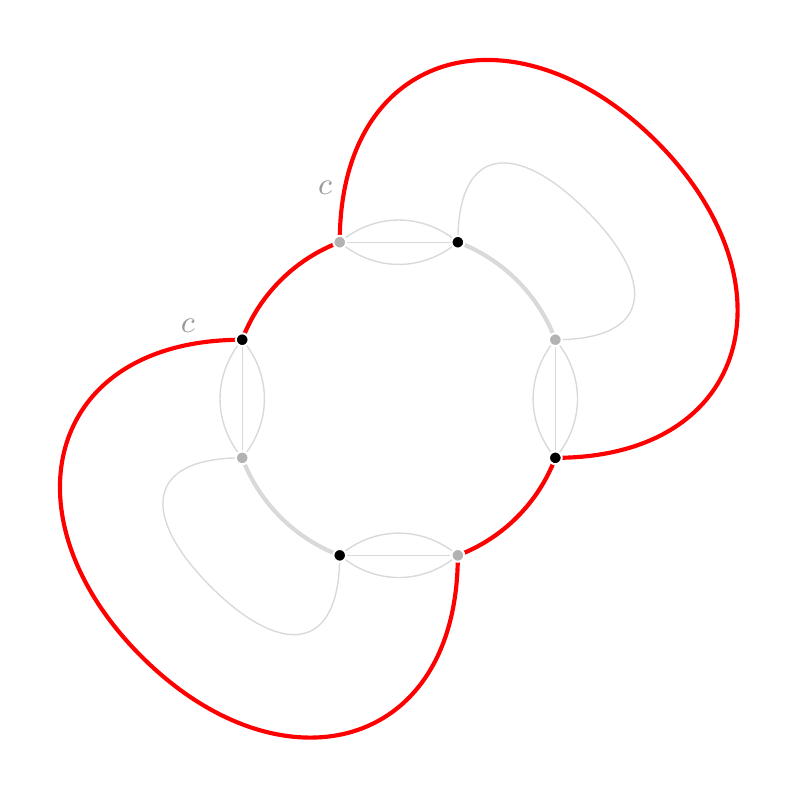}
    \caption{Two planar tadpoles}
    \label{fig-c42plantad}
  \end{subfigure}%
  \begin{subfigure}[c]{.5\linewidth}
    \centering
    \includegraphics[scale=.8]{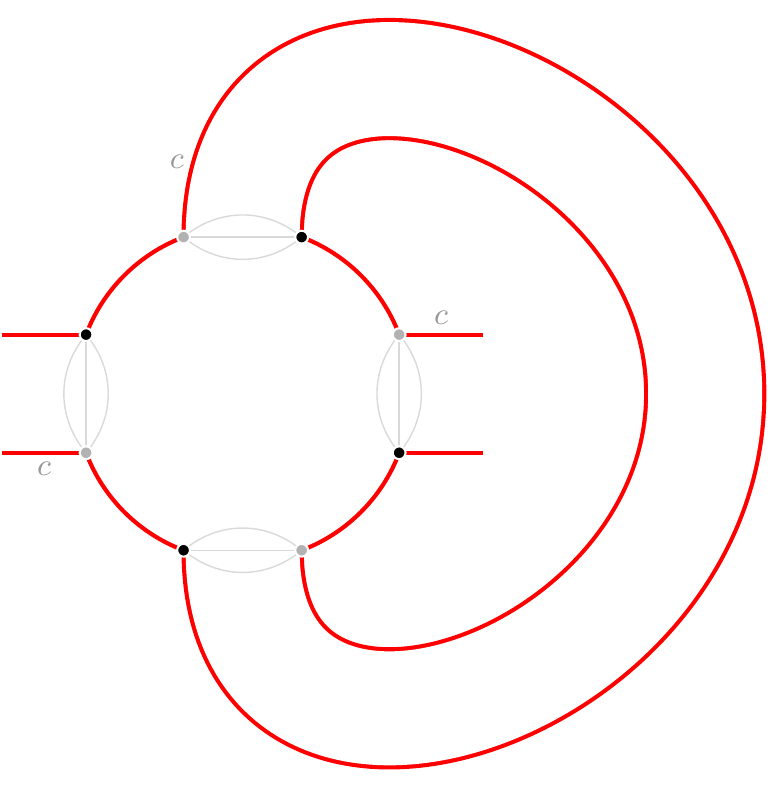}
    \caption{One non-planar tadpole}
    \label{fig-c41nptad}
  \end{subfigure}
  \caption{The $c^{4}$-case and some of its possible faces.}
  \label{fig-c4}
\end{figure}

\paragraph{The $c_1^3c_2$-case}

Here there can be zero or one tapole which could be planar or not, see
\cref{fig-c3c}. The worst cost with tadpole is $M^{j_{1}+3j_{4}}$ and
$M^{(j_{1}+j_{2}+6j_{4})/2}$ without, see \cref{tab-c3c}.
\begin{table}[!htp]
  \begin{displaymath}
    \begin{array}{|c|c|c|c|c||c|}
      \hline
      \multicolumn{2}{|c|}{\thead{Tadpoles}}&\thead{Colour}&\thead{Locality}&\thead{Length}&\multicolumn{1}{c|}{\thead{Worst cost}}\\
      \hline
      \hline
      \multicolumn{2}{|c}{}&c\neq c_{1},c_{2}&\text{local $\times
      2$}&4&M^{2j_{4}}\\
      \multicolumn{2}{|c}{}&c_{2}&\text{non-local}&>4&M^{j_{4}/2}\\
      \hline
      \multirow{3}{*}{1}&\multirow{2}{*}{\text{planar}}&\multirow{5}{*}{$c_{1}$}&\text{local}&1&M^{j_{1}}\\
      &&&\text{non-local}&>3&M^{j_{4}/2}\\
      \cline{2-2}\cline{4-6}
      &\text{non-planar}&&\text{non-local $\times 2$}&>2&M^{(j_{2}+j_{4})/2}\\
      \cline{1-2}\cline{4-6}
      \multicolumn{2}{|c|}{\multirow{2}{*}{$0$}}&&\text{non-local $\times 2$}&>1&M^{(j_{1}+j_{2})/2}\\
      \multicolumn{2}{|c|}{}&&\text{non-local}&>2&M^{j_{4}/2}\\
      \hline
    \end{array}
  \end{displaymath}
  \caption{The $c_{1}^{3}c_{2}$-case}
  \label{tab-c3c}
\end{table}

\paragraph{The $c_1^2c_2^2$-cases}

We have two main cases here: either the $\sigma$-fields of colour
$c_{1}$ are contiguous or not, see \cref{fig-c2c2a,fig-c2c2b}. Anyway, there can be again zero, one or
two tadpoles after Wick-contraction. All in all, the worst cost with
tadpoles is $M^{j_{1}+j_{2}+4j_{4}}$ and $M^{(j_{1}+j_{2}+6j_{4})/2}$
without, see \cref{tab-c2c2}.
\begin{table}[!htp]
  \centering
  \begin{subtable}[c]{\linewidth}
    \centering
    \begin{tabular}{|c|c|c|c|c||c|}
      \hline
      \multicolumn{2}{|c|}{\theadt{Tadpoles}}&\theadt{Colour}&\theadt{Locality}&\theadt{Length}&\multicolumn{1}{c|}{\theadt{Worst cost}}\\
      \hline
      \hline
      \multicolumn{2}{|c}{}&$c\neq c_{1},c_{2}$&local $\times
      2$&4&$M^{2j_{4}}$\\
      \hline
      \multicolumn{2}{|c|}{\multirow{4}{*}{$2$}}&\multirow{2}{*}{$c_{1}$}&local&1&$M^{j_{1}}$\\
      \multicolumn{2}{|c|}{}&&local&3&$M^{j_{4}}$\\
      \multicolumn{2}{|c|}{}&\multirow{2}{*}{$c_{2}$}&local&1&$M^{j_{2}}$\\
      \multicolumn{2}{|c|}{}&&local&3&$M^{j_{4}}$\\
      \hline
      \multicolumn{2}{|c|}{\multirow{4}{*}{$1$}}&\multirow{2}{*}{$c_{1}$}&local&1&$M^{j_{1}}$\\
      \multicolumn{2}{|c|}{}&&local&$3$&$M^{j_{4}}$\\
      \multicolumn{2}{|c|}{}&\multirow{2}{*}{$c_{2}$}&non-local&$>1$&$M^{j_{2}/2}$\\
      \multicolumn{2}{|c|}{}&&non-local&$>3$&$M^{j_{4}/2}$\\
      \hline
      \multicolumn{2}{|c|}{\multirow{4}{*}{$0$}}&\multirow{2}{*}{$c_{1}$}&non-local&$>1$&$M^{j_{1}/2}$\\
      \multicolumn{2}{|c|}{}&&non-local&$>3$&$M^{j_{4}/2}$\\
      \multicolumn{2}{|c|}{}&\multirow{2}{*}{$c_{2}$}&non-local&$>1$&$M^{j_{2}/2}$\\
      \multicolumn{2}{|c|}{}&&non-local&$>3$&$M^{j_{4}/2}$\\
      \hline
    \end{tabular}
    \caption{Contiguous}
    \label{tab-c2c2a}
  \end{subtable}\\
  \bigskip   \bigskip
  \begin{subtable}[c]{\linewidth}
    \centering
    \begin{tabular}{|c|c|c|c|c||c|}
      \hline
      \multicolumn{2}{|c|}{\theadt{Tadpoles}}&\theadt{Colour}&\theadt{Locality}&\theadt{Length}&\multicolumn{1}{c|}{\theadt{Worst cost}}\\
      \hline \hline \multicolumn{2}{|c}{}&$c\neq c_{1},c_{2}$&local
      $\times
      2$&4&$M^{2j_{4}}$\\
      \hline
      \multicolumn{2}{|c|}{\multirow{2}{*}{$2$}}&$c_{1}$&local
      $\times 2$&2&$M^{j_{2}+j_{4}}$\\
      \multicolumn{2}{|c|}{}&$c_{2}$&local $\times 2$&2&$M^{j_{3}+j_{4}}$\\
      \hline
      \multicolumn{2}{|c|}{\multirow{2}{*}{$1$}}&$c_{1}$&local
      $\times 2$&2&$M^{j_{2}+j_{4}}$\\
      \multicolumn{2}{|c|}{}&$c_{2}$&non-local $\times 2$&$>2$&$M^{(j_{3}+j_{4})/2}$\\
      \hline
      \multicolumn{2}{|c|}{\multirow{2}{*}{$0$}}&$c_{1}$&non-local $\times 2$&$>2$&$M^{(j_{2}+j_{4})/2}$\\
      \multicolumn{2}{|c|}{}&$c_{2}$&non-local $\times 2$&$>2$&$M^{(j_{3}+j_{4})/2}$\\
      \hline
    \end{tabular}
    \caption{Alternating}
    \label{tab-c2c2b}
  \end{subtable}
  \caption{The $c_{1}^{2}c_{2}^{2}$-cases}
  \label{tab-c2c2}
\end{table}

\clearpage
\paragraph{The $c_1^2c_2c_3$-cases}

Here again the $\sigma$-fields of colour $c_{1}$ are either contiguous
or not. There can be zero or one tadpole. The worst cost with tadpole
is $M^{j_{1}+3j_{4}}$ and $M^{(j_{1}+5j_{4})/2}$ without, see \cref{tab-c2cc}.
\begin{table}[!htp]
  \centering
  \begin{subtable}[c]{\linewidth}
    \centering
    \begin{tabular}{|c|c|c|c|c||c|}
      \hline
      \multicolumn{2}{|c|}{\theadt{Tadpoles}}&\theadt{Colour}&\theadt{Locality}&\theadt{Length}&\multicolumn{1}{c|}{\theadt{Worst cost}}\\
        \hline
        \hline
        \multicolumn{2}{|c}{}&$c\neq
        c_{1},c_{2},c_{3}$&local&4&$M^{j_{4}}$\\
        \multicolumn{2}{|c}{}&$c_{2}$&non-local&$>4$&$M^{j_{4}/2}$\\
        \multicolumn{2}{|c}{}&$c_{3}$&non-local&$>4$&$M^{j_{4}/2}$\\
        \hline
        \multicolumn{2}{|c|}{\multirow{2}{*}{$1$}}&\multirow{4}{*}{$c_{1}$}&local&1&$M^{j_{1}}$\\
        \multicolumn{2}{|c|}{}&&local&$3$&$M^{j_{4}}$\\
        \multicolumn{2}{|c|}{\multirow{2}{*}{$0$}}&&non-local&$>1$&$M^{j_{1}/2}$\\
        \multicolumn{2}{|c|}{}&&non-local&$>3$&$M^{j_{4}/2}$\\
        \hline
     \end{tabular}
     \caption{Contiguous}
     \label{tab-c2cca}
  \end{subtable}\\
  \bigskip \bigskip
  \begin{subtable}[c]{\linewidth}
    \centering
    \begin{tabular}{|c|c|c|c|c||c|}
      \hline
      \multicolumn{2}{|c|}{\theadt{Tadpoles}}&\theadt{Colour}&\theadt{Locality}&\theadt{Length}&\multicolumn{1}{c|}{\theadt{Worst cost}}\\
        \hline
        \hline
        \multicolumn{2}{|c}{}&$c\neq
        c_{1},c_{2},c_{3}$&local&4&$M^{j_{4}}$\\
        \multicolumn{2}{|c}{}&$c_{2}$&non-local&$>4$&$M^{j_{4}/2}$\\
        \multicolumn{2}{|c}{}&$c_{3}$&non-local&$>4$&$M^{j_{4}/2}$\\
        \hline
        \multicolumn{2}{|c|}{$1$}&\multirow{2}{*}{$c_{1}$}&local
        $\times 2$&2&$M^{j_{2}+j_{4}}$\\
        \multicolumn{2}{|c|}{$0$}&&non-local
        $\times 2$&$>2$&$M^{(j_{2}+j_{4})/2}$\\
        \hline
      \end{tabular}
      \caption{Alternating}
      \label{tab-c2ccb}
  \end{subtable}
  \caption{The $c_{1}^{2}c_{2}c_{3}$-cases}
  \label{tab-c2cc}
\end{table}

\paragraph{The $c_1c_2c_3c_4$-case}

All $\sigma$-fields bear different colours. No tadpole is
possible. The cost is $M^{2j_{4}}$.
\begin{table}[!htp]
  \begin{displaymath}
    \begin{array}{|c|c|c|c|c||c|}
      \hline
      \multicolumn{2}{|c|}{\thead{Tadpoles}}&\thead{Colour}&\thead{Locality}&\thead{Length}&\multicolumn{1}{c|}{\thead{Worst cost}}\\
      \hline
      \hline
      \multicolumn{2}{|c|}{$0$}&c_{1},c_{2},c_{3},c_{4}&\text{non-local $\times
      4$}&>4&M^{2j_{4}}\\
      \hline
    \end{array}
  \end{displaymath}
  \caption{The $c_{1}c_{2}c_{3}c_{4}$-case}
  \label{tab-cccc}
\end{table}

\subsubsection{Quadratic loop vertices}
\label{sec-quadr-loop-vert}

There are five different types of quadratic vertices
$U_{j_{1}}^{2,\alpha}$, see \cref{def-cvLoopVertices}. They involve
$D$ and $D_{2}$ operators. Recall that $D=D_{1}+D_{2}$ where
$D_{1}=C^{1/2}\Ar_{\cM_{1}}C^{1/2}$ and
$D_{2}=C^{1/2}\Ar_{\cM_{2}}C^{1/2}$. Both are diagonal operators in
the momentum basis. From \cref{thm-Aren}, we get
\begin{equation*}
  \sup\set{\abs{(D_{1})_{\mtup,\ntup}},\abs{(D)_{\mtup,\ntup}}}\les\delta_{\mtup\ntup}\frac{\Oun}{\norm{\mtup}+1},\qquad \abs{(D_{2})_{\mtup,\ntup}}\les\delta_{\mtup\ntup}\frac{\Oun}{\norm{m}^{2-\epsilon}+1}.
\end{equation*}
Each quadratic vertex contains two $C$-propagators plus a $D$-type
operator. At worst (depending on the position of the $\indic_{j_{1}}$
cutoff in the trace), the $D^{2}$ operator brings $M^{-2j_{2}}$,
$D^{4}$ brings $M^{-4j_{2}}$ and $D_{2}$ brings
$M^{-(2-\epsilon)j_{2}}$. Thus the worst vertex is
$U_{j_{1}}^{2,e}$. Note that because of the
conservation of the $4$-tuple of indices through the
$D$-type insertion, the scales on both of its sides are equal (to
$j_{2}$). The face data and worst costs for a generic quadratic loop
vertex are available in \cref{tab-c2,tab-cc}. The worst cost with tadpole is then $M^{j_{1}+4j_{2}}$ and
$M^{(j_{1}+7j_{2})/2}$ without.
\begin{table}[!htp]
  \centering
      \begin{tabular}{|c|c|c|c|c||c|}
      \hline
      \multicolumn{2}{|c|}{\theadt{Tadpoles}}&\theadt{Colour}&\theadt{Locality}&\theadt{Length}&\multicolumn{1}{c|}{\theadt{Worst cost}}\\
        \hline
        \hline
        \multicolumn{2}{|c}{}&$c'\neq c$&local $\times 3$&3&$M^{3j_{2}}$\\
        \hline
        \multicolumn{2}{|c|}{\multirow{2}{*}{$1$}}&\multirow{4}{*}{$c$}&local&1&$M^{j_{1}}$\\
        \multicolumn{2}{|c|}{}&&local&2&$M^{j_{2}}$\\
        \multicolumn{2}{|c|}{\multirow{2}{*}{$0$}}&&non-local&$>1$&$M^{j_{1}/2}$\\
        \multicolumn{2}{|c|}{}&&non-local&$>2$&$M^{j_{2}/2}$\\
        \hline
      \end{tabular}
  \caption{The $c^2$-case}
  \label{tab-c2}
\end{table}
\begin{table}[!htp]
  \centering
      \begin{tabular}{|c|c|c|c|c||c|}
      \hline
      \multicolumn{2}{|c|}{\theadt{Tadpoles}}&\theadt{Colour}&\theadt{Locality}&\theadt{Length}&\multicolumn{1}{c|}{\theadt{Worst cost}}\\
        \hline
        \hline
        \multicolumn{2}{|c|}{\multirow{2}{*}{$0$}}&$c'\neq c_{1},c_{2}$&local
        $\times 2$&3&$M^{2j_{2}}$\\
        \multicolumn{2}{|c|}{}&$c_{1},c_{2}$&non-local
        $\times 2$&$>3$&$M^{j_{2}}$\\
        \hline
      \end{tabular}
  \caption{The $c_{1}c_{2}$-case}
  \label{tab-cc}
\end{table}

\subsubsection{Loop vertices of degree one}
\label{sec-loop-vertices-degree-one}

The most dangerous case is $U_{j}^{1,a}$, the $D^{2}$-insertion of
which brings $M^{-2j}$. The cost is $M^{7j/2}$, see \cref{tab-degone}.
\begin{table}[!htp]
  \centering
      \begin{tabular}{|c|c|c|c|c||c|}
      \hline
      \multicolumn{2}{|c|}{\theadt{Tadpoles}}&\theadt{Colour}&\theadt{Locality}&\theadt{Length}&\multicolumn{1}{c|}{\theadt{Worst cost}}\\
        \hline
        \hline
        \multicolumn{2}{|c|}{\multirow{2}{*}{$0$}}&$c'\neq c$&local
        $\times 3$&2&$M^{3j}$\\
        \multicolumn{2}{|c|}{}&$c$&non-local&$>2$&$M^{j/2}$\\
        \hline
      \end{tabular}
  \caption{The degree one case}
  \label{tab-degone}
\end{table}

\subsection{Perturbative bounds}
\label{sec-perturbative-bounds}

Multiscale analysis  is a powerful tool to bound the Feynman amplitudes of convergent graphs both in the standard \cite{Feldman1985xv}
and in the tensor case \cite{Ben-Geloun2011aa}. It is especially easy in the
superrenormalizable case, as it not only proves uniform bounds on any (renormalized)
Feynman amplitude, but it does this and in addition allows to spare a uniform small fraction $\eta>0$ of the scale factor of every line of the graph,
which can then be used for other purposes.
 
Let us state precisely a \namecref{thm-PowCountSpare} of this type for our model, for which we can take $\eta = \frac{1}{24}$. More precisely

\begin{lemma}\label{thm-PowCountSpare}
 Let $\mu = \{ i_\ell \} $ be a scale attribution for the lines of a
 vacuum Feynman graph of the $T^4_4$ theory. There exists $K >0$ such that
\begin{equation}
  \sum_\mu \Ar_{G, \mu} \,\bigl[ \prod_{\ell \in G} M^{ \frac{i_\ell}{24}} \bigr] \les K^{n(G)}  \label{lastlemmaeq}
\end{equation}
where $n(G)$ is the order of $G$, and
$\Ar_{G, \mu}$ is the renormalised Feynman amplitude for scale attribution $\mu$.
\end{lemma}
\begin{proof}
  Let us denote $f\, \bowtie\, \ell$ if face $f$ runs through line
  $\ell$. When $G$ is convergent, we return to the direct representation and consider
  \cref{eq-pertbou}.  We have to add the factor
  $ \prod_\ell M^{\frac{i_\ell}{24}} $ to the previous factor
  $ \prod_\ell M^{-2i_\ell} $. Hence we have to prove
  \begin{equation*}
    \sum_\mu \prod_{\ell} M^{-\frac{47}{24}i (\ell)} \prod_f M^{i_m (f)} \les K^{n(G)}
  \end{equation*}
  where we recall that $i_m (f)= \inf_{f\, \bowtie\, \ell} i_\ell$ is
  the smallest scale of the edges along face $f$.  We obtain an upper
  bound by replacing the factor $M^{i_m (f)} $ for each face by
  $M^{\sum_{\ell \,\bowtie\, f} \frac{i_\ell}{L(f)}}$ where $L(f)$ is
  the length of the face $f$.  Reordering the product we simply have
  now to check that
  \begin{equation*}
    \sum_{f\, \bowtie\, \ell} \frac{1}{L(f)} \les\frac{23}{12} \qquad \forall \ell
\end{equation*}
since
  $- \frac{47}{24} + \frac{23}{12} = -\frac{1}{24} $ and obviously
  \begin{equation*}
\sum_\mu \prod_{\ell} M^{- \frac{i (\ell)}{24}} \les K^{n(G)}.
\end{equation*}
\begin{itemize}
\item If no face of length 1 or 2 runs through $\ell$ obviously
  $\sum_{f\, \bowtie\, \ell} \frac{1}{L(f)} \les \frac{4}{3} <
  \frac{23}{12}$.
\item If a face of length 2 runs through $\ell$, but none of length
  1, it is easy to check that there can be at most three such faces
  (not four, otherwise the graph would be ${\kN}_2$ in \cref{f-VacuumNonMelonicDivergences}). Hence
  $\sum_{f\, \bowtie\, \ell} \frac{1}{L(f)} \les \frac{3}{2} +
  \frac{1}{3} < \frac{23}{12}$.
\item Finally if a face of length 1 runs through $\ell$ it must be a
  tadpole. It cannot be a melonic tadpole of a $\cM_1$ type (otherwise
  it would not be a convergent graph).  It can be of the non melonic
  type, but it cannot be the non-melonic tadpole of ${\kN}_1$ or
  ${\frak N}_3$ in \cref{f-VacuumNonMelonicDivergences}
  because these graphs diverge.  Hence the other faces through $f$
  cannot be of length 2 or all of length 3. Hence
  $\sum_{f\, \bowtie\, \ell} \frac{1}{L(f)} \les 1 + \frac{2}{3} +
  \frac{1}{4} = \frac{23}{12}$.
\end{itemize}

When $G$ contains a divergent subgraph of type ${\cal M}_1$ or
${\cal M}_2$ in \cref{f-masdivergences}, renormalization
bringing an additional factor $M^{-i_\ell}$ for the critical melonic
tadpoles lines, \cref{lastlemmaeq} still holds.

Finally when $G$ is one of the ten divergent vacuum graphs, in 
\cref{f-VacuumMelonicDivergences,f-VacuumNonMelonicDivergences},
its renormalized amplitude is 0 and there is nothing to prove.
\end{proof}

\newpage
\section*{Index of notation}
\label{sec-index-notations}
\etoctoccontentsline*{section}{Index of notation}{1}

\setglossarysection{subparagraph}
\printglossary[type=graphs]
\smallskip
\printglossary[type=space]
\smallskip
\printglossary[type=tensor]
\smallskip
\printglossary[type=operatortens]
\smallskip
\printglossary[type=operatordirect]
\smallskip
\printglossary[type=constant]
\smallskip
\printglossary[type=misc]

\newpage
\printbibliography

\contactrule
\contactVRivasseau
\contactFVT
\end{document}
